\documentclass[12pt,final,onecolumn]{IEEEtran}
\usepackage{graphicx}
\usepackage{amsmath}
\usepackage{mathrsfs}
\usepackage{mathtools}
\usepackage{amssymb}
\usepackage{float}
\usepackage{url}
\usepackage{verbatim}
\usepackage{dsfont}
\usepackage{enumerate}
\usepackage{layout}

\usepackage[bottom=0.8in]{geometry}

\newtheorem{theorem}{Theorem}
\newtheorem{lemma}{Lemma}
\newtheorem{proposition}{Proposition}
\newtheorem{corollary}{Corollary}

\newtheorem{definition}{Definition}
\newtheorem{example}{Example}
\newtheorem{remark}{Remark}
\newcommand{\prob}{\ensuremath{\mathbb{P}}}
\newcommand{\naturals}{\ensuremath{\mathbb{N}}}
\newcommand{\Reals}{\ensuremath{\mathbb{R}}}
\newcommand{\expectation}{\ensuremath{\mathbb{E}}}

\newcommand{\set}{\ensuremath{\mathcal}}

\begin{document}
\title{\LARGE{On Data-Processing and Majorization Inequalities\\[0.1cm]
for $f$-Divergences with Applications\\[1cm]}}
\author{Igal Sason
\thanks{
\normalsize{\vspace*{0.2cm}

{\small
I. Sason is  with the Department of Electrical Engineering,
Technion--Israel Institute of Technology, Haifa 32000, Israel
(e-mail: sason@ee.technion.ac.il).\newline
This paper is published in the {\em Entropy} journal, vol.~21,
no.~10, paper~1022, pages~1--80, October 21, 2019.
Available at \tt{https://www.mdpi.com/1099-4300/21/10/1022}.}}}}

\maketitle

\begin{abstract}
\normalsize{This paper is focused on derivations of data-processing and majorization
inequalities for $f$-divergences, and their applications in information theory
and statistics. For the accessibility of the material, the main results are first
introduced without proofs, followed by exemplifications of the theorems with further
related analytical results, interpretations, and information-theoretic applications.
One application refers to the performance analysis of list decoding with either fixed
or variable list sizes; some earlier bounds on the list decoding error probability are
reproduced in a unified way, and new bounds are obtained and exemplified numerically.
Another application is related to a study of the quality of approximating a
probability mass function, induced by the leaves of a Tunstall tree, by an equiprobable
distribution. The compression rates of finite-length Tunstall codes are further analyzed
for asserting their closeness to the Shannon entropy of a memoryless and stationary
discrete source. Almost all the analysis is relegated to the appendices, which form the
major part of this manuscript.}
\end{abstract}


\break
\section{Introduction}
\label{section: introduction}

Divergences are non-negative measures of the dissimilarity between arbitrary pairs
of probability measures which are defined on the same measurable space. They play
a key role in the development of information theory, probability theory,
statistics, learning, signal processing, and other related fields. One important
class of divergence measures is defined by means of convex functions $f$, and it
is called the class of $f$-divergences. It unifies fundamental and
independently-introduced concepts in several branches of mathematics such as the
chi-squared test for the goodness of fit in statistics, the total variation distance
in functional analysis, the relative entropy in information theory and statistics,
and it is also closely related to the R\'{e}nyi divergence which generalizes the
relative entropy. The class of $f$-divergences was independently introduced in the
sixties by Ali and Silvey \cite{AliS}, Csisz\'{a}r \cite{Csiszar63}--\cite{Csiszar72},
and Morimoto \cite{Morimoto63}. This class satisfies pleasing features such as the
data-processing inequality, convexity, (semi)continuity and duality properties, and it
finds nice applications in information theory and statistics (see, e.g.,
\cite{Csiszar67a, Csiszar72, LieseV_book87, Merhav11, Pardo05, PardoV97, StummerV10,
Vajda89, ZakaiZ75, ZivZ73}).

This manuscript is a research paper which is focused on the derivation of data-processing
and majorization inequalities for $f$-divergences, and a study of some of their potential
applications in information theory and statistics. Preliminaries are next provided.

\subsection{Preliminaries and Related Works}
\label{subsection: preliminaries}

We provide here definitions and known results from the literature which serve as a
background to the presentation in this paper. We first provide a definition for the
family of $f$-divergences.
\begin{definition} \label{def:fD} \cite[p.~4398]{LieseV_IT2006}
Let $P$ and $Q$ be probability measures, let $\mu$ be a dominating measure of $P$ and $Q$
(i.e., $P, Q \ll \mu$),
and let $p := \frac{\text{d}P}{\text{d}\mu}$ and $q := \frac{\text{d}Q}{\text{d}\mu}$.
The $f$-divergence from $P$ to $Q$ is given, independently of $\mu$, by
\begin{align} \label{eq:fD}
D_f(P\|Q) := \int q \, f \Bigl(\frac{p}{q}\Bigr) \, \text{d}\mu,
\end{align}
where
\begin{align}
& f(0) := \underset{t \to 0^+}{\lim} \, f(t), \\[0.1cm]
& 0 f\biggl(\frac{0}{0}\biggr) := 0, \\[0.1cm]
& 0 f\biggl(\frac{a}{0}\biggr)
:= \lim_{t \to 0^+} \, t f\biggl(\frac{a}{t}\biggr)
= a \lim_{u \to \infty} \frac{f(u)}{u}, \quad a>0.
\end{align}
\end{definition}

\begin{definition}
\label{def: contraction}
Let $Q_X$ be a probability distribution which is defined on a set $\set{X}$,
and that is not a point mass, and let $W_{Y|X} \colon \set{X} \to \set{Y}$
be a stochastic transformation.
The contraction coefficient for $f$-divergences is defined as
\begin{align}
\label{contraction coef.}
\mu_f(Q_X, W_{Y|X}) := \underset{P_X: \, D_f(P_X \| Q_X) \in (0, \infty)}{\sup} \,
\frac{D_f(P_Y \| Q_Y)}{D_f(P_X \| Q_X)},
\end{align}
where, for all $y \in \set{Y}$,
\begin{align}
\label{23062019a1}
& P_Y(y) = (P_X W_{Y|X}) \: (y) := \int_{\set{X}} \mathrm{d}P_X(x) \, W_{Y|X}(y|x), \\[0.1cm]
\label{23062019a2}
& Q_Y(y) = (Q_X W_{Y|X}) \: (y) := \int_{\set{X}} \mathrm{d}Q_X(x) \, W_{Y|X}(y|x).
\end{align}
The notation in \eqref{23062019a1} and \eqref{23062019a2}, and also in \eqref{transP}, \eqref{transQ},
\eqref{MC1 in DMC}, \eqref{MC2 in DMC}, \eqref{MC3 in DMC} in the continuation of this paper, is
consistent with the standard notation used in information theory (see, e.g., the first displayed
equation after (3.2) in \cite{Csiszar_Korner}).
\end{definition}

Contraction coefficients for $f$-divergences play a key role in
strong data-processing inequalities (see
\cite{AhlswedeG76, CalmonPW18, Cohen93},
\cite[Chapter~II]{CohenKZ98}, \cite{MakurP18, MakurZ18,
PolyanskiyW16, PolyanskiyW17, Raginsky16}).
The following are essential definitions and results which are
related to maximal correlation and strong data-processing inequalities.

\begin{definition}
The maximal correlation between two random variables $X$ and $Y$
is defined as
\begin{align}
\label{def: max. correlation}
\rho_{\mathrm{m}}(X;Y) := \sup_{f,g} \, \expectation[f(X) g(Y)],
\end{align}
where the supremum is taken over all real-valued functions $f$ and $g$ such that
\begin{align}
\expectation[f(X)] = \expectation[g(Y)] = 0,
\quad \expectation[f^2(X)] \leq 1, \; \expectation[g^2(Y)] \leq 1.
\end{align}
\end{definition}

\begin{definition} \label{def:chi-squared}
Pearson's $\chi^2$-divergence \cite{Pearson1900x} from $P$ to $Q$ is
defined to be the $f$-divergence from $P$ to $Q$ (see Definition~\ref{def:fD})
with $f(t) = (t-1)^2$ or $f(t) = t^2-1$ for all $t>0$,
\begin{align}
\label{eq: chi-square 1}
\chi^2(P\|Q) &:= D_f(P\|Q) \\
\label{eq: chi-square 1b}
&= \int \frac{(p-q)^2}{q} \, \text{d}\mu \\
\label{eq: chi-square 2}
&= \int \frac{p^2}{q} \, \text{d}\mu - 1
\end{align}
independently of the dominating measure $\mu$ (i.e.,
$P,Q \ll \mu$, e.g., $\mu = P+Q$).

Neyman's $\chi^2$-divergence from $P$ to $Q$ is the
Pearson's $\chi^2$-divergence from $Q$ to $P$, i.e., it
is equal to \cite{Neyman49}
\begin{align}
\label{eq: chi-square 3}
\chi^2(Q\|P) = D_g(P\|Q)
\end{align}
with $g(t) = \frac{(t-1)^2}{t}$ or $g(t) = \frac{1}{t}-t$ for all $t>0$.
\end{definition}

\begin{proposition} (\cite[Theorem 3.2]{Raginsky16}, \cite{Sarmanov62})
\label{propos: maximal correlation}
The contraction coefficient for the $\chi^2$-divergence satisfies
\begin{align}
\mu_{\chi^2}(Q_X, W_{Y|X}) = \rho_{\mathrm{m}}^2(X;Y),
\end{align}
with $X \sim Q_X$ and $Y \sim Q_Y$ (see \eqref{23062019a2}{\em)}.
\end{proposition}

\begin{proposition} \cite[Theorem~2]{PolyanskiyW17}
\label{propos: mu chi^2 is minimal}
Let $f \colon (0, \infty) \to \Reals$ be convex and twice continuously
differentiable with $f(1)=0$ and $f''(1) > 0$. Then, for any $Q_X$ that
is not a point mass,
\begin{align}
\label{chi^2 contraction is minimal}
\mu_{\chi^2}(Q_X, W_{Y|X}) \leq \mu_f(Q_X, W_{Y|X}),
\end{align}
i.e., the contraction coefficient for the $\chi^2$-divergence is
the minimal contraction coefficient among all $f$-divergences
with $f$ satisfying the above conditions.
\end{proposition}

\begin{remark}
A weaker version of \eqref{chi^2 contraction is minimal} was presented in
\cite[Proposition~II.6.15]{CohenKZ98} in the general alphabet setting, and the
result in \eqref{chi^2 contraction is minimal} was obtained in \cite[Theorem~3.3]{Raginsky16}
for finite alphabets.
\end{remark}

The following result provides an upper bound on the contraction coefficient for
a subclass of $f$-divergences in the finite alphabet setting.
\begin{proposition} \cite[Theorem~8]{MakurZ18}
\label{prop.: MakurZ18}
Let $f \colon [0, \infty) \to \Reals$ be a continuous convex function which is three
times differentiable at unity with $f(1)=0$ and $f''(1)>0$, and let it further satisfy
the following conditions:
\begin{enumerate}[a)]
\item  \label{MakurZ18.1}
\begin{align}
\label{Gilardoni}
\Bigl( f(t)-f'(1) \, (t-1) \Bigr) \left(1 - \frac{f^{(3)}(1) (t-1)}{3 f''(1)} \right)
\geq \tfrac12 f''(1) (t-1)^2, \quad \forall \, t > 0.
\end{align}
\item  \label{MakurZ18.2}
The function $g \colon (0, \infty) \to \Reals$, given by $g(t):= \frac{f(t)-f(0)}{t}$
for all $t>0$, is concave.
\end{enumerate}
Then, for a probability mass function $Q_X$ supported over a finite set $\set{X}$,
\begin{align}
\mu_f(Q_X, W_{Y|X}) \leq \left( \frac{f'(1)+f(0)}{f''(1) \;
\underset{x \in \set{X}}{\min} \, Q_X(x)} \right) \mu_{\chi^2}(Q_X, W_{Y|X}).
\end{align}
\end{proposition}

For the presentation of our majorization inequalities
for $f$-divergences and related entropy bounds (see
Section~\ref{subsection: Majorization and entropy bounds}),
essential definitions and basic results are next provided
(see, e.g., \cite[Chapter~2]{Bhatia}, \cite{MarshallOA}
and \cite[Chapter~13]{Steele}).
Let $P$ be a probability mass function defined on a finite set $\set{X}$, let
$p_{\max}$ be the maximal mass of $P$, and let $G_P(k)$ be the sum of
the $k$ largest masses of $P$ for $k \in \{1, \ldots, |\set{X}|\}$ (hence,
it follows that $G_P(1) = p_{\max}$ and $G_P(|\set{X}|) = 1$).

\begin{definition}
\label{definition: majorization}
Consider discrete probability mass functions $P$ and $Q$ defined on
a finite set $\set{X}$. It is said that $P$ is majorized by $Q$
(or $Q$ majorizes $P$), and it is denoted by $P \prec Q$, if
$G_P(k) \leq G_Q(k)$ for all $k \in \{1, \ldots, |\set{X}|\}$ (recall
that $G_P(|\set{X}|)=G_Q(|\set{X}|)=1$).
\end{definition}

A unit mass majorizes any other distribution; on the other hand, the
equiprobable distribution on a finite set is majorized by any other
distribution defined on the same set.

\begin{definition}
Let $\set{P}_n$ denote the set of all the probability mass functions
that are defined on $\set{A}_n := \{1, \ldots, n\}$.
A function $f \colon \set{P}_n \to \Reals$ is said to be {\em Schur-convex}
if for every $P,Q \in \set{P}_n$ such that $P \prec Q$, we have $f(P) \leq f(Q)$.
Likewise, $f$ is said to be {\em Schur-concave} if $-f$ is Schur-convex, i.e.,
$P,Q \in \set{P}_n$ and $P \prec Q$ imply that $f(P) \geq f(Q)$.
\end{definition}

Characterization of Schur-convex functions is provided, e.g., in \cite[Chapter~3]{MarshallOA}.
For example, there exist some connections between convexity and Schur-convexity (see, e.g.,
\cite[Section~3.C]{MarshallOA} and \cite[Chapter~2.3]{Bhatia}). However, a
Schur-convex function is not necessarily convex (\cite[Example~2.3.15]{Bhatia}).

Finally, what is the connection between data processing and majorization, and
why these types of inequalities are both considered in the same manuscript ?
This connection is provided in the following fundamental well-known result
(see, e.g., \cite[Theorem~2.1.10]{Bhatia}, \cite[Theorem~B.2]{MarshallOA}
and \cite[Chapter~13]{Steele}):
\begin{proposition}
\label{proposition: majorization and DP}
Let $P$ and $Q$ be probability mass functions defined on a finite set $\set{A}$.
Then, $P \prec Q$ if and only if there is a doubly-stochastic transformation
$W_{Y|X} \colon \set{A} \to \set{A}$
(i.e., $\underset{x \in \set{A}}{\sum} W_{Y|X}(y|x) = 1$ for all $y \in \set{A}$, and
$\underset{y \in \set{A}}{\sum} W_{Y|X}(y|x) = 1$ for all $x \in \set{A}$ with
$W_{Y|X}(\cdot|\cdot) \geq 0$) such that
$Q \rightarrow W_{Y|X} \rightarrow P$. In other words, $P \prec Q$ if and only if
in their representation as column vectors,
there exists a doubly-stochastic matrix ${\bf{W}}$ (i.e., a square matrix with non-negative
entries such that the sum of each column or each row in ${\bf{W}}$ is equal to~1)
such that $P = {\bf{W}} Q$.
\end{proposition}

\subsection{Contributions}
This paper is focused on the derivation of data-processing and majorization
inequalities for $f$-divergences, and it applies these inequalities to
information theory and statistics.

The starting point for obtaining strong data-processing inequalities in this
paper relies on the derivation of bounds on the difference
$D_f(P_X \| Q_X) - D_f(P_Y \| Q_Y)$ where $(P_X, Q_X)$ and $(P_Y, Q_Y)$ denote,
respectively, pairs of input and output probability distributions with a given
stochastic transformation $W_{Y|X}$ (i.e., $P_X \to W_{Y|X} \to P_Y$, and
$Q_X \to W_{Y|X} \to Q_Y$). These bounds are expressed in terms of the respective
difference in the Pearson's or Neyman's $\chi^2$-divergence, and they hold for
all $f$-divergences (see Theorems~\ref{thm: SDPI-IS}--\ref{Thm: DMS-DMC}).
By a different approach, we derive an upper bound on the contraction coefficient
for $f$-divergences of a certain type, which gives an alternative strong
data-processing inequality for the considered type of $f$-divergences (see
Theorems~\ref{theorem: contraction coef}--\ref{theorem: DMS/DMC - ver2}). In
this framework, a parametric subclass of $f$-divergences is introduced, its
interesting properties are studied (see Theorem~\ref{thm: f_alpha-divergence}),
all the data-processing inequalities which are derived in this paper are applied
to this subclass, and these inequalities are exemplified numerically to examine
their tightness (see Section~\ref{subsection: Illustration of Thm. DMS-DMC}).

This paper also derives majorization inequalities for $f$-divergences where
part of these inequalities rely on the earlier data-processing inequalities
(see Theorem~\ref{thm: majorization Df}).
A different approach, which relies on the concept of majorization, serves to
derive tight bounds on the maximal value of an $f$-divergence from a probability
mass function $P$ to an equiprobable distribution; the maximization is
carried over all $P$ with a fixed finite support where the ratio of their maximal
to minimal probability masses does not exceed a given value (see
Theorem~\ref{thm: LB/UB f-div}). These bounds
lead to accurate asymptotic results which apply to general $f$-divergences,
and they strengthen and generalize recent results of this type with respect to the
relative entropy \cite{CicaleseGV18}, and the R\'{e}nyi divergence \cite{Sason18b}.
Furthermore, we explore in Theorem~\ref{thm: LB/UB f-div} the convergence rates to
the asymptotic results. Data-processing and majorization inequalities
also serve to strengthen the Schur-concavity property of the
Tsallis entropy (see Theorem~\ref{thm: bounds Tsallis}), showing by a
comparison to earlier bounds in \cite{HoS-IT10} and \cite{HoS-ISIT15}
that none of these bounds is superseded by the other.
Further analytical results which are related to the specialization of
our central result on majorization inequalities in Theorem~\ref{thm: LB/UB f-div},
applied to several important sub-classes of $f$-divergences, are provided in
Section~\ref{subsection: Illustration of Thm. f_alpha-divergence} (including
Theorem~\ref{theorem: Delta_alpha}). A quantity which is involved in our majorization
inequalities in Theorem~\ref{thm: LB/UB f-div} is interpreted by relying on a
variational representation of $f$-divergences (see Theorem~\ref{thm: conjugate}).

As an application of the data-processing inequalities for $f$-divergences,
the setup of list decoding is further studied, reproducing
in a unified way some known bounds on the list decoding error probability,
and deriving new bounds for fixed and variable list sizes (see
Theorems~\ref{theorem: generalized Fano Df}--\ref{theorem: LB - variable list size}).

As an application of the majorization inequalities in this paper, we study
properties of a measure which is used to quantify the quality of
approximating probability mass functions, induced by the leaves of a
Tunstall tree, by an equiprobable distribution (see
Theorem~\ref{theorem: closeness to equiprobable}). An application of
majorization inequalities for the relative entropy is used to derive
a sufficient condition, expressed in terms of the principal and secondary
real branches of the Lambert $W$ function \cite{Corless96}, for asserting
the proximity of compression rates of finite-length (lossless and
variable-to-fixed) Tunstall codes to the Shannon entropy of a memoryless
and stationary discrete source (see Theorem~\ref{theorem: p_min Tunstall}).

\subsection{Paper Organization}
The paper is structured as follows:
Section~\ref{section: main results} provides our main new results on
data-processing and majorization inequalities for $f$-divergences and related
entropy measures. Illustration of the theorems in Section~\ref{section: main results},
and further mathematical results which follow from these theorems are introduced
in Section~\ref{section: Examples}. Applications in information theory and statistics
are considered in Section~\ref{section: applications}. Proofs of all theorems are
relegated to the appendices, which form a major part of this paper.

\section{Main Results on $f$-divergences}
\label{section: main results}
This section provides strong data-processing inequalities for $f$-divergences
(see Section~\ref{subsection: new DPIs for f-divergences}), followed by a study
of a new subclass of $f$-divergences (see
Section~\ref{subsection: a new class of f-divergences}) which later serves to
exemplify our data-processing inequalities. The third part of this section
(see Section~\ref{subsection: Majorization and entropy bounds}) provides
majorization inequalities for $f$-divergences, and for the Tsallis entropy,
whose derivation relies in part on the new data-processing inequalities.

\subsection{Data-processing inequalities for $f$-divergences}
\label{subsection: new DPIs for f-divergences}

Strong data-processing inequalities are provided in the following,
bounding the difference $D_f(P_X \| Q_X) - D_f(P_Y \| Q_Y)$ and ratio
$\frac{D_f(P_Y \| Q_Y)}{D_f(P_X \| Q_X)}$ where $(P_X, Q_X)$ and
$(P_Y, Q_Y)$ denote, respectively, pairs of input and output probability
distributions with a given stochastic transformation.

\begin{theorem} \label{thm: SDPI-IS}
Let $\set{X}$ and $\set{Y}$ be finite or countably infinite sets,
let $P_X$ and $Q_X$ be probability mass functions that are
supported on $\set{X}$, and let
\begin{align}
& \xi_1 := \inf_{x \in \set{X}} \frac{P_X(x)}{Q_X(x)} \in [0,1], \label{xi1} \\
& \xi_2 := \sup_{x \in \set{X}} \frac{P_X(x)}{Q_X(x)} \in [1, \infty] \label{xi2}.
\end{align}
Let $W_{Y|X} \colon \set{X} \to \set{Y}$ be a
stochastic transformation such that for every $y \in \set{Y}$, there
exists $x \in \set{X}$ with $W_{Y|X}(y|x) > 0$, and let (see
\eqref{23062019a1} and \eqref{23062019a2})
\begin{align}
& P_Y := P_X W_{Y|X}, \label{transP} \\
& Q_Y := Q_X W_{Y|X}. \label{transQ}
\end{align}
Furthermore, let $f \colon (0, \infty) \to \Reals$ be a convex function with
$f(1)=0$, and let the non-negative constant $c_f := c_f(\xi_1, \xi_2)$ satisfy
\begin{align}
\label{condition on c_f}
f'_{+}(v) - f'_{+}(u) \geq 2 c_f \, (v-u), \quad
\forall \, u,v \in \set{I}, \; u<v
\end{align}
where $f'_{+}$ denotes the right-side derivative of $f$, and
\begin{align}
\label{I_interval}
\set{I} := \set{I}(\xi_1, \xi_2) = [\xi_1, \xi_2] \cap (0, \infty).
\end{align}
Then,
\begin{enumerate}[a)]

\item \label{Th. 1.a}
\begin{align}
\label{key}
D_f(P_X \| Q_X) - D_f(P_Y \| Q_Y) &\geq c_f(\xi_1, \xi_2)
\left[ \chi^2(P_X \| Q_X) - \chi^2(P_Y \| Q_Y) \right] \\
&\geq 0, \label{DPI1}
\end{align}
where equality holds in \eqref{key} if $D_f(\cdot \| \cdot)$ is Pearson's
$\chi^2$-divergence with $c_f \equiv 1$.

\item \label{Th. 1.b}
If $f$ is twice differentiable on $\set{I}$, then
the largest possible coefficient in the right side of
\eqref{condition on c_f} is given by
\begin{align}  \label{c_f}
c_f(\xi_1, \xi_2) = \tfrac12 \, \inf_{t \in \set{I}(\xi_1, \xi_2)} f''(t).
\end{align}

\item \label{Th. 1.c}
Under the assumption in Item~\ref{Th. 1.b}), the following dual inequality
also holds:
\begin{align}
\label{key-dual}
D_f(P_X \| Q_X) - D_f(P_Y \| Q_Y) &\geq
c_{f^\ast}\hspace*{-0.1cm}\left(\tfrac1{\xi_2}, \tfrac1{\xi_1} \right) \,
\left[ \chi^2(Q_X \| P_X) - \chi^2(Q_Y \| P_Y) \right] \\
&\geq 0, \label{DPI2}
\end{align}
where $f^\ast \colon (0, \infty) \to \Reals$ is the dual convex
function which is given by
\begin{align}
\label{dual f}
f^\ast(t) := t \, f\biggl(\frac1t\biggr), \quad \forall \, t>0,
\end{align}
and the coefficient in the right side of \eqref{key-dual}
satisfies
\begin{align}
\label{c_f^*}
c_{f^\ast}\hspace*{-0.1cm}\left(\tfrac1{\xi_2}, \tfrac1{\xi_1} \right) =
\tfrac12 \, \inf_{t \in \set{I}(\xi_1, \xi_2)} \{t^3 \, f''(t) \}
\end{align}
with the convention that $\frac1{\xi_1} = \infty$ if $\xi_1 = 0$.
Equality holds in \eqref{key-dual} if $D_f(\cdot \| \cdot)$ is Neyman's
$\chi^2$-divergence (i.e., $D_f(P\|Q) := \chi^2(Q\|P)$ for all $P$ and
$Q$) with $c_{f^\ast} \equiv 1$.

\item \label{Th. 1.d}
Under the assumption in Item~\ref{Th. 1.b}), if
\begin{align}
\label{e_f}
e_f(\xi_1, \xi_2) := \tfrac12 \, \sup_{t \in \set{I}(\xi_1, \xi_2)} f''(t) < \infty,
\end{align}
then,
\begin{align}
\label{UB1}
D_f(P_X \| Q_X) - D_f(P_Y \| Q_Y)
& \leq e_f(\xi_1, \xi_2) \left[ \chi^2(P_X \| Q_X) - \chi^2(P_Y \| Q_Y) \right].
\end{align}
Furthermore,
\begin{align}
\label{UB2}
D_f(P_X \| Q_X) - D_f(P_Y \| Q_Y)
& \leq e_{f^\ast}\hspace*{-0.1cm}\left(\tfrac1{\xi_2}, \tfrac1{\xi_1} \right)
\left[ \chi^2(Q_X \| P_X) - \chi^2(Q_Y \| P_Y) \right]
\end{align}
where the coefficient in the right side of \eqref{UB2} satisfies
\begin{align}
\label{e_f^*}
e_{f^\ast}\hspace*{-0.1cm}\left(\tfrac1{\xi_2}, \tfrac1{\xi_1} \right)
= \tfrac12 \, \sup_{t \in \set{I}(\xi_1, \xi_2)} \{t^3 \, f''(t) \},
\end{align}
which is assumed to be finite.
Equalities hold in \eqref{UB1} and \eqref{UB2} if $D_f(\cdot \| \cdot)$ is
Pearson's or Neyman's $\chi^2$-divergence with $e_f \equiv 1$ or
$e_{f^\ast} \equiv 1$, respectively.

\item  \label{Th. 1.e}
The lower and upper bounds in \eqref{key}, \eqref{key-dual}, \eqref{UB1}
and \eqref{UB2} are locally tight. More precisely, let $\{P_X^{(n)}\}$ be
a sequence of probability mass functions defined on $\set{X}$ and pointwise
converging to $Q_X$ which is supported on $\set{X}$, and let $P_Y^{(n)}$
and $Q_Y$ be the probability mass functions defined on $\set{Y}$ via
\eqref{transP} and \eqref{transQ} with inputs $P_X^{(n)}$ and $Q_X$,
respectively. Suppose that
\begin{align}
\label{13062019a1}
& \lim_{n \to \infty} \inf_{x \in \set{X}} \frac{P_X^{(n)}(x)}{Q_X(x)} = 1, \\
\label{13062019a2}
& \lim_{n \to \infty} \sup_{x \in \set{X}} \frac{P_X^{(n)}(x)}{Q_X(x)} = 1.
\end{align}
If $f$ has a continuous second derivative at unity, then
\begin{align}
\label{tight}
& \lim_{n \to \infty} \frac{D_f(P_X^{(n)} \| Q_X) - D_f(P_Y^{(n)} \| Q_Y)}
{\chi^2(P_X^{(n)} \| Q_X) - \chi^2(P_Y^{(n)} \| Q_Y)} = \tfrac12 f''(1), \\[0.1cm]
\label{tight-dual}
& \lim_{n \to \infty} \frac{D_f(P_X^{(n)} \| Q_X) - D_f(P_Y^{(n)} \| Q_Y)}
{\chi^2(Q_X \| P_X^{(n)}) - \chi^2(Q_Y \| P_Y^{(n)})} = \tfrac12 f''(1),
\end{align}
which indicate the local tightness of the lower
and upper bounds in Items~\ref{Th. 1.a})--\ref{Th. 1.d}).
\end{enumerate}
\end{theorem}

\begin{proof}
See Appendix~\ref{appendix: SDPI-IS}.
\end{proof}

An application of Theorem~\ref{thm: SDPI-IS} gives the following result.

\begin{theorem} \label{Thm: DMS-DMC}
Let $\set{X}$ and $\set{Y}$ be finite or countably infinite sets, let $n \in \naturals$
be an arbitrary natural number,
and let $X^n := (X_1, \ldots, X_n)$ and $Y^n := (Y_1, \ldots, Y_n)$ be random vectors
taking values on $\set{X}^n$ and $\set{Y}^n$, respectively. Let $P_{X^n}$ and
$Q_{X^n}$ be the probability mass functions of discrete memoryless sources
where, for all $\underline{x} \in \set{X}^n$,
\begin{align}
\label{2DMS}
P_{X^n}(\underline{x}) = \prod_{i=1}^n P_{X_i}(x_i), \quad
Q_{X^n}(\underline{x}) = \prod_{i=1}^n Q_{X_i}(x_i),
\end{align}
with $P_{X_i}$ and $Q_{X_i}$ supported on $\set{X}$ for all $i \in \{1, \ldots, n\}$.
Let each symbol $X_i$ be independently selected from one of the source outputs at
time instant $i$ with probabilities $\lambda$ and $1-\lambda$, respectively, and let
it be transmitted over a discrete memoryless channel with transition probabilities
\begin{align}
\label{DMC}
W_{Y^n | X^n}(\underline{y} \, | \, \underline{x})
= \prod_{i=1}^n W_{Y_i | X_i}(y_i | x_i), \quad \forall \,
\underline{x} \in \set{X}^n, \; \underline{y} \in \set{Y}^n.
\end{align}
Let $R_{X^n}^{(\lambda)}$ be the probability mass function of the symbols
at the channel input, i.e.,
\begin{align}
\label{prod. RX}
R_{X^n}^{(\lambda)}(\underline{x}) = \prod_{i=1}^n \bigl( \lambda P_{X_i}(x_i)
+ (1-\lambda) Q_{X_i}(x_i) \bigr), \quad \forall \, \underline{x} \in \set{X}^n,
\; \lambda \in [0,1],
\end{align}
let
\begin{align}
\label{MC1 in DMC}
& R_{Y^n}^{(\lambda)} := R_{X^n}^{(\lambda)} \, W_{Y^n | X^n}, \\
\label{MC2 in DMC}
& P_{Y^n} := P_{X^n} W_{Y^n | X^n}, \\
\label{MC3 in DMC}
& Q_{Y^n} := Q_{X^n} W_{Y^n | X^n},
\end{align}
and let $f \colon (0, \infty) \to \Reals$ be a convex and twice
differentiable function with $f(1)=0$. Then,
\begin{enumerate}[a)]
\item \label{Th. 2.a.}
For all $\lambda \in [0,1]$,
\begin{align}
& D_f(R_{X^n}^{(\lambda)} \, \| \, Q_{X^n}) - D_f(R_{Y^n}^{(\lambda)} \, \| \, Q_{Y^n})
\nonumber \\
\label{LB1 - DMC}
& \geq c_f\bigl(\xi_1(n, \lambda), \, \xi_2(n, \lambda) \bigr)
\left[ \, \prod_{i=1}^n \bigl(1 + \lambda^2 \, \chi^2(P_{X_i} \| Q_{X_i}) \bigr)
- \prod_{i=1}^n \bigl(1 + \lambda^2 \, \chi^2(P_{Y_i} \| Q_{Y_i}) \bigr) \right] \\
\label{LB2 - DMC}
& \geq c_f\bigl(\xi_1(n, \lambda), \, \xi_2(n, \lambda) \bigr) \, \lambda^2 \,
\sum_{i=1}^n \bigl[ \chi^2(P_{X_i} \| Q_{X_i}) - \chi^2(P_{Y_i} \| Q_{Y_i}) \bigr]
\geq 0,
\end{align}
where $c_f(\cdot, \cdot)$ in the right sides of \eqref{LB1 - DMC} and \eqref{LB2 - DMC}
is given in \eqref{c_f}, and
\begin{align}
\label{xi1_n}
& \xi_1(n, \lambda) := \prod_{i=1}^n \left(1 - \lambda + \lambda \,
\inf_{x \in \set{X}} \frac{P_{X_i}(x)}{Q_{X_i}(x)} \right) \in [0,1], \\[0.1cm]
\label{xi2_n}
& \xi_2(n, \lambda) := \prod_{i=1}^n \left(1 - \lambda + \lambda \,
\sup_{x \in \set{X}} \frac{P_{X_i}(x)}{Q_{X_i}(x)} \right) \in [1, \infty].
\end{align}

\item  \label{Th. 2.b.}
For all $\lambda \in [0,1]$,
\begin{align}
& D_f(R_{X^n}^{(\lambda)} \, \| \, Q_{X^n}) - D_f(R_{Y^n}^{(\lambda)} \, \| \, Q_{Y^n})
\nonumber \\[0.1cm]
\label{UB1 - DMC}
& \leq e_f\bigl(\xi_1(n, \lambda), \, \xi_2(n, \lambda) \bigr)
\left[ \, \prod_{i=1}^n \bigl(1 + \lambda^2 \, \chi^2(P_{X_i} \| Q_{X_i}) \bigr)
- \prod_{i=1}^n \bigl(1 + \lambda^2 \, \chi^2(P_{Y_i} \| Q_{Y_i}) \bigr) \right]
\end{align}
where $e_f(\cdot, \cdot)$, $\xi_1(\cdot, \cdot)$ and $\xi_2(\cdot, \cdot)$
in the right side of \eqref{UB1 - DMC} are given in \eqref{e_f}, \eqref{xi1_n}
and \eqref{xi2_n}, respectively.

\item  \label{Th. 2.c.}
If $f$ has a continuous second derivative at unity, and
$\underset{x \in \set{X}}{\sup} \frac{P_{X_i}(x)}{Q_{X_i}(x)} < \infty$
for all $i \in \{1, \ldots, n\}$, then
\begin{align}
& \lim_{\lambda \to 0^+}
\frac{D_f(R_{X^n}^{(\lambda)} \, \| \, Q_{X^n})
- D_f(R_{Y^n}^{(\lambda)} \, \| \, Q_{Y^n})}{\lambda^2}
\nonumber \\[0.1cm]
\label{lim - DMC}
& = \tfrac12 \, f''(1) \,
\sum_{i=1}^n \bigl[ \chi^2(P_{X_i} \| Q_{X_i})
- \chi^2(P_{Y_i} \| Q_{Y_i}) \bigr].
\end{align}
The lower bounds in the right sides of \eqref{LB1 - DMC} and
\eqref{LB2 - DMC}, and the upper bound in the right side of
\eqref{UB1 - DMC} are tight as we let $\lambda \to 0^+$,
yielding the limit in the right side of \eqref{lim - DMC}.
\end{enumerate}
\end{theorem}

\begin{proof}
See Appendix~\ref{appendix: DMS-DMC}.
\end{proof}

\vspace*{0.2cm}
\begin{remark}
Similar upper and lower bounds on
$D_f(P_{X^n} \, \| \, R_{X^n}^{(\lambda)}) -
D_f(P_{Y^n} \, \| \, R_{Y^n}^{(\lambda)})$
can be obtained for all $\lambda \in [0,1]$.
To that end, in \eqref{LB1 - DMC}--\eqref{UB1 - DMC},
one needs to replace $f$ with $f^\ast$, switch between
$P_{X_i}$ and $Q_{X_i}$ for all $i$, and replace $\lambda$
with $1-\lambda$.
\end{remark}

\vspace*{0.1cm}
In continuation to \cite[Theorem~8]{MakurZ18} (see Proposition~\ref{prop.: MakurZ18}
in Section~\ref{subsection: preliminaries}), we next provide an upper bound on the
contraction coefficient for another subclass of $f$-divergences. Although the first
part of the next result is stated for finite or countably infinite alphabets, it is
clear from its proof that it also holds in the general alphabet setting. Connections
to the literature are provided in Remarks~\ref{remark: Raginsky16}--\ref{remark: ISSV16}
(see Appendix~\ref{appendix: contraction coef.}, Part~A).

\begin{theorem}
\label{theorem: contraction coef}
Let $f \colon (0, \infty) \to \Reals$ satisfy the conditions:
\begin{itemize}
\item $f$ is a convex function, differentiable at~1, $f(1)=0$, and
$f(0) := \underset{t \to 0^+}{\lim} f(t) < \infty$;
\item The function $g \colon (0, \infty) \to \Reals$, defined by
$g(t) := \frac{f(t)-f(0)}{t}$ for all $t>0$, is convex.
\end{itemize}
Let
\begin{align}
\label{def: kappa}
\kappa(\xi_1, \xi_2) :=
\sup_{t \in (\xi_1, 1) \cup (1, \xi_2)} \frac{f(t) + f'(1) \, (1-t)}{(t-1)^2}
\end{align}
where, for $P_X$ and $Q_X$ which are non-identical probability mass functions,
$\xi_1 \in [0,1)$ and $\xi_2 \in (1, \infty]$ are given in \eqref{xi1}
and \eqref{xi2}. Then, in the setting of \eqref{transP} and \eqref{transQ},
\begin{align}
\label{13062019c1}
\frac{D_f(P_Y \| Q_Y)}{D_f(P_X \| Q_X)} \leq \frac{\kappa(\xi_1, \xi_2)}{f(0)+f'(1)}
\cdot \frac{\chi^2(P_Y \| Q_Y)}{\chi^2(P_X \| Q_X)}.
\end{align}
Consequently, if $Q_X$ is finitely supported on $\set{X}$,
\begin{align}
\label{13062019c2}
\mu_f(Q_X, W_{Y|X}) \leq
\frac1{f(0)+f'(1)} \cdot \kappa\biggl(0, \frac1{\underset{x \in \set{X}}{\min} \, Q_X(x)}\biggr)
\cdot \mu_{\chi^2}(Q_X, W_{Y|X}).
\end{align}
\end{theorem}

\begin{proof}
See Appendix~\ref{appendix: contraction coef.} (Part~A).
\end{proof}

Similarly to the extension of Theorem~\ref{thm: SDPI-IS} to Theorem~\ref{Thm: DMS-DMC},
a similar extension of Theorem~\ref{theorem: contraction coef} leads to the following
result.

\begin{theorem}
\label{theorem: DMS/DMC - ver2}
In the setting of \eqref{2DMS}--\eqref{MC3 in DMC} in Theorem~\ref{Thm: DMS-DMC},
and under the assumptions on $f$ in Theorem~\ref{theorem: contraction coef}, the
following holds for all $\lambda \in (0,1]$:
\begin{align}
\label{13062019d1}
\frac{D_f\bigl(R_{Y^n}^{(\lambda)} \, \| \, Q_{Y^n}\bigr)}{D_f\bigl(R_{X^n}^{(\lambda)} \, \| \, Q_{X^n}\bigr)}
\leq \frac{\kappa\bigl(\xi_1(n, \lambda), \, \xi_2(n, \lambda)\bigr)}{f(0)+f'(1)} \;
\frac{\overset{n}{\underset{i=1}{\prod}} \Bigl( 1+ \lambda^2 \,
\chi^2(P_{Y_i} \, \| \, Q_{Y_i} \bigr) \Bigr) - 1}{\overset{n}{\underset{i=1}{\prod}}
\Bigl( 1+ \lambda^2 \, \chi^2(P_{X_i} \, \| \, Q_{X_i} \bigr) \Bigr) - 1},
\end{align}
with $\xi_1(n,\lambda)$ and $\xi_2(n,\lambda)$ and $\kappa(\cdot, \cdot)$ defined in
\eqref{xi1_n}, \eqref{xi2_n} and \eqref{def: kappa}, respectively.
\end{theorem}

\begin{proof}
See Appendix~\ref{appendix: contraction coef.} (Part~B).
\end{proof}

\subsection{A subclass of $f$-divergences}
\label{subsection: a new class of f-divergences}

A subclass of $f$-divergences with interesting properties is
introduced in Theorem~\ref{thm: f_alpha-divergence}. The data-processing
inequalities in Theorems~\ref{Thm: DMS-DMC} and~\ref{theorem: DMS/DMC - ver2}
are applied to these $f$-divergences in Section~\ref{section: Examples}.

\begin{theorem} \label{thm: f_alpha-divergence}
Let $f_\alpha \colon [0, \infty) \to \Reals$ be given by
\begin{align}
\label{f_alpha}
f_\alpha(t) := (\alpha+t)^2 \log(\alpha+t)
- (\alpha+1)^2 \log(\alpha+1), \quad t \geq 0
\end{align}
for all $\alpha \geq \mathrm{e}^{-\frac32}$. Then,
\begin{enumerate}[a)]
\item \label{Thm. f-1}
$D_{f_\alpha}(\cdot \| \cdot)$ is an $f$-divergence which is
monotonically increasing and concave in $\alpha$, and its first
three derivatives are related to the relative
entropy and $\chi^2$-divergence as follows:
\begin{align}
\label{1st partial der.}
& \frac{\partial}{\partial \alpha} \bigl\{ D_{f_\alpha}(P\|Q) \bigr\}
= 2(\alpha+1) \, D\Bigl( \tfrac{\alpha Q + P}{\alpha+1} \, \| \, Q \Bigr), \\
\label{2nd partial der.}
& \frac{\partial^2}{\partial \alpha^2} \bigl\{ D_{f_\alpha}(P\|Q) \bigr\}
= -2 \, D\Bigl(Q \, \| \, \tfrac{\alpha Q + P}{\alpha+1} \Bigr), \\
\label{3rd partial der.}
& \frac{\partial^3}{\partial \alpha^3} \bigl\{ D_{f_\alpha}(P\|Q) \bigr\}
= \frac{2 \log \mathrm{e}}{\alpha+1} \cdot
\chi^2\Bigl(Q \, \| \, \tfrac{\alpha Q + P}{\alpha+1} \Bigr).
\end{align}

\item \label{Thm. f-1b}
For every $n \in \naturals$,
\begin{align}
\label{generalized}
(-1)^{n-1} \, \frac{\partial^n}{\partial \alpha^n} \bigl\{ D_{f_\alpha}(P\|Q) \bigr\} \geq 0,
\end{align}
and, in addition to \eqref{1st partial der.}--\eqref{3rd partial der.}, for all $n>3$
\begin{align}
& \frac{\partial^n}{\partial \alpha^n} \bigl\{ D_{f_\alpha}(P\|Q) \bigr\} \nonumber \\[0.1cm]
\label{nth der. - RD}
&= \frac{2(-1)^{n-1} (n-3)! \, \log \mathrm{e}}{(\alpha+1)^{n-2}}
\left[ \exp \biggl( (n-2) \, D_{n-1}\Bigl(Q \, \| \, \tfrac{\alpha Q + P}{\alpha+1} \Bigr)
\biggr) -1 \right],
\end{align}
where $D_{n-1}(\cdot \| \cdot)$ in the right side of \eqref{nth der. - RD} denotes the
R\'{e}nyi divergence of order $n-1$.

\item \label{Thm. f-2}
\begin{align} \label{Df-chi2}
D_{f_\alpha}(P\|Q) & \geq k(\alpha) \, \chi^2(P\|Q) \\
\label{Df-KL}
& \geq k(\alpha) \, \left[ \exp\bigl(D(P\|Q)\bigr) - 1 \right]
\end{align}
where the function $k \colon [\mathrm{e}^{-\frac32}, \infty) \to \Reals$
is defined as
\begin{align}
\label{k_alpha}
k(\alpha) := \log (\alpha+1) + \tfrac32 \log \mathrm{e} - \frac{\log \mathrm{e}}{3\alpha},
\end{align}
which is monotonically increasing in $\alpha$, satisfying
$k(\alpha)\geq 0.2075 \log \mathrm{e}$ for all $\alpha \geq \mathrm{e}^{-\frac32}$,
and it tends to infinity as we let $\alpha \to \infty$.
Consequently, unless $P \equiv Q$,
\begin{align} \label{lim Infty}
\lim_{\alpha \to \infty} D_{f_\alpha}(P\|Q) = +\infty.
\end{align}

\item \label{Thm. f-2b}
\begin{align} \label{Df_UB}
D_{f_\alpha}(P\|Q) & \leq \Bigl[ \log(\alpha+1) + \tfrac32 \log \mathrm{e}
- \frac{\log \mathrm{e}}{\alpha+1} \Bigr] \, \chi^2(P\|Q) \nonumber \\
& \hspace*{0.4cm} + \frac{\log \mathrm{e}}{3(\alpha+1)}
\Bigl[ \exp\bigl(2 D_3(P\|Q)\bigr) - 1 \Bigr].
\end{align}

\item \label{Thm. f-2c}
For every $\varepsilon > 0$ and a pair of probability mass
functions $(P,Q)$ where $D_3(P\|Q) < \infty$, there exists
$\alpha^\ast := \alpha(P,Q, \varepsilon)$ such that for all
$\alpha > \alpha^\ast$
\begin{align} \label{asymp.}
\Bigl| D_{f_\alpha}(P\|Q) - \bigl[ \log(\alpha+1) + \tfrac32 \log \mathrm{e} \bigr]
\, \chi^2(P\|Q) \Bigr| < \varepsilon.
\end{align}

\item \label{Thm. f-3}
If a sequence of probability measures $\{P_n\}$ converges to a probability
measure $Q$ such that
\begin{align}
\label{eq: 1st condition}
\lim_{n \to \infty} \text{ess\,sup} \, \frac{\text{d}P_n}{\text{d}Q} \, (Y) = 1,
\quad Y \sim Q,
\end{align}
where $P_n \ll Q$ for all sufficiently large $n$, then
\begin{align} \label{eq: limit of ratio of f-div}
\lim_{n \to \infty} \frac{D_{f_\alpha}(P_n \| Q)}{\chi^2(P_n \| Q)}
= \log(\alpha+1) + \tfrac32 \log \mathrm{e}.
\end{align}

\item \label{Thm. f-4}
If $\alpha > \beta \geq \mathrm{e}^{-\frac32}$, then
\begin{align} \label{diff Df1}
0 & \leq (\alpha-\beta)(\alpha+\beta+2)
\, D\Bigl( \tfrac{\alpha Q + P}{\alpha+1} \, \| \, Q \Bigr) \\
\label{diff Df1b}
& \leq D_{f_\alpha}(P\|Q) - D_{f_\beta}(P\|Q) \\
\label{diff Df2}
& \leq (\alpha-\beta)\, \min \left\{ (\alpha+\beta+2)
\, D\Bigl( \tfrac{\beta Q + P}{\beta+1} \, \| \, Q \Bigr),
\; 2 D(P\|Q) \right\}.
\end{align}

\item \label{Thm. f-5}
The function $f_\alpha \colon [0, \infty) \to \Reals$, as given in
\eqref{f_alpha}, satisfies the conditions in Theorems~\ref{theorem: contraction coef}
and~\ref{theorem: DMS/DMC - ver2} for all $\alpha \geq \mathrm{e}^{-\frac32}$.
Furthermore, the corresponding function in \eqref{def: kappa} is equal to
\begin{align}
\label{kappa_alpha 1}
\kappa_\alpha(\xi_1, \xi_2) &:=
\sup_{t \in (\xi_1, 1) \cup (1, \xi_2)} \frac{f_\alpha(t) + f_\alpha'(1)
\, (1-t)}{(t-1)^2} \\[0.1cm]
\label{kappa_alpha 2}
&= \frac{f_\alpha(\xi_2) + f_\alpha'(1) \, (1-\xi_2)}{(\xi_2-1)^2}
\end{align}
for all $\xi_1 \in [0,1)$ and $\xi_2 \in (1, \infty)$.
\end{enumerate}
\end{theorem}

\begin{proof}
See Appendix~\ref{appendix: f_alpha-divergence}.
\end{proof}

\subsection{$f$-divergence Inequalities via Majorization}
\label{subsection: Majorization and entropy bounds}

Let $U_n$ denote an equiprobable distribution on
$\{1, \ldots, n\}$ ($n \in \naturals$), i.e., $U_n(i) := \tfrac1n$
for all $i \in \{1, \ldots, n\}$.
By majorization theory and Theorem~\ref{thm: SDPI-IS},
the next result strengthens the Schur-convexity property of the
$f$-divergence $D_f(\cdot\|U_n)$ (see \cite[Lemma~1]{CicaleseGV06}).

\begin{theorem} \label{thm: majorization Df}
Let $P$ and $Q$ be probability mass functions which are supported on
$\{1, \ldots, n\}$, and suppose that $P \prec Q$.
Let $f \colon (0, \infty) \to \Reals$ be twice differentiable and convex
with $f(1)=0$, and let $q_{\max}$ and $q_{\min}$ be, respectively, the
maximal and minimal positive masses of $Q$. Then,
\begin{enumerate}[a)]
\item
\begin{align}
n e_f(n q_{\min}, n q_{\max})
\, \bigl( \| Q \|_2^2 - \| P \|_2^2 \bigr)
\label{UB diff Df - equiprob.}
& \geq D_f(Q\|U_n) - D_f(P\|U_n) \\
\label{LB diff Df - equiprob.}
& \geq n c_f(n q_{\min}, n q_{\max})
\, \bigl( \| Q \|_2^2 - \| P \|_2^2 \bigr) \geq 0,
\end{align}
where $c_f(\cdot, \cdot)$ and $e_f(\cdot, \cdot)$ are given
in \eqref{c_f} and \eqref{e_f}, respectively, and
$\| \cdot \|_2$ denotes the Euclidean norm. Furthermore,
\eqref{UB diff Df - equiprob.} and \eqref{LB diff Df - equiprob.}
hold with equality if $D_f(\cdot \| \cdot) = \chi^2(\cdot \| \cdot)$.

\item
If $P \prec Q$ and $\frac{q_{\max}}{q_{\min}} \leq \rho$ for an arbitrary
$\rho \geq 1$, then
\begin{align}
\label{bounds on diff. norms}
0 \leq \| Q \|_2^2 - \| P \|_2^2 \leq \frac{(\rho-1)^2}{4\rho n}.
\end{align}
\end{enumerate}
\end{theorem}

\begin{proof}
See Appendix~\ref{appendix: majorization}.
\end{proof}

\begin{remark}
If $P$ is not supported on $\{1, \ldots, n\}$, then
\eqref{UB diff Df - equiprob.} and \eqref{LB diff Df - equiprob.}
hold if $f$ is also right continuous at zero.
\end{remark}

The next result provides bounds on $f$-divergences
from any probability mass function to an equiprobable distribution. It
relies on majorization theory, and Theorem~\ref{thm: majorization Df}.

\begin{theorem}
\label{thm: LB/UB f-div}
Let $\set{P}_n$ denote the set of all the probability mass functions
that are defined on $\set{A}_n := \{1, \ldots, n\}$. For $\rho \geq 1$,
let $\set{P}_n(\rho)$ be the set of all $Q \in \set{P}_n$ which are
supported on $\set{A}_n$ with $\frac{q_{\max}}{q_{\min}} \leq \rho$,
and let $f \colon (0, \infty) \to \Reals$ be a convex function with
$f(1)=0$. Then,
\begin{enumerate}[a)]
\item \label{Thm. 5-a}
The set $\set{P}_n(\rho)$, for any $\rho \geq 1$, is a non-empty,
convex and compact set.

\item \label{Thm. 5-b}
For a given $Q \in \set{P}_n$, which is supported on $\set{A}_n$, the
$f$-divergences $D_f(\cdot \| Q)$ and $D_f(Q \| \cdot)$ attain their
maximal values over the set $\set{P}_n(\rho)$.

\item \label{Thm. 5-c}
For $\rho \geq 1$ and an integer $n \geq 2$, let
\begin{align}
\label{def: u_f}
& u_f(n, \rho) := \max_{Q \in \set{P}_n(\rho)} D_f(Q \| U_n), \\[0.1cm]
\label{def: v_f}
& v_f(n, \rho) := \max_{Q \in \set{P}_n(\rho)} D_f(U_n \| Q),
\end{align}
let
\begin{align}
\label{Gamma interval}
\Gamma_n(\rho) := \biggl[\frac1{1+(n-1)\rho}, \, \frac1n \biggr],
\end{align}
and let the probability mass function $Q_\beta \in \set{P}_n(\rho)$
be defined on the set $\set{A}_n$ as follows:
\begin{align}
\label{Q_beta}
Q_\beta(j) :=
\begin{dcases}
\rho \beta, & \quad \mbox{if $j \in \{1, \ldots, i_\beta\}$,} \\
1-\bigl(n+i_\beta(\rho-1)-1\bigr) \beta, & \quad \mbox{if $j=i_\beta+1$,} \\
\beta, & \quad \mbox{if $j \in \{i_\beta + 2, \ldots, n\}$}
\end{dcases}
\end{align}
where
\begin{align}
\label{i_beta}
i_\beta := \biggl\lfloor \frac{1-n \beta}{(\rho-1) \beta} \biggr\rfloor.
\end{align}
Then,
\begin{align}
\label{opt1 Df: beta}
&  u_f(n, \rho)
= \max_{\beta \in \Gamma_n(\rho)} D_f(Q_\beta \| U_n), \\
\label{opt2 Df: beta}
&  v_f(n, \rho)
= \max_{\beta \in \Gamma_n(\rho)} D_f(U_n \| Q_\beta).
\end{align}

\item \label{Thm. 5-d}
For $\rho \geq 1$ and an integer $n \geq 2$, let the non-negative
function $g_{f}^{(\rho)} \colon [0,1] \to \Reals_+$ be given by
\begin{align}
\label{def: g_f}
g_f^{(\rho)}(x) := x f\biggl(\frac{\rho}{1+(\rho-1)x} \biggr)
+ (1-x) f\biggl(\frac1{1+(\rho-1)x} \biggr), \quad x \in [0,1].
\end{align}
Then,
\begin{align}
\label{LB/UB u_f}
& \max_{m \in \{0, \ldots, n\}} \;
g_f^{(\rho)}\bigl(\tfrac{m}{n}\bigr) \leq u_f(n, \rho) \leq
\max_{x \in [0,1]} \; g_f^{(\rho)}(x), \\
\label{LB/UB v_f}
& \max_{m \in \{0, \ldots, n\}} \;
g_{f^\ast}^{(\rho)}\bigl(\tfrac{m}{n}\bigr) \leq v_f(n, \rho)
\leq \max_{x \in [0,1]} \; g_{f^\ast}^{(\rho)}(x)
\end{align}
with the convex function $f^\ast \colon (0, \infty) \to \Reals$
in \eqref{dual f}.

\item \label{Thm. 5-e}
The right-side inequalities in \eqref{LB/UB u_f} and \eqref{LB/UB v_f}
are asymptotically tight ($n \to \infty$). More explicitly,
\begin{align}
& \lim_{n \to \infty} u_f(n, \rho) \nonumber \\
\label{asympt. max1}
& = \max_{x \in [0,1]} \left\{ x f\biggl(\frac{\rho}{1+(\rho-1)x} \biggr)
+ (1-x) f\biggl(\frac1{1+(\rho-1)x} \biggr) \right\}, \\[0.1cm]
& \lim_{n \to \infty} v_f(n, \rho) \nonumber \\
\label{asympt. max2}
& = \max_{x \in [0,1]} \Biggl\{ \frac{\rho x}{1+(\rho-1)x} \;
f\biggl(\frac{1+(\rho-1)x}{\rho} \biggr)
+ \frac{(1-x) \; f\bigl(1+(\rho-1)x \bigr)}{1+(\rho-1)x} \Biggr\}.
\end{align}

\item \label{Thm. 5-f}
If $g_{f}^{(\rho)}(\cdot)$ in \eqref{def: g_f} is differentiable on
$(0,1)$ and its derivative is upper bounded by $K_f(\rho) \geq 0$,
then for every integer $n \geq 2$
\begin{align}
\label{convergence rate 1}
0 \leq \lim_{n' \to \infty} \bigl\{ u_f(n', \rho) \bigr\} - u_f(n, \rho)
\leq \frac{K_f(\rho)}{n}.
\end{align}

\item \label{Thm. 5-g}
Let $f(0) := \underset{t \to 0}{\lim} \, f(t) \in (-\infty, +\infty]$,
and let $n \geq 2$ be an integer. Then,
\begin{align}
\label{asympt. max3}
& \lim_{\rho \to \infty} u_f(n, \rho)
= \left(1-\frac1n \right) f(0) + \frac{f(n)}{n}.
\end{align}
Furthermore, if $f(0) < \infty$, $f$ is differentiable on $(0,n)$, and
$ K_n := \underset{t \in (0,n)}{\sup} \, \bigl| f'(t) \bigr| < \infty$,
then, for every $\rho \geq 1$,
\begin{align}
\label{convergence rate 2}
0 \leq \lim_{\rho' \to \infty} \bigl\{ u_f(n, \rho') \bigr\} - u_f(n, \rho)
\leq \frac{2K_n \; (n-1)}{n+\rho-1}.
\end{align}

\item \label{Thm. 5-h}
For $\rho \geq 1$, let the function $f$ be also twice differentiable, and
let $M$ and $m$ be constants such that the following condition holds:
\begin{align}
\label{f'' bounded}
0 \leq m \leq f''(t) \leq M, \quad \forall \,
t \in \bigl[ \tfrac1\rho, \rho \bigr].
\end{align}
Then, for all $Q \in \set{P}_n(\rho)$,
\begin{align}
0 &\leq \tfrac12 m \bigl( n \| Q \|_2^2 - 1 \bigr) \\
\label{LB Df-m}
& \leq D_f(Q \| U_n) \\
\label{UB0 Df-M}
&\leq \tfrac12 M \bigl( n \| Q \|_2^2 - 1 \bigr) \\
\label{UB Df-M}
& \leq \frac{M (\rho-1)^2}{8 \rho}
\end{align}
with equalities in \eqref{LB Df-m} and \eqref{UB0 Df-M} for the
$\chi^2$ divergence (with $M=m=2$).

\item \label{Thm. 5-i}
Let $d>0$. If $f''(t) \leq M_f \in (0, \infty)$ for all $t>0$, then
$D_f(Q \| U_n) \leq d$ for all $Q \in \set{P}_n(\rho)$, if
\begin{align}
\label{UB rho}
\rho \leq 1 + \frac{4 d}{M_f} + \sqrt{\frac{8d}{M_f}
+ \frac{16 d^2}{M_f^2}}.
\end{align}
\end{enumerate}
\end{theorem}

\begin{proof}
See Appendix~\ref{appendix: LB/UB f-div}.
\end{proof}

Tsallis entropy was introduced in \cite{Tsallis88} as a generalization of the
Shannon entropy (similarly to the R\'{e}nyi entropy \cite{Renyientropy}), and
it was applied to statistical physics in \cite{Tsallis88}.

\begin{definition} \cite{Tsallis88}
\label{definition: Tsallis entropy}
Let $P_X$  be a probability mass function defined on a discrete set $\set{X}$.
The \textit{Tsallis entropy of order} $\alpha \in (0,1) \cup (1, \infty)$
of $X$, denoted by $S_{\alpha}(X)$ or $S_{\alpha}(P_X)$, is defined as
\begin{align}
\label{eq: Tsallis entropy}
S_{\alpha}(X) & = \frac1{1-\alpha}
\left( \, \sum_{x \in \set{X}} P_X^{\alpha}(x) - 1 \right) \\
\label{eq 2: Tsallis entropy}
&= \frac{\|P_X\|_{\alpha}^{\alpha}-1}{1-\alpha},
\end{align}
where $\|P_X\|_{\alpha} := \left( \, \underset{x \in \set{X}}{\sum}
P_X^{\alpha}(x) \right)^{\frac1\alpha}$.
The Tsallis entropy is continuously
extended at orders $0$, $1$, and $\infty$; at order~1, it
coincides with the Shannon entropy on base $\mathrm{e}$
(expressed in nats).
\end{definition}

Theorem~\ref{thm: majorization Df} enables to strengthen the Schur-concavity
property of the Tsallis entropy (see \cite[Theorem~13.F.3.a.]{MarshallOA})
as follows.
\begin{theorem} \label{thm: bounds Tsallis}
Let $P$ and $Q$ be probability mass functions which are supported on a finite
set, and let $P \prec Q$. Then, for all $\alpha > 0$,
\begin{enumerate}[a)]
\item \label{Thm. 6.a}
\begin{equation}
\vspace*{-0.3cm}
\label{bounds Tsallis}
0 \leq L(\alpha, P, Q) \leq S_\alpha(P) - S_\alpha(Q) \leq U(\alpha, P, Q),
\end{equation}
where
\begin{align}
\label{LB Tsallis}
& L(\alpha, P, Q) :=
\begin{dcases}
\tfrac12 \, \alpha q_{\max}^{\alpha-2} \, \bigl( \| Q \|_2^2 - \| P \|_2^2 \bigr),
& \quad \mbox{if $\alpha \in (0,2]$,} \\
\tfrac12 \, \alpha q_{\min}^{\alpha-2} \, \bigl( \| Q \|_2^2 - \| P \|_2^2 \bigr),
& \quad \mbox{if $\alpha \in (2, \infty)$,}
\end{dcases} \\[0.2cm]
\label{UB Tsallis}
& U(\alpha, P, Q) :=
\begin{dcases}
\tfrac12 \, \alpha q_{\min}^{\alpha-2} \, \bigl( \| Q \|_2^2 - \| P \|_2^2 \bigr),
& \quad \mbox{if $\alpha \in (0,2]$,} \\
\tfrac12 \, \alpha q_{\max}^{\alpha-2} \, \bigl( \| Q \|_2^2 - \| P \|_2^2 \bigr),
& \quad \mbox{if $\alpha \in (2, \infty)$,}
\end{dcases}
\end{align}
and the bounds in \eqref{LB Tsallis} and \eqref{UB Tsallis} are attained at $\alpha=2$.

\item \label{Thm. 6.b}
\begin{equation}
\label{inf,sup}
\inf_{P \prec Q, \, P \neq Q} \frac{S_\alpha(P) - S_\alpha(Q)}{L(\alpha, P, Q)} =
\sup_{P \prec Q, \, P \neq Q} \frac{S_\alpha(P) - S_\alpha(Q)}{U(\alpha, P, Q)} = 1,
\end{equation}
where the infimum and supremum in \eqref{inf,sup} can be restricted to probability mass
functions $P$ and $Q$ which are supported on a binary alphabet.
\end{enumerate}
\end{theorem}

\begin{proof}
See Appendix~\ref{appendix: bounds Tsallis}.
\end{proof}

\begin{remark}
\label{remark: Tsallis}
The lower bound in \cite[Theorem~1]{HoS-ISIT15} also strengthens
the Schur-concavity property of the Tsallis entropy. It can be
verified that none of the lower bounds in \cite[Theorem~1]{HoS-ISIT15}
and Theorem~\ref{thm: bounds Tsallis} supersedes the other. For
example, let $\alpha>0$, and let $P_\varepsilon$ and $Q_\varepsilon$
be probability mass functions supported on $\set{A} := \{0,1\}$ with
$P_\varepsilon(0) = \tfrac12 + \varepsilon$ and
$Q_\varepsilon(0) = \tfrac12 + \beta \varepsilon$
where $\beta > 1$ and $0 < \varepsilon < \frac1{2\beta}$. This yields
$P_\varepsilon \prec Q_\varepsilon$. From \eqref{limit 1 - Tsallis}
(see Appendix~\ref{appendix: bounds Tsallis}),
\begin{align}
\label{01072019a2}
& \lim_{\varepsilon \to 0^+} \frac{S_\alpha(P_\varepsilon) -
S_\alpha(Q_\varepsilon)}{L(\alpha, P_\varepsilon, Q_\varepsilon)} = 1.
\end{align}
If $\alpha=1$, then $S_1(P_\varepsilon) - S_1(Q_\varepsilon) =
\tfrac1{\log \mathrm{e}} \bigl(H(P_\varepsilon) - H(Q_\varepsilon)\bigr)$,
and the continuous extension of the lower bound in \cite[Theorem~1]{HoS-ISIT15}
at $\alpha=1$ is specialized to the earlier result by the same authors
in \cite[Theorem~3]{HoS-IT10}; it states that if $P \prec Q$, then
$H(P) - H(Q) \geq D(Q\|P)$.
In contrast to \eqref{01072019a2}, it can be verified that
\begin{align}
\label{01072019a3}
& \lim_{\varepsilon \to 0^+} \frac{S_1(P_\varepsilon) -
S_1(Q_\varepsilon)}{\tfrac1{\log \mathrm{e}} \,
D(Q_\varepsilon\|P_\varepsilon)}
= \frac{\beta+1}{\beta-1} > 1, \quad \forall \, \beta > 1,
\end{align}
which can be made arbitrarily large by selecting $\beta$ to be
sufficiently close to~1 (from above). This provides a case
where the lower bound in Theorem~\ref{thm: bounds Tsallis}
outperforms the one in \cite[Theorem~3]{HoS-IT10}.
\end{remark}

\begin{remark}
Due to the one-to-one correspondence between Tsallis and R\'{e}nyi entropies
of the same positive order, similar to the transition from \cite[Theorem~1]{HoS-ISIT15}
to \cite[Theorem~2]{HoS-ISIT15}, also Theorem~\ref{thm: bounds Tsallis} enables to
strengthen the Schur-concavity property of the R\'{e}nyi entropy.
For information-theoretic implications of the Schur-concavity of the R\'{e}nyi
entropy, the reader is referred to, e.g.,
\cite[Theorem~3]{CicaleseGV19}, \cite[Theorem~11]{ISSV18} and \cite{Sason18b}.
\end{remark}

\section{Illustration of the Main Results and Implications}
\label{section: Examples}

\subsection{Illustration of Theorems~\ref{Thm: DMS-DMC} and~\ref{theorem: DMS/DMC - ver2}}
\label{subsection: Illustration of Thm. DMS-DMC}

We apply here the data-processing inequalities in Theorems~\ref{Thm: DMS-DMC}
and~\ref{theorem: DMS/DMC - ver2} to the new class of $f$-divergences introduced
in Theorem~\ref{thm: f_alpha-divergence}.

In the setup of Theorems~\ref{Thm: DMS-DMC} and~\ref{theorem: DMS/DMC - ver2},
consider communication over a time-varying binary-symmetric channel (BSC).
Consequently, let $\set{X} = \set{Y} = \{0,1\}$, and let
\begin{align} \label{Bernoulli two inputs}
& P_{X_i}(1) = p_i, \quad Q_{X_i}(1) = q_i,
\end{align}
with $p_i \in (0,1)$ and $q_i \in (0,1)$ for every $i \in \{1, \ldots, n\}$.
Let the transition probabilities $P_{Y_i | X_i}(\cdot | \cdot)$ correspond to
$\text{BSC}(\delta_i)$ (i.e., a BSC with a crossover probability $\delta_i$), i.e.,
\begin{align} \label{TV-BSC}
P_{Y_i | X_i}(y | x) =
\begin{dcases}
1-\delta_i  & \quad \mbox{if $x=y$,} \\
\delta_i    & \quad \mbox{if $x \neq y$.}
\end{dcases}
\end{align}
For all $\lambda \in [0,1]$ and $\underline{x} \in \set{X}^n$, the probability
mass function at the channel input is given by
\begin{align}
& R_{X^n}^{(\lambda)}(\underline{x}) =
\prod_{i=1}^n  R_{X_i}^{(\lambda)}(x_i),
\end{align}
with
\begin{align}
\label{Bernoulli input}
R_{X_i}^{(\lambda)}(x) = \lambda P_{X_i}(x) + (1-\lambda) Q_{X_i}(x), \quad x \in \{0,1\},
\end{align}
where the probability mass function in \eqref{Bernoulli input} refers to a Bernoulli
distribution with parameter $\lambda p_i + (1-\lambda) q_i$.
At the output of the time-varying BSC (see \eqref{MC1 in DMC}--\eqref{MC3 in DMC} and
\eqref{TV-BSC}), for all $\underline{y} \in \set{Y}^n$,
\begin{align}
& R_{Y^n}^{(\lambda)}(\underline{y}) =
\prod_{i=1}^n  R_{Y_i}^{(\lambda)}(y_i), \quad
P_{Y^n}(\underline{y}) = \prod_{i=1}^n P_{Y_i}(y_i), \quad
Q_{Y^n}(\underline{y}) = \prod_{i=1}^n Q_{Y_i}(y_i),
\end{align}
where
\begin{align}
& R_{Y_i}^{(\lambda)}(1) = \bigl( \lambda p_i + (1-\lambda) q_i \bigr) \ast \delta_i, \\
& P_{Y_i}(1) = p_i \ast \delta_i, \\
& Q_{Y_i}(1) = q_i \ast \delta_i,
\end{align}
with
\begin{align}
a \ast b := a (1-b) + (1-a) b, \quad 0 \leq a, b \leq 1.
\end{align}

The $\chi^2$-divergence from $\text{Bernoulli}(p)$ to $\text{Bernoulli}(q)$ is given by
\begin{align}
\chi^2\bigl(\text{Bernoulli}(p) \, \| \, \text{Bernoulli}(q)\bigr)
= \frac{(p-q)^2}{q(1-q)},
\end{align}
and since the probability mass functions $P_{X_i}$, $Q_{X_i}$, $P_{Y_i}$ and $Q_{Y_i}$
correspond to Bernoulli distributions with parameters $p_i$, $q_i$, $p_i \ast \delta_i$
and $q_i \ast \delta_i$, respectively, Theorem~\ref{Thm: DMS-DMC} gives that
\begin{align}
& \hspace*{-0.3cm} c_{f_\alpha}\bigl(\xi_1(n, \lambda), \, \xi_2(n, \lambda) \bigr) \,
\left[ \, \prod_{i=1}^n \left(1 + \frac{\lambda^2 (p_i-q_i)^2}{q_i (1-q_i)} \right)
- \prod_{i=1}^n \left(1 + \frac{\lambda^2 (p_i \ast \delta_i
- q_i \ast \delta_i)^2}{(q_i \ast \delta_i) (1-q_i \ast \delta_i)} \right) \right]
\nonumber \\[0.1cm]
& \hspace*{-0.3cm} \leq D_{f_\alpha}(R_{X^n}^{(\lambda)} \, \| \, Q_{X^n})
- D_{f_\alpha}(R_{Y^n}^{(\lambda)} \, \| \, Q_{Y^n}) \label{diff f_alpha} \\[0.1cm]
\label{UB - diff f_alpha}
& \hspace*{-0.3cm} \leq e_{f_\alpha}\bigl(\xi_1(n, \lambda), \, \xi_2(n, \lambda) \bigr)
\left[ \, \prod_{i=1}^n \left(1 + \frac{\lambda^2 (p_i-q_i)^2}{q_i (1-q_i)} \right)
- \prod_{i=1}^n \left(1 + \frac{\lambda^2 (p_i \ast \delta_i
- q_i \ast \delta_i)^2}{(q_i \ast \delta_i) (1-q_i \ast \delta_i)} \right) \right]
\end{align}
for all $\lambda \in [0,1]$ and $n \in \naturals$. From \eqref{c_f}, \eqref{e_f}
and \eqref{f_alpha}, we get that for all $\xi_1 < 1< \xi_2$,
\begin{align}
c_{f_\alpha}(\xi_1, \xi_2) &= \tfrac12 \, \inf_{t \in [\xi_1, \xi_2]} f_\alpha''(t) \\
&= \log(\alpha + \xi_1) + \tfrac32 \, \log \mathrm{e}, \\[0.2cm]
e_{f_\alpha}(\xi_1, \xi_2) &= \tfrac12 \, \sup_{t \in [\xi_1, \xi_2]} f_\alpha''(t) \\
&= \log(\alpha + \xi_2) + \tfrac32 \, \log \mathrm{e},
\end{align}
and, from \eqref{xi1_n}, \eqref{xi2_n} and \eqref{Bernoulli two inputs}, for all
$\lambda \in (0,1]$,
\begin{align}
\label{xi_1,n-BSC}
& \xi_1(n, \lambda) := \prod_{i=1}^n \left(1 - \lambda + \lambda \, \min \left\{
\frac{p_i}{q_i}, \frac{1-p_i}{1-q_i} \right\} \right) \in [0,1), \\[0.1cm]
\label{xi_2,n-BSC}
& \xi_2(n, \lambda) := \prod_{i=1}^n \left(1 - \lambda + \lambda \, \max \left\{
\frac{p_i}{q_i}, \frac{1-p_i}{1-q_i} \right\} \right) \in (1, \infty),
\end{align}
provided that $p_i \neq q_i$ for some $i \in \{1, \ldots, n\}$ (otherwise, both
$f$-divergences in the right side of \eqref{diff f_alpha} are equal to zero since
$P_{X_i} \equiv Q_{X_i}$ and therefore $R_{X_i}^{(\lambda)} \equiv Q_{X_i}$ for
all $i$ and $\lambda \in [0,1]$). Furthermore, from Item~\ref{Th. 2.c.}) of
Theorem~\ref{Thm: DMS-DMC}, for every $n \in \naturals$ and
$\alpha \geq \mathrm{e}^{-\frac32}$,
\begin{align}
\label{limit f_alpha div.}
& \lim_{\lambda \to 0^+} \frac{D_{f_\alpha}(R_{X^n}^{(\lambda)} \, \| \, Q_{X^n})
- D_{f_\alpha}(R_{Y^n}^{(\lambda)} \, \| \, Q_{Y^n})}{\lambda^2} \nonumber \\
& = \left( \log(\alpha+1) + \tfrac32 \, \log \mathrm{e} \right)
\sum_{i=1}^n \left\{ \frac{(p_i-q_i)^2}{q_i (1-q_i)} -
\frac{(p_i \ast \delta_i - q_i \ast \delta_i)^2}{(q_i \ast \delta_i)
(1-q_i \ast \delta_i)} \right\},
\end{align}
and the lower and upper bounds in the left side of \eqref{diff f_alpha}
and the right side of \eqref{UB - diff f_alpha}, respectively, are tight
as we let $\lambda \to 0$, and they both coincide with the limit in the
right side of \eqref{limit f_alpha div.}.

Figure~\ref{figure: 20190512} illustrates the upper and lower bounds
in \eqref{diff f_alpha} and \eqref{UB - diff f_alpha} with $\alpha=1$,
$p_i \equiv \tfrac14$, $q_i \equiv \tfrac12$ and $\delta_i \equiv 0.110$
for all $i$, and $n \in \{1, 10, 50\}$.
In the special case where $\{\delta_i\}$ are fixed for all $i$, the
communication channel is a time-invariant BSC whose capacity is equal
to $\tfrac12$ bit per channel use.
\begin{figure}
\vspace*{-4cm}
\begin{center}
\hspace*{-0.8cm}
\includegraphics[width=9.2cm, angle=0]{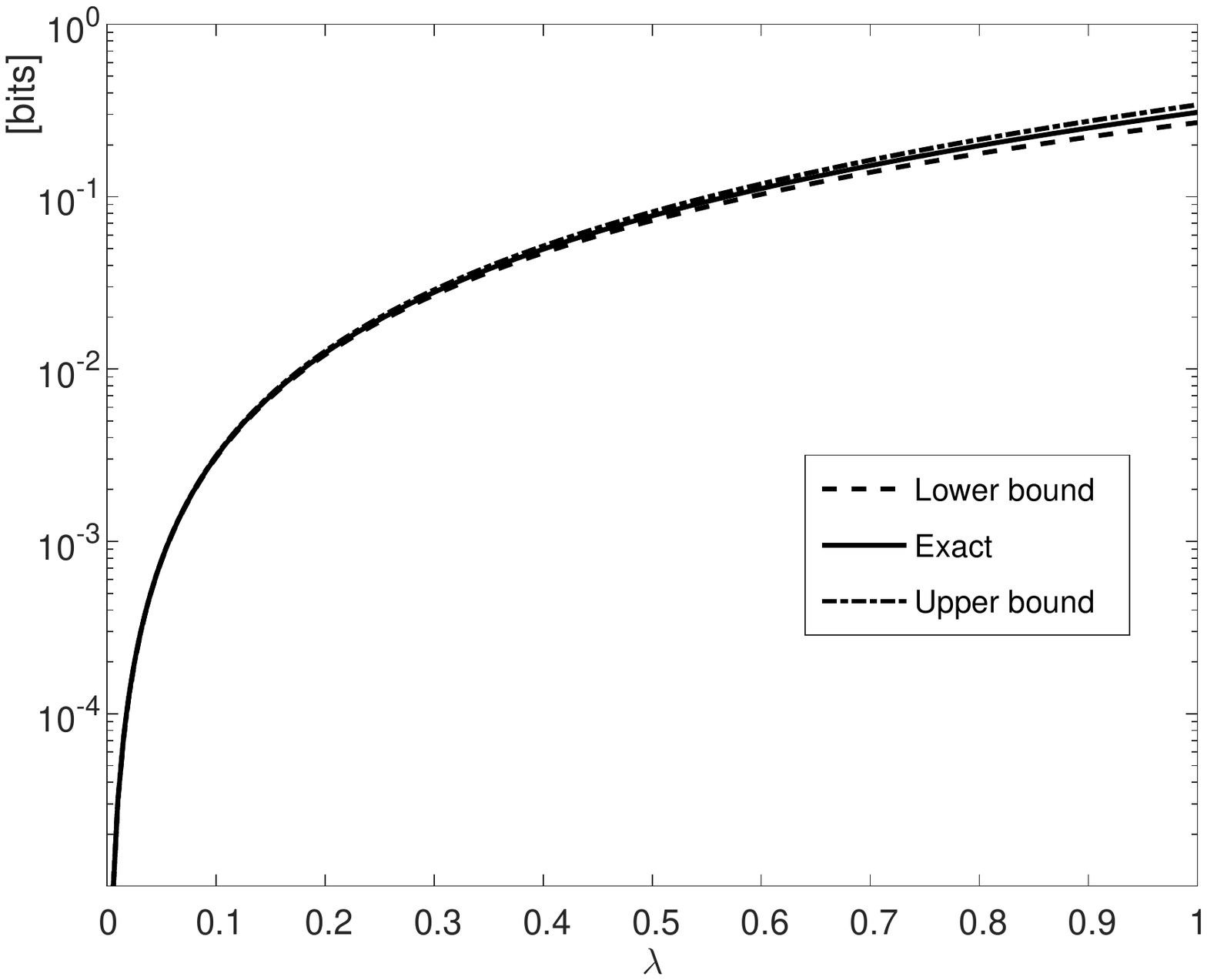}  \\[-6.3cm]
\hspace*{-0.8cm}
\includegraphics[width=9.2cm, angle=0]{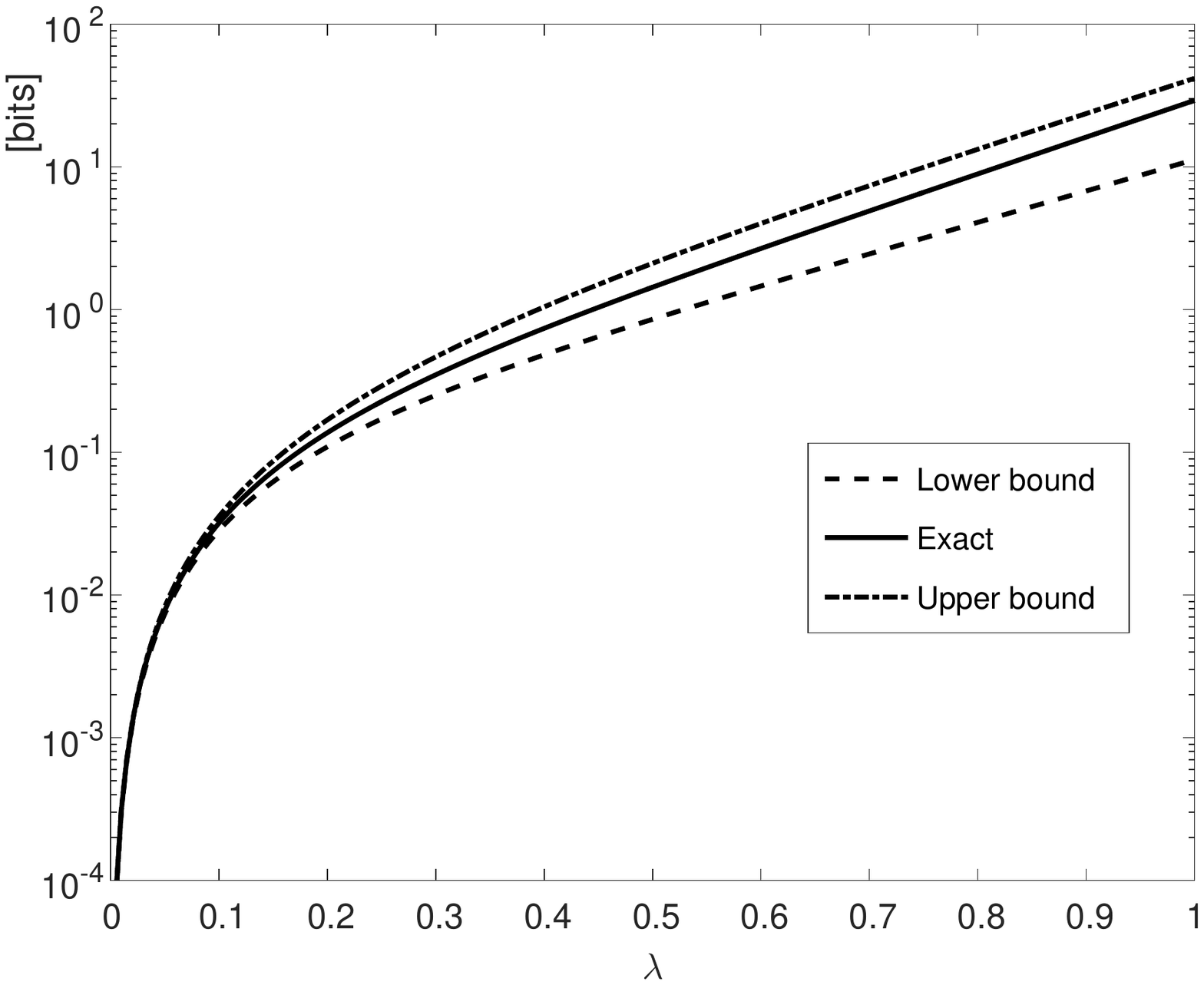} \\[-6.3cm]
\hspace*{-0.8cm}
\includegraphics[width=9.2cm]{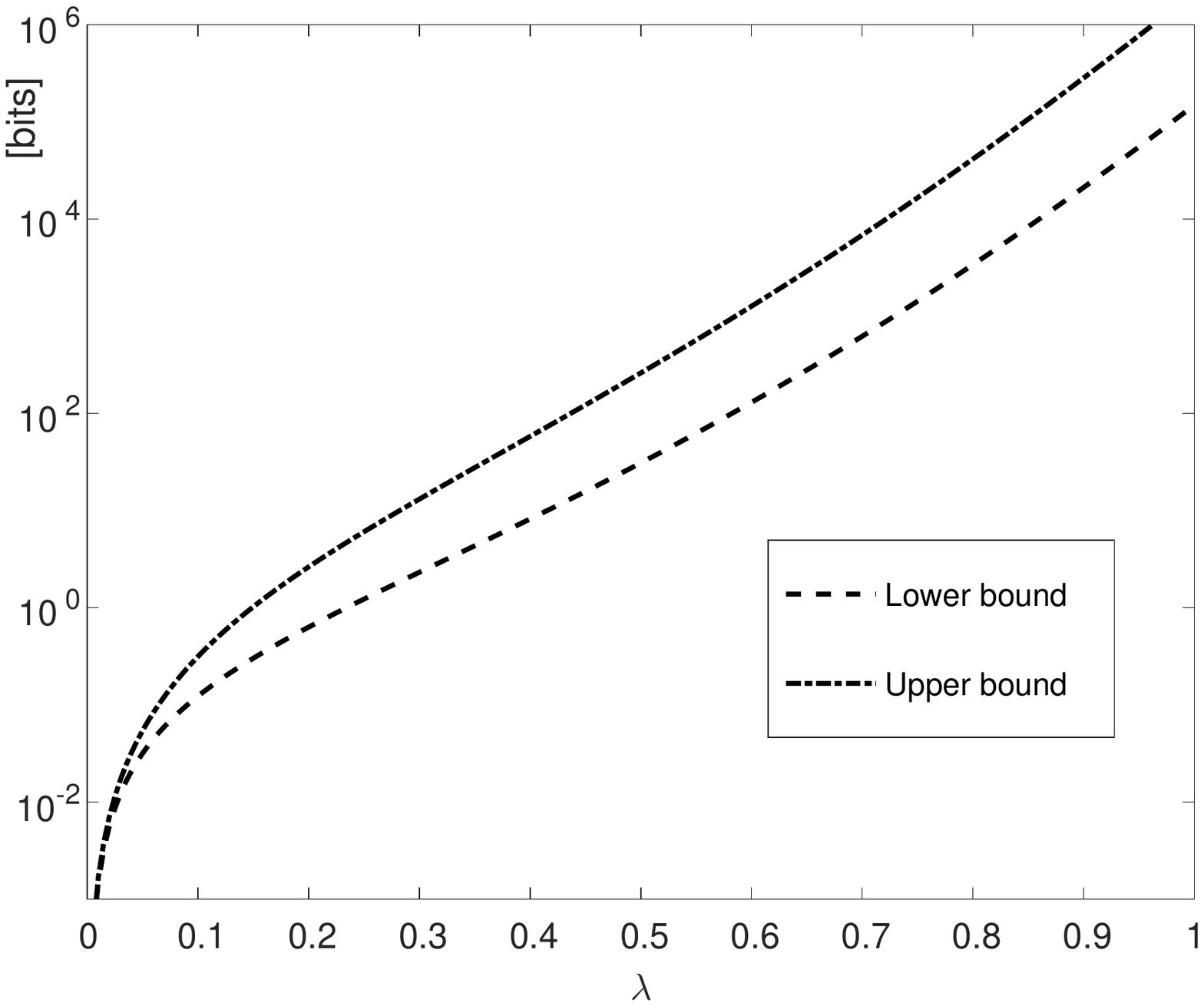}
\end{center}
\vspace*{-3.3cm}
\caption{\label{figure: 20190512} The bounds in
Theorem~\ref{Thm: DMS-DMC} applied to
$D_{f_\alpha}(R_{X^n}^{(\lambda)} \, \| \, Q_{X^n})
- D_{f_\alpha}(R_{Y^n}^{(\lambda)} \, \| \, Q_{Y^n})$ (vertical axis)
versus $\lambda \in [0,1]$ (horizontal axis). The $f_\alpha$-divergence
refers to Theorem~\ref{thm: f_alpha-divergence}.
The probability mass functions $P_{X^n}$ and $Q_{X^n}$ correspond,
respectively, to discrete memoryless sources emitting $n$ i.i.d.
$\text{Bernoulli}(p)$ and $\text{Bernoulli}(q)$ symbols; the symbols
are transmitted over $\text{BSC}(\delta)$ with
$(\alpha, p, q, \delta) = \bigl(1, \tfrac14, \tfrac12, 0.110\bigr)$.
The upper, middle and lower plots correspond, respectively, to $n=1$,
$10$, and $50$. The bounds in the upper and middle plots are compared
to the exact values.}
\end{figure}
By referring to the upper and middle plots of Figure~\ref{figure: 20190512},
if $n=1$ or $n=10$, then the exact values of the differences of
the $f_\alpha$-divergences in the right side of \eqref{diff f_alpha} are
calculated numerically, being compared to the lower and upper bounds
in the left side of \eqref{diff f_alpha} and the right side of
\eqref{UB - diff f_alpha} respectively.
Since the $f_\alpha$-divergence does not tensorize, the
computation of the exact value of each of the two $f_\alpha$-divergences
in the right side of \eqref{diff f_alpha} involves a pre-computation of
$2^n$ probabilities for each of the probability mass functions $P_{X^n}$,
$Q_{X^n}$, $P_{Y^n}$ and $Q_{Y^n}$; this computation is prohibitively
complex unless $n$ is small enough.

We now apply the bound in Theorem~\ref{theorem: DMS/DMC - ver2}.
In view of \eqref{def: kappa}, \eqref{13062019d1}, \eqref{f_alpha} and
\eqref{kappa_alpha 2}, for
all $\lambda \in (0,1]$ and $\alpha \geq \mathrm{e}^{-\frac32}$,
\begin{align}
\label{15062019a1}
& \frac{D_{f_\alpha}\bigl(R_{Y^n}^{(\lambda)} \, \| \,
Q_{Y^n}\bigr)}{D_{f_\alpha} \bigl(R_{X^n}^{(\lambda)} \, \| \,
Q_{X^n}\bigr)} \nonumber \\
& \leq \frac{\kappa_\alpha\bigl(\xi_1(n, \lambda), \,
\xi_2(n, \lambda) \bigr)}{f_\alpha(0)+f_\alpha'(1)} \;
\frac{\overset{n}{\underset{i=1}{\prod}} \Bigl( 1+ \lambda^2 \,
\chi^2(P_{Y_i} \, \| \, Q_{Y_i} \bigr) \Bigr)
- 1}{\overset{n}{\underset{i=1}{\prod}}
\Bigl( 1+ \lambda^2 \, \chi^2(P_{X_i} \, \| \,
Q_{X_i} \bigr) \Bigr) - 1} \\[0.1cm]
\label{15062019a2}
&= \frac{f_\alpha\bigl(\xi_2(n, \lambda)\bigr) + f_\alpha'(1) \,
\bigl(1-\xi_2(n, \lambda)\bigr)}{\bigl(\xi_2(n, \lambda)-1\bigr)^2 \,
\bigl(f_\alpha(0)+f_\alpha'(1)\bigr)} \cdot
\frac{\overset{n}{\underset{i=1}{\prod}} \left(1 + \dfrac{\lambda^2
(p_i \ast \delta_i - q_i \ast \delta_i)^2}{(q_i \ast \delta_i)
(1-q_i \ast \delta_i)} \right) - 1}{\overset{n}{\underset{i=1}{\prod}}
\left(1 + \dfrac{\lambda^2 (p_i-q_i)^2}{q_i (1-q_i)} \right) - 1},
\end{align}
where $\xi_1(n, \lambda) \in [0,1)$ and $\xi_2(n, \lambda)
\in (1, \infty)$ are given in \eqref{xi_1,n-BSC} and
\eqref{xi_2,n-BSC}, respectively, and for $t \geq 0$,
\begin{align}
& f_\alpha(t) + f_\alpha'(1) (1-t) \nonumber \\
&= (\alpha+t)^2 \log(\alpha+t) - (\alpha+1)^2 \log(\alpha+1) \nonumber \\
& \hspace*{0.4cm} + \bigl[2(\alpha+1) \log(\alpha+1)
+ (\alpha+1) \log \mathrm{e} \bigr] (1-t).  \label{15062019a3}
\end{align}

Figure~\ref{figure: 20190616} illustrates the upper bound on
$\frac{D_{f_\alpha}(R_{Y^n}^{(\lambda)} \, \| \,
Q_{Y^n})}{D_{f_\alpha}(R_{X^n}^{(\lambda)} \, \| \, Q_{X^n})}$
(see \eqref{15062019a1}--\eqref{15062019a3})
as a function of $\lambda \in (0,1]$. It refers to the case where
$p_i \equiv \tfrac14$, $q_i \equiv \tfrac12$, and $\delta_i \equiv
0.110$ for all $i$ (similarly to Figure~\ref{figure: 20190512}).
The upper and middle plots correspond to
$n=10$ with $\alpha = 10$ and $\alpha = 100$, respectively; the middle
and lower plots correspond to $\alpha=100$ with $n=10$ and $n=100$,
respectively. The bounds in the upper and middle plots are compared to
their exact values since their numerical computations are feasible for
$n=10$. It is observed from the numerical comparisons for $n=10$ (see
the upper and middle plots in Figure~\ref{figure: 20190616}) that the
upper bounds are informative, especially for large values of $\alpha$
where the $f_\alpha$-divergence becomes closer to a scaled version
of the $\chi^2$-divergence (see Item~\ref{Thm. f-2c}) in
Theorem~\ref{thm: f_alpha-divergence}).

\begin{figure}
\vspace*{-4.1cm}
\begin{center}
\hspace*{-0.8cm}
\includegraphics[width=9.2cm, angle=0]{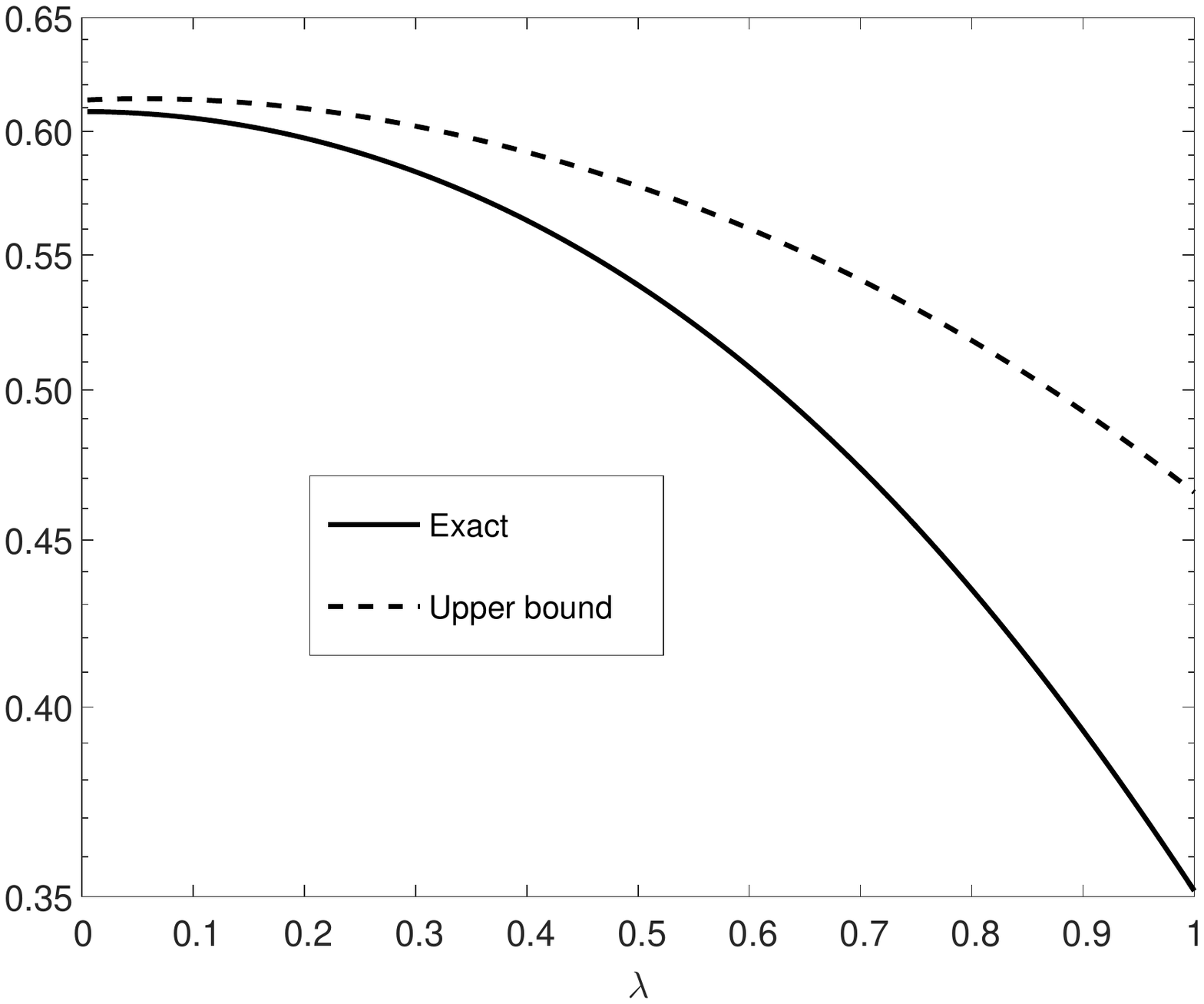}  \\[-6.4cm]
\hspace*{-0.8cm}
\includegraphics[width=9.2cm, angle=0]{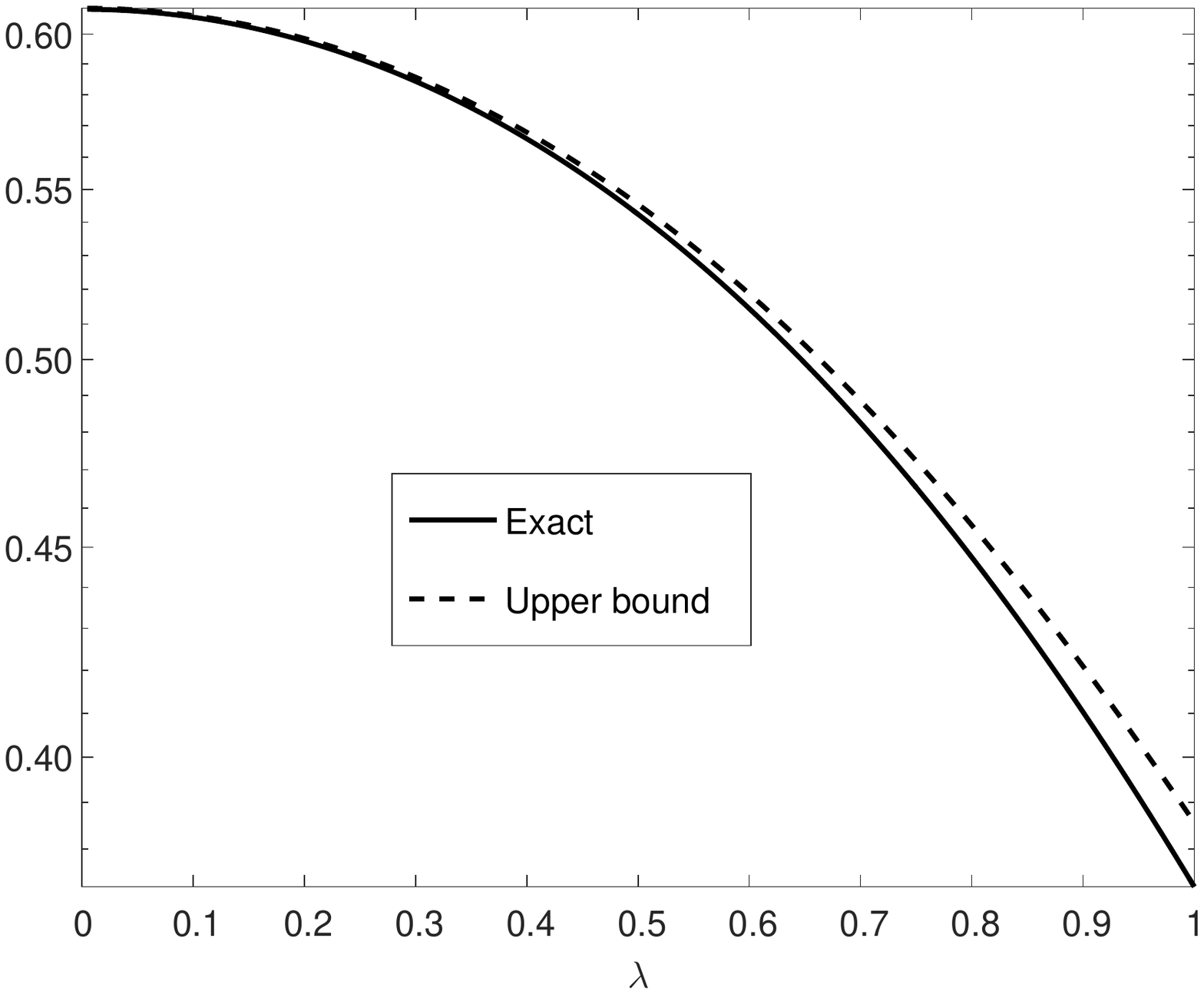} \\[-6.4cm]
\hspace*{-0.8cm}
\includegraphics[width=9.2cm]{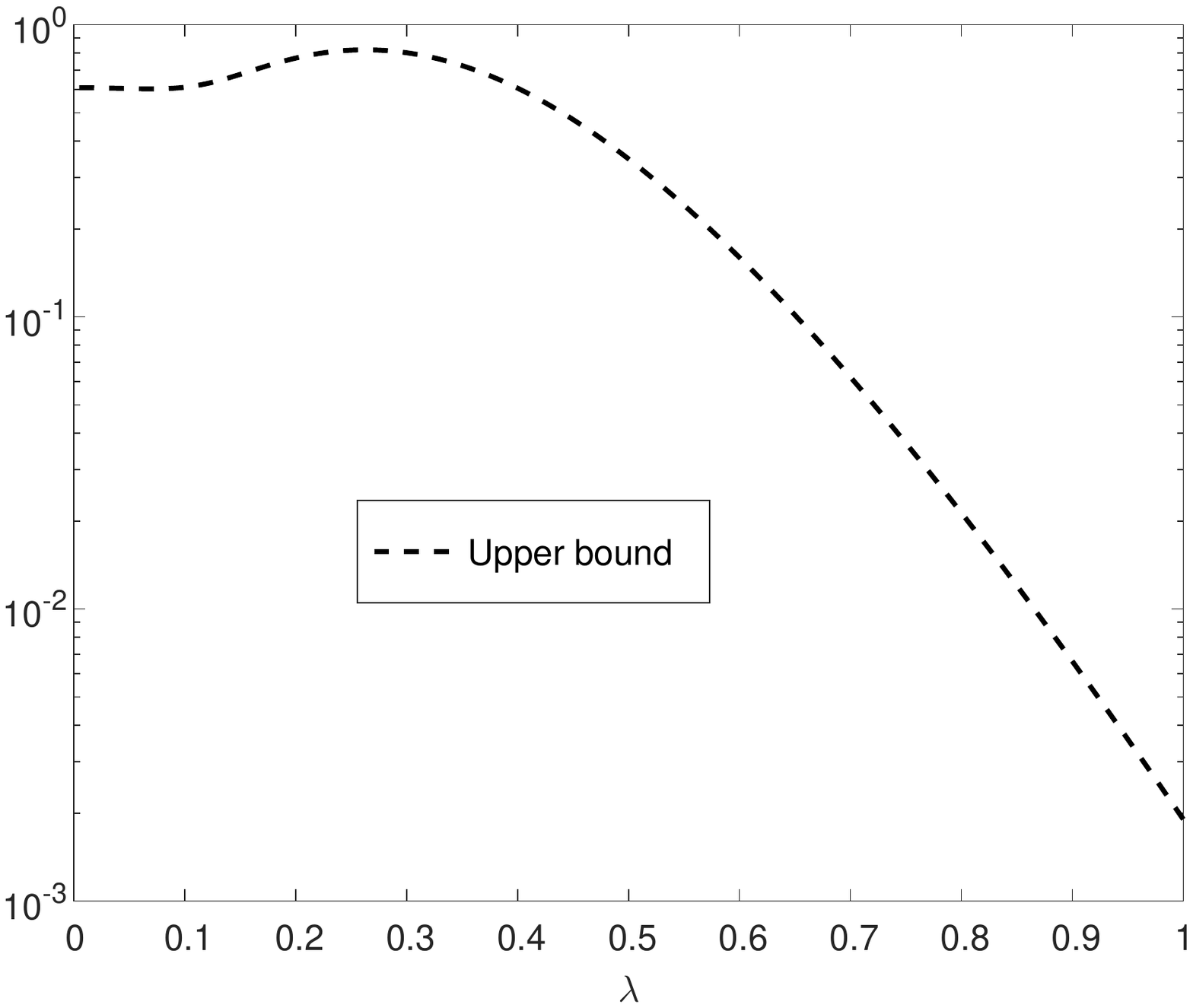}
\end{center}
\vspace*{-3.3cm}
\caption{\label{figure: 20190616} The upper bound on
$\frac{D_{f_\alpha}(R_{Y^n}^{(\lambda)} \, \| \, Q_{Y^n})}{D_{f_\alpha}
(R_{X^n}^{(\lambda)} \, \| \, Q_{X^n})}$ (see
\eqref{15062019a1}--\eqref{15062019a3}),
for the $f_\alpha$-divergence in Theorem~\ref{thm: f_alpha-divergence},
is shown in the vertical axis versus $\lambda \in [0,1]$ in the horizontal
axis. The probability mass functions $P_{X_i}$ and $Q_{X_i}$ are
$\text{Bernoulli}(p)$ and $\text{Bernoulli}(q)$, respectively, for all
$i \in \{1, \ldots, n\}$ with $n$ uses of $\text{BSC}(\delta)$ and
parameters $(p, q, \delta) = \bigl(\tfrac14, \tfrac12, 0.110\bigr)$.
The upper and middle plots correspond to
$n=10$ with $\alpha = 10$ and $\alpha = 100$, respectively; the middle
and lower plots correspond to $\alpha=100$ with $n=10$ and $n=100$,
respectively. The bounds in the upper and middle plots are compared to
their respective exact values, being computationally feasible for $n=10$.}
\end{figure}

\subsection{Illustration of Theorems~\ref{theorem: contraction coef}
and~\ref{thm: f_alpha-divergence}}
\label{subsection: Illustration of Thm. f_alpha-divergence}

Following the application of the data-processing inequalities
in Theorems~\ref{Thm: DMS-DMC} and \ref{theorem: DMS/DMC - ver2}
to a class of $f$-divergences (see
Section~\ref{subsection: Illustration of Thm. DMS-DMC}),
some interesting properties of this class are introduced in
Theorem~\ref{thm: f_alpha-divergence}.

For $\alpha \geq \mathrm{e}^{-\frac32}$, let
$d_{f_\alpha} \colon (0,1)^2 \to [0, \infty)$ be the binary
$f_\alpha$-divergence (see \eqref{f_alpha}), defined as
\begin{align}
d_{f_\alpha}(p\|q) &:= D_{f_\alpha}\bigl( \, \text{Bernoulli}(p)
\, \| \, \text{Bernoulli}(q) \, \bigr) \\
& \hspace*{0.1cm} = q \left(\alpha + \frac{p}{q} \right)^2
\log\left(\alpha + \frac{p}{q} \right) +
(1-q) \left(\alpha + \frac{1-p}{1-q} \right)^2 \log\left(\alpha
+ \frac{1-p}{1-q} \right) \nonumber \\
& \hspace*{0.5cm} - (\alpha+1)^2 \log(\alpha+1),
\quad \forall \, (p,q) \in (0,1)^2.
\end{align}

\begin{figure}[ht]
\vspace*{-4.4cm}
\begin{center}
\centerline{\includegraphics[width=9.5cm]{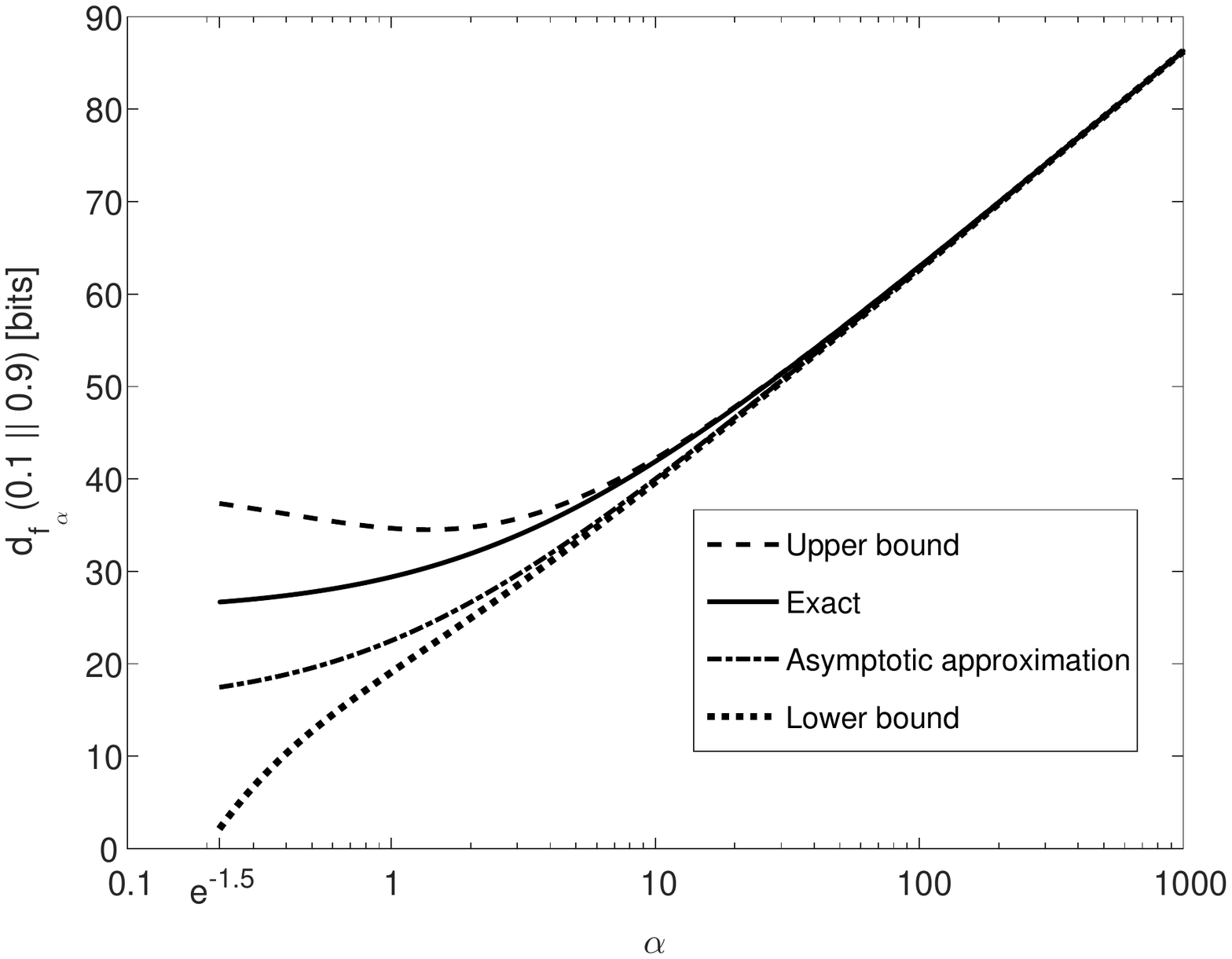}}
\vspace*{-6.4cm}
\centerline{\includegraphics[width=9.5cm]{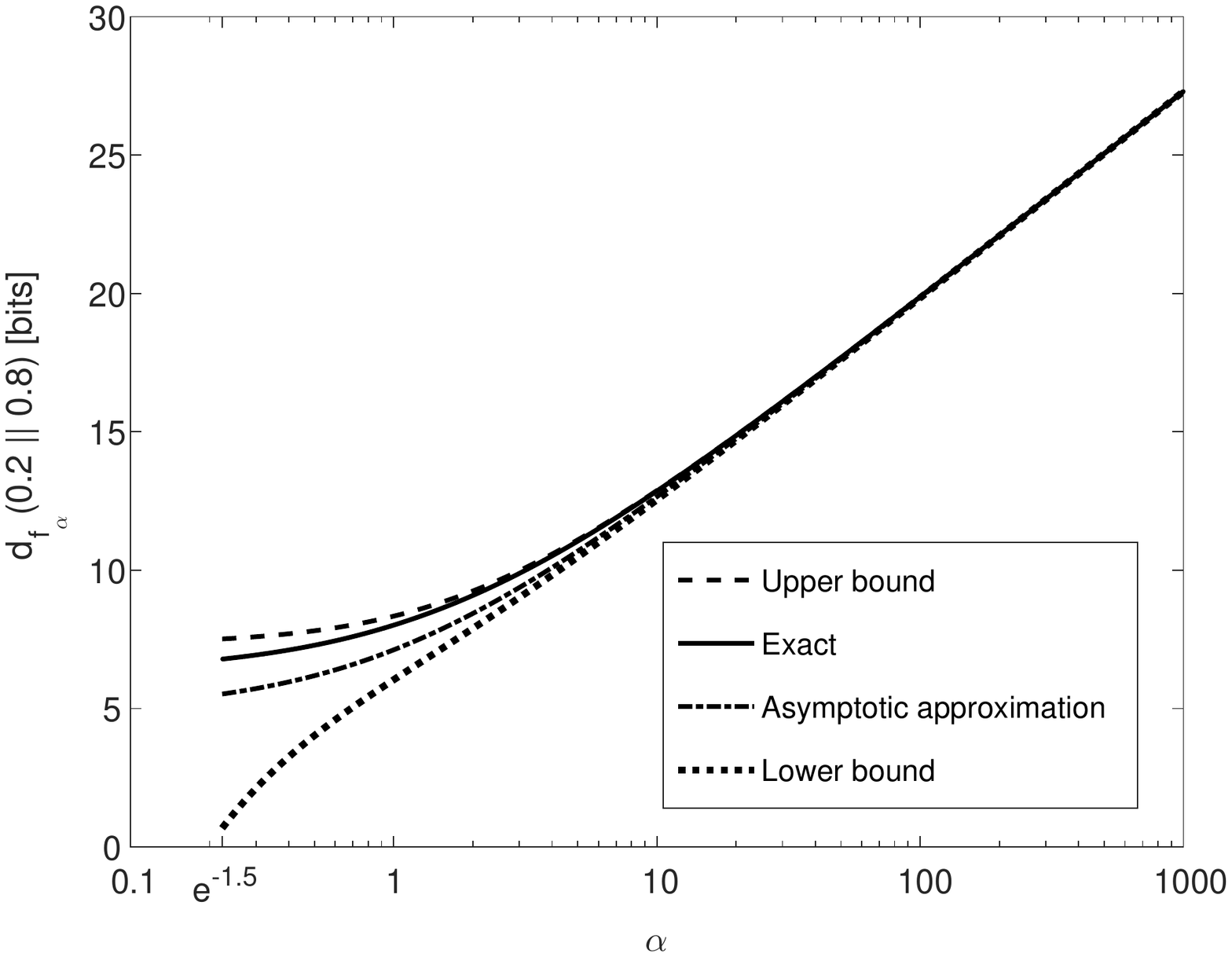}}
\end{center}
\vspace*{-3.7cm}
\caption{\label{figure: f_alpha div}
Plots of $d_{f_\alpha}(p \| q)$, its upper and lower bounds
in \eqref{Df-chi2} and \eqref{Df_UB}, respectively, and its
asymptotic approximation in \eqref{asymp.} for large values
of $\alpha$. The plots are shown as a function of
$\alpha \in \bigl[\mathrm{e}^{-\frac32}, 1000 \bigr]$.
The upper and lower plots refer, respectively, to
$(p,q) = (0.1, 0.9)$ and $(p,q) = (0.2, 0.8)$.}
\end{figure}

Theorem~\ref{thm: f_alpha-divergence} is illustrated in
Figure~\ref{figure: f_alpha div}, showing that
$d_{f_\alpha}(p \| q)$ is monotonically increasing as a
function of $\alpha \geq \mathrm{e}^{-\frac32}$ (note that
the concavity in $\alpha$ is not reflected from these plots
because the horizontal axis of $\alpha$ is in logarithmic
scaling). The binary divergence $d_{f_\alpha}(p \| q)$ is
also compared in Figure~\ref{figure: f_alpha div} with its
lower and upper bounds in \eqref{Df-chi2} and \eqref{Df_UB},
respectively, illustrating that these bounds are both
asymptotically tight for large values of $\alpha$.
The asymptotic approximation of $d_{f_\alpha}(p \| q)$ for
large $\alpha$, expressed as a function of $\alpha$ and
$\chi^2(p\|q)$ (see \eqref{asymp.}), is also depicted in
Figure~\ref{figure: f_alpha div}. The upper and lower plots
in Figure~\ref{figure: f_alpha div} refer, respectively, to
$(p,q) = (0.1, 0.9)$ and $(0.2, 0.8)$; a comparison of these
plots show a better match between the exact value of the binary
divergence, its upper and lower bounds, and its asymptotic
approximation when the values of $p$ and $q$ are getting closer.

In view of the results in \eqref{asymp.} and
\eqref{eq: limit of ratio of f-div}, it is interesting to note
that the asymptotic value of $D_{f_\alpha}(P\|Q)$ for large
values of $\alpha$ is also the exact scaling of this $f$-divergence
for {\em any finite} value of $\alpha \geq \mathrm{e}^{-\frac32}$
when the probability mass functions $P$ and $Q$ are close enough
to each other.

We next consider the ratio of the contraction coefficients
$\frac{\mu_{f_\alpha}(Q_X, W_{Y|X})}{\mu_{\chi^2}(Q_X, W_{Y|X})}$
where $Q_X$ is finitely supported on $\set{X}$ and it is not
a point mass (i.e., $|\set{X}| \geq 2$), and $W_{Y|X}$ is arbitrary.
For all $\alpha \geq \mathrm{e}^{-\frac32}$,
\begin{align}
\label{ratio contractions - 1}
1 \leq \frac{\mu_{f_\alpha}(Q_X, W_{Y|X})}{\mu_{\chi^2}(Q_X, W_{Y|X})}
\leq \frac{f_\alpha(\xi)+f'_\alpha(1) (1-\xi)}{(\xi-1)^2
\bigl(f_\alpha(0)+f'_\alpha(1)\bigr)},
\end{align}
where $f_\alpha \colon (0, \infty) \to \Reals$ is given in \eqref{f_alpha}, and
\begin{align}
\label{xi - reciprocal of min Q_X}
\xi := \frac1{\underset{x \in \set{X}}{\min} \, Q_X(x)} \in [|\set{X}|, \infty).
\end{align}
The left-side inequality in \eqref{ratio contractions - 1} is due to
\cite[Theorem~2]{PolyanskiyW17} (see Proposition~\ref{propos: mu chi^2 is minimal}),
and the right-side inequality in \eqref{ratio contractions - 1} holds due to
\eqref{13062019c2} and \eqref{kappa_alpha 2}.

Figure~\ref{figure: contraction ratio} shows the upper bound on the ratio of the
contraction coefficients $\frac{\mu_{f_\alpha}(Q_X, W_{Y|X})}{\mu_{\chi^2}(Q_X, W_{Y|X})}$,
as it is given in the right-side inequality of \eqref{ratio contractions - 1},
as a function of the parameter $\alpha \geq \mathrm{e}^{-\frac32}$. The curves
in Figure~\ref{figure: contraction ratio} correspond to different values of
$\xi \in [|\set{X}|, \infty)$, as it is given in \eqref{xi - reciprocal of min Q_X};
these upper bounds are monotonically decreasing in $\alpha$, and they
asymptotically tend to~1 as we let $\alpha \to \infty$. Hence, in view of the
left-side inequality in \eqref{ratio contractions - 1}, the upper bound on the
ratio of the contraction coefficients (in the right-side inequality) is asymptotically
tight in $\alpha$. The fact that the ratio of the contraction coefficients in the
middle of \eqref{ratio contractions - 1} tends asymptotically to~1, as $\alpha$
gets large, is not directly implied by Item~\ref{Thm. f-2c}) of
Theorem~\ref{thm: f_alpha-divergence}. The latter implies that, for fixed
probability mass functions $P$ and $Q$ and for sufficiently large $\alpha$,
\begin{align}
\label{approximate}
D_{f_\alpha}(P\|Q) \approx \bigl[ \log(\alpha+1) + \tfrac32 \log \mathrm{e} \bigr]
\, \chi^2(P\|Q);
\end{align}
however, there is no guarantee that for fixed $Q$ and sufficiently large $\alpha$,
the approximation in \eqref{approximate} holds for all $P$. By the upper bound in
the right side of \eqref{ratio contractions - 1}, it follows however that
$\mu_{f_\alpha}(Q_X, W_{Y|X})$ tends asymptotically (as we let $\alpha \to \infty$)
to the contraction coefficient of the $\chi^2$ divergence.

\begin{figure}[ht]
\vspace*{-4.3cm}
\begin{center}
\centerline{\includegraphics[width=10cm]{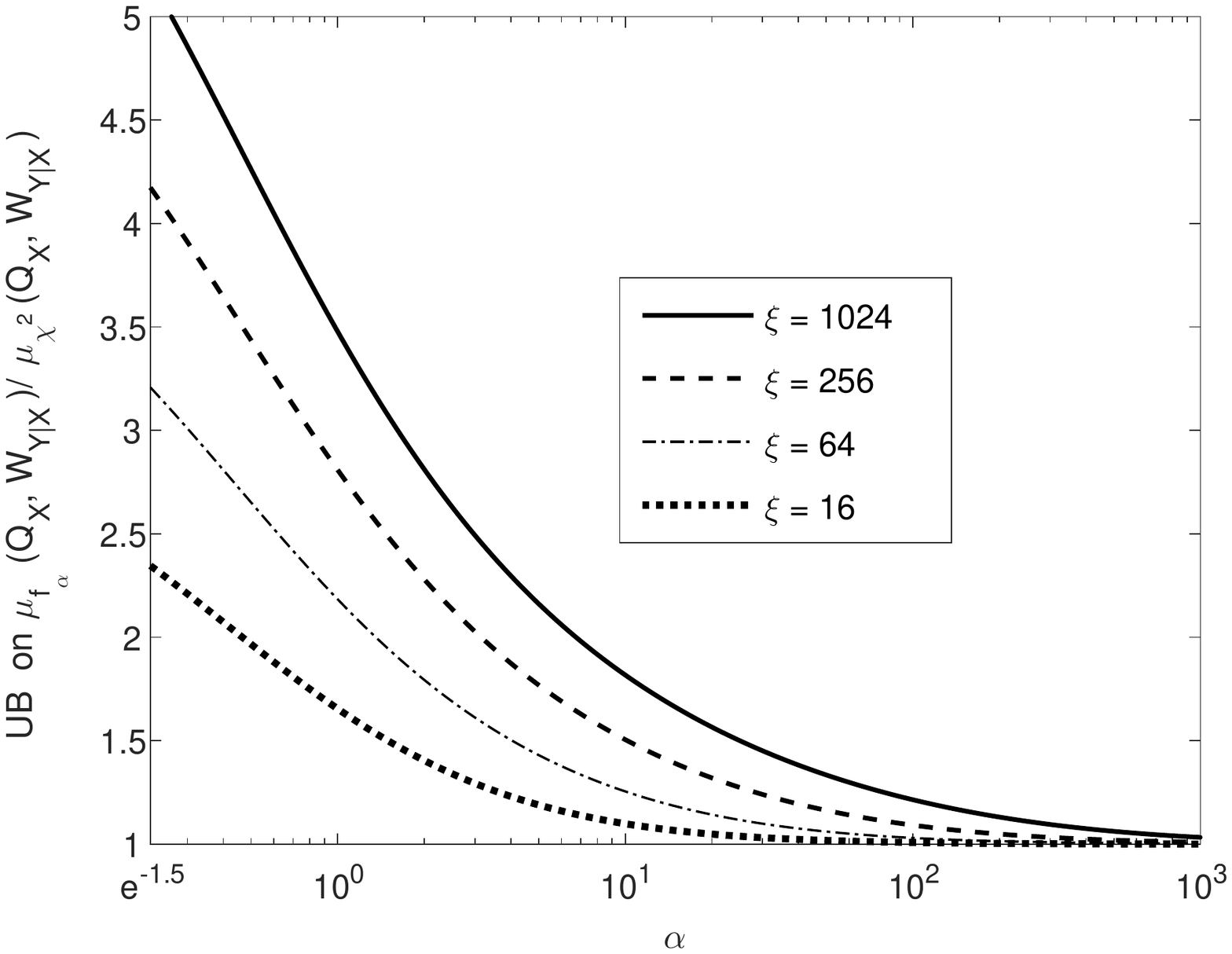}}
\end{center}
\vspace*{-3.8cm}
\caption{\label{figure: contraction ratio}
Curves of the upper bound on the ratio of the contraction coefficients
$\frac{\mu_{f_\alpha}(Q_X, W_{Y|X})}{\mu_{\chi^2}(Q_X, W_{Y|X})}$
(see the right-side inequality of \eqref{ratio contractions - 1})
as a function of the parameter $\alpha \geq \mathrm{e}^{-\frac32}$.
The curves correspond to different values of $\xi$ in
\eqref{xi - reciprocal of min Q_X}.}
\end{figure}

\subsection{Illustration of Theorem~\ref{thm: LB/UB f-div} and Further Results}
\label{subsection: Illustration of Thm. LB/UB f-div}

Theorem~\ref{thm: LB/UB f-div} provides upper and lower bounds on
an $f$-divergence, $D_f(Q \| U_n)$, from any probability
mass function $Q$ supported on a finite set of cardinality $n$
to an equiprobable distribution over this set.

We apply in the following, the exact formula for
\begin{align}
\label{d_f - def}
d_f(\rho) := \underset{n \to \infty}{\lim} \,
\underset{Q \in \set{P}_n(\rho)}{\max} D_f(Q \| U_n),
\quad \rho \geq 1
\end{align}
to several important $f$-divergences. From \eqref{asympt. max1},
\begin{align}
\label{d_f}
d_f(\rho) = \max_{x \in [0,1]} \left\{ x f\biggl(\frac{\rho}{1+(\rho-1)x} \biggr)
+ (1-x) f\biggl(\frac1{1+(\rho-1)x} \biggr) \right\}, \quad \rho \geq 1.
\end{align}
Since $f$ is a convex function on $(0, \infty)$ with $f(1)=0$,
Jensen's inequality implies that the function which is subject to maximization
in the right-side of \eqref{d_f} is non-negative over the interval $[0,1]$.
It is equal to zero at the endpoints of the interval $[0,1]$, so the maximum
over this interval is attained at an interior point.
Note also that, in view of Items~\ref{Thm. 5-d}) and~\ref{Thm. 5-e}) of
Theorem~\ref{thm: LB/UB f-div}, the exact asymptotic expression
in \eqref{d_f} satisfies
\begin{align}
\label{UB - finite n}
\max_{Q \in \set{P}_n(\rho)} D_f(Q \| U_n) \leq d_f(\rho), \quad
\forall \, n \in \{2, 3, \ldots\}, \; \rho \geq 1.
\end{align}

\subsubsection{Total variation distance}
This distance is an $f$-divergence with $f(t) := |t-1|$ for $t > 0$.
Substituting $f$ into \eqref{d_f} gives
\begin{align}
\label{d_f - TV 1}
d_f(\rho) = \max_{x \in [0,1]} \Biggl\{ \frac{2(\rho-1)x(1-x)}{1+(\rho-1)x} \Biggr\}.
\end{align}
By setting to zero the derivative of the function which is subject to maximization
in the right side of \eqref{d_f - TV 1}, it can be verified that the maximizer over
this interval is equal to $x = \frac1{1+\sqrt{\rho}}$, which implies that
\begin{align}
d_f(\rho) = \frac{2(\sqrt{\rho}-1)}{\sqrt{\rho}+1}, \quad \forall \, \rho \geq 1.
\end{align}

\subsubsection{Alpha divergences}
\label{subsubsection: Alpha divergences}
The class of Alpha divergences forms a parametric subclass of the $f$-divergences,
which includes in particular the relative entropy, $\chi^2$-divergence, and the
squared-Hellinger distance. For $\alpha \in \Reals$, let
\begin{align}
\label{Alpha-divergence}
D_{\mathrm{A}}^{(\alpha)}(P \| Q) := D_{u_\alpha}(P\|Q),
\end{align}
where $u_\alpha \colon (0, \infty) \to \Reals$ is a non-negative and convex
function with $u_\alpha(1) = 0$, which is defined for $t>0$ as follows (see
\cite[Chapter~2]{LieseV_book87}, followed by studies in, e.g., \cite{AmariN00},
\cite{CichockiA10}, \cite{LieseV_IT2006}, \cite{PardoV97} and \cite{Sason18}):
\begin{align}
\label{f of Alpha-divergence}
u_\alpha(t) :=
\begin{dcases}
\frac{t^\alpha - \alpha(t-1) - 1}{\alpha (\alpha-1)}, & \quad
\alpha \in (-\infty, 0) \cup (0,1) \cup (1, \infty), \\
t \log_{\mathrm{e}}t + 1-t, & \quad \alpha=1, \\
-\log_{\mathrm{e}}t, & \quad \alpha=0.
\end{dcases}
\end{align}
The functions $u_0$ and $u_1$ are defined in the right side of
\eqref{f of Alpha-divergence} by a continuous extension of $u_\alpha$
at $\alpha=0$ and $\alpha=1$, respectively. The following relations
hold (see, e.g., \cite[(10)--(13)]{CichockiA10}):
\begin{align}
\label{KL1}
& D_{\mathrm{A}}^{(1)}(P\|Q) = \tfrac1{\log \mathrm{e}} \, D(P\|Q), \\
\label{KL2}
& D_{\mathrm{A}}^{(0)}(P\|Q) = \tfrac1{\log \mathrm{e}} \, D(Q\|P), \\
\label{chi2}
& D_{\mathrm{A}}^{(2)}(P\|Q) = \tfrac12 \, \chi^2(P\|Q), \\
\label{chi2b}
& D_{\mathrm{A}}^{(-1)}(P\|Q) = \tfrac12 \, \chi^2(Q\|P), \\
\label{sqr. Hel}
& D_{\mathrm{A}}^{(\frac12)}(P\|Q) = 4 \mathscr{H}^2(P\|Q).
\end{align}

Substituting $f := u_\alpha$ (see \eqref{f of Alpha-divergence}) into the
right side of \eqref{d_f} gives that
\begin{align}
\label{Delta def 0}
\Delta(\alpha,\rho) &:= d_{u_\alpha}(\rho) \\
\label{Delta def 1}
& \; = \lim_{n \to \infty} \max_{Q \in \set{P}_n(\rho)}
D_{\mathrm{A}}^{(\alpha)}(Q \| U_n) \\
\label{Delta 1: alpha-div.}
& \; = \max_{x \in [0,1]}
\left\{ \frac{1+(\rho^\alpha-1)x}{\bigl(1+(\rho-1)x \bigr)^\alpha} - 1 \right\}.
\end{align}
Setting to zero the derivative of the function which is subject to maximization
in the right side of \eqref{Delta 1: alpha-div.} gives
\begin{align}
\label{x: alpha-div.}
x = x^\ast := \frac{1+\alpha(\rho-1)-\rho^\alpha}{(1-\alpha)(\rho-1)(\rho^\alpha-1)},
\end{align}
where it can be verified that $x^\ast \in (0,1)$ for all $\alpha \in (-\infty, 0)
\cup (0,1) \cup (1, \infty)$ and $\rho > 1$. Substituting \eqref{x: alpha-div.}
into the right side of \eqref{Delta 1: alpha-div.} gives that, for all such
$\alpha$ and $\rho$,
\begin{align}
\label{Delta 2: alpha-div.}
\Delta(\alpha,\rho) = \frac1{\alpha(\alpha-1)} \left[ \frac{(1-\alpha)^{\alpha-1}
(\rho^\alpha-1)^\alpha (\rho - \rho^\alpha)^{1-\alpha}}{(\rho-1) \alpha^\alpha} - 1 \right].
\end{align}
By a continuous extension of $\Delta(\alpha,\rho)$ in \eqref{Delta 2: alpha-div.} at
$\alpha=1$ and $\alpha=0$, it follows that for all $\rho > 1$
\begin{align}
\label{continuous extension of Delta}
\Delta(1, \rho) = \Delta(0, \rho) = \frac{\rho \log \rho}{\rho-1}
- \log \left( \frac{\mathrm{e} \rho \log_{\mathrm{e}} \rho}{\rho-1} \right).
\end{align}
Consequently, for all $\rho > 1$,
\begin{align}
& \lim_{n \to \infty} \max_{Q \in \set{P}_n(\rho)} D(Q \| U_n) \nonumber \\
\label{05062019-a1}
&= \log \mathrm{e} \; \lim_{n \to \infty} \max_{Q \in \set{P}_n(\rho)}
D_{\mathrm{A}}^{(1)}(Q \| U_n) \\
\label{05062019-a2}
&=  \Delta(1,\rho) \; \log \mathrm{e} \\
\label{04062019a}
&= \frac{\rho \log \rho}{\rho-1} - \log \left( \frac{\mathrm{e}
\rho \log_{\mathrm{e}} \rho}{\rho-1} \right),
\end{align}
where \eqref{05062019-a1} holds due to \eqref{KL1}; \eqref{05062019-a2}
is due to \eqref{Delta def 1}, and \eqref{04062019a} holds due to
\eqref{continuous extension of Delta}. This sharpens the result in
\cite[Theorem~2]{CicaleseGV18} for the relative entropy from the
equiprobable distribution, $D(Q \| U_n) = \log n - H(Q)$, by showing
that the bound in \cite[(7)]{CicaleseGV18} is asymptotically tight as
we let $n \to \infty$. The result in \cite[Theorem~2]{CicaleseGV18}
can be further tightened for finite $n$ by applying the result in
Theorem~\ref{thm: LB/UB f-div}-~\ref{Thm. 5-d}) with
$$f(t) :=  u_1(t) \, \log \mathrm{e} = t \log t + (1-t) \log \mathrm{e}$$
for all $t>0$ (although, unlike the asymptotic result in
\eqref{Delta 2: alpha-div.}, the refined bound for a finite $n$ does
not lend itself to a closed-form expression as a function of $n$; see
also \cite[Remark~3]{Sason18b}, which provides such a refinement of
the bound on $D(Q\|U_n)$ for finite $n$ in a different approach).

From \eqref{KL2}, \eqref{Delta def 1} and \eqref{continuous extension of Delta},
it follows similarly to \eqref{04062019a} that for all $\rho > 1$
\begin{align}
\label{10062019c1}
\lim_{n \to \infty} \max_{Q \in \set{P}_n(\rho)} D(U_n \| Q)
&=  \Delta(0,\rho) \; \log \mathrm{e} \\
\label{04062019b}
&= \frac{\rho \log \rho}{\rho-1} -
\log \left( \frac{\mathrm{e} \rho \log_{\mathrm{e}} \rho}{\rho-1} \right).
\end{align}
It should be noted that in view of the one-to-one correspondence between
the R\'{e}nyi divergence and the Alpha divergence of the same order $\alpha$
where, for $\alpha \neq 1$,
\begin{align}
\label{1-1 Renyi and Alpha div.}
D_\alpha(P\|Q) = \frac1{\alpha-1} \, \log \Bigl(1 + \alpha(\alpha-1)
D_{\mathrm{A}}^{(\alpha)}(P \| Q) \Bigr),
\end{align}
the asymptotic result in \eqref{Delta 2: alpha-div.} can be obtained from
\cite[Lemma~4]{Sason18b} and vice versa; however, in \cite{Sason18b}, the
focus is on the R\'{e}nyi divergence from the equiprobable distribution,
whereas the result in \eqref{Delta 2: alpha-div.} is obtained by specializing
the asymptotic expression in \eqref{d_f} for a general $f$-divergence. Note
also that the result in \cite[Lemma~4]{Sason18b} is restricted to $\alpha > 0$,
whereas the result in \eqref{Delta 2: alpha-div.} and
\eqref{continuous extension of Delta} covers all values of $\alpha \in \Reals$.

In view of \eqref{Delta def 1}, \eqref{Delta 2: alpha-div.}, \eqref{04062019a},
\eqref{04062019b}, and the special cases of the Alpha divergences in
\eqref{KL1}--\eqref{sqr. Hel}, it follows that for all $\rho > 1$ and for all
integer $n \geq 2$
\begin{align}
\label{KL1 asymp.}
\max_{Q \in \set{P}_n(\rho)} D(Q \| U_n)
& \leq \Delta(1,\rho) \, \log \mathrm{e}
= \frac{\rho \log \rho}{\rho-1} - \log \left( \frac{\mathrm{e} \rho
\log_{\mathrm{e}} \rho}{\rho-1} \right), \\[0.1cm]
\label{KL2 asymp.}
\max_{Q \in \set{P}_n(\rho)} D(U_n \| Q)
& \leq \Delta(0,\rho) \, \log \mathrm{e}
= \frac{\rho \log \rho}{\rho-1} - \log \left( \frac{\mathrm{e} \rho
\log_{\mathrm{e}} \rho}{\rho-1} \right), \\
\label{chi2 asymp.}
\max_{Q \in \set{P}_n(\rho)} \chi^2(Q \| U_n) &
\leq 2 \Delta(2,\rho) = \frac{(\rho-1)^2}{4 \rho}, \\
\label{chi2-b asymp.}
\max_{Q \in \set{P}_n(\rho)} \chi^2(U_n \| Q) & \leq 2 \Delta(-1,\rho)
= \frac{(\rho-1)^2}{4 \rho}, \\
\label{sqr. Hel. asymp.}
\max_{Q \in \set{P}_n(\rho)} \mathscr{H}^2(Q \| U_n)
& \leq \tfrac14 \, \Delta(\tfrac12, \rho)
= \frac{(\sqrt[4]{\rho}-1)^2}{\sqrt{\rho}+1},
\end{align}
and, furthermore, the upper bounds on the right sides of
\eqref{KL1 asymp.}--\eqref{sqr. Hel. asymp.} are asymptotically
tight in the limit where $n$ tends to infinity.

The next result characterizes the function
$\Delta \colon (0, \infty) \times (1, \infty) \to \Reals$
as it is given in \eqref{Delta 2: alpha-div.} and
\eqref{continuous extension of Delta}.
\begin{theorem}
\label{theorem: Delta_alpha}
The function $\Delta$ satisfies the following properties:
\begin{enumerate}[a)]
\item \label{Thm. 7-a}
For every $\rho > 1$, $\Delta(\alpha, \rho)$ is a convex function
of $\alpha$ over the real line, and it is symmetric around
$\alpha = \tfrac12$ with a global minimum at $\alpha = \tfrac12$.

\item \label{Thm. 7-b}
The following inequalities hold:
\begin{align}
\label{mon. 1}
& \alpha \, \Delta(\alpha, \rho) \leq \beta \, \Delta(\beta, \rho),
\hspace*{2.7cm} 0 < \alpha \leq \beta < \infty, \\
\label{mon. 2}
& (1-\beta) \, \Delta(\beta, \rho) \leq (1-\alpha) \, \Delta(\alpha, \rho),
\quad -\infty < \alpha \leq \beta < 1.
\end{align}

\item \label{Thm. 7-c}
For every $\alpha \in \Reals$, $\Delta(\alpha, \rho)$ is monotonically
increasing and continuous in $\rho \in (1, \infty)$, and
$\underset{\rho \to 1^+}{\lim} \Delta(\alpha, \rho) = 0$.
\end{enumerate}
\end{theorem}
\begin{proof}
See Appendix~\ref{appendix: Delta_alpha} (Part~A).
\end{proof}

\begin{remark}
\label{remark: symmetry - Alpha-divergence}
The symmetry of $\Delta(\alpha, \rho)$ around $\alpha = \tfrac12$
(see Theorem~\ref{theorem: Delta_alpha}~\ref{Thm. 7-a})) is not
implied by the following symmetry property of the Alpha divergence
around $\alpha=\tfrac12$ (see, e.g., \cite[p.~36]{LieseV_book87}):
\begin{align}
D_{\mathrm{A}}^{(\frac12 + \alpha)}(P \| Q)
= D_{\mathrm{A}}^{(\frac12 - \alpha)}(Q \| P).
\end{align}
\end{remark}

Relying on Theorem~\ref{theorem: Delta_alpha}, the following corollary
gives a similar result to \eqref{Delta def 1} where the order of $Q$
and $U_n$ in $D_{\mathrm{A}}^{(\alpha)}(\cdot \| \cdot)$ is switched.

\begin{corollary}
\label{corollary: Delta}
For all $\alpha \in \Reals$ and $\rho > 1$,
\begin{align}
\label{04062019c}
\lim_{n \to \infty} \max_{Q \in \set{P}_n(\rho)}
D_{\mathrm{A}}^{(\alpha)}(U_n \| Q) = \Delta(\alpha, \rho).
\end{align}
\end{corollary}
\begin{proof}
See Appendix~\ref{appendix: Delta_alpha} (Part~B).
\end{proof}

We next further exemplify Theorem~\ref{thm: LB/UB f-div} for the relative
entropy. \newline
Let $f(t) := t \log t + (1-t) \log \mathrm{e}$ for $t>0$. Then,
$f''(t)=\tfrac{\log \mathrm{e}}{t}$, so the bounds on the second derivative of $f$
over the interval $\bigl[\frac1\rho, \rho]$ are given by $M=\rho \log \mathrm{e}$
and $m=\tfrac{\log \mathrm{e}}{\rho}$. Theorem~\ref{thm: LB/UB f-div}~\ref{Thm. 5-h})
gives the following bounds:
\begin{align}
\label{LB/UB Sh. entropy}
\frac{\bigl(n \| Q \|_2^2 - 1 \bigr) \log \mathrm{e}}{2 \rho}
\leq D(Q \| U_n)
\leq \frac{\rho \, \bigl(n \| Q \|_2^2 - 1 \bigr) \log \mathrm{e}}{2}.
\end{align}
From \cite[Theorem~2]{CicaleseGV18} (and \eqref{KL1 asymp.}),
\begin{align}
\label{280519a}
D(Q \| U_n) \leq \frac{\rho \log \rho}{\rho-1}
- \log \left(\frac{\mathrm{e} \rho \log_{\mathrm{e}}\rho}{\rho-1} \right).
\end{align}
Furthermore, \eqref{UB Df-M} gives that
\begin{align}
\label{270519k}
D(Q \| U_n) \leq \tfrac18 (\rho-1)^2 \log \mathrm{e},
\end{align}
which, for $\rho>1$, is a looser bound in comparison to \eqref{280519a}. It can be
verified, however, that the dominant term in the Taylor series expansion (around
$\rho=1$) of the right side of \eqref{280519a} coincides with the right side of
\eqref{270519k}, so the bounds scale similarly for small values of $\rho \geq 1$.

Suppose that we wish to assert that, for every integer $n \geq 2$ and for all
probability mass functions $Q \in \set{P}_n(\rho)$, the condition
\begin{align}
\label{280519b}
D(Q \| U_n) \leq d \log \mathrm{e}
\end{align}
holds with a fixed $d>0$. Due to the left side inequality in \eqref{convergence rate 1},
this condition is equivalent to the requirement that
\begin{align}
\label{10062019a}
\lim_{n \to \infty} \max_{Q \in \set{P}_n(\rho)} D(Q \| U_n) \leq d \log \mathrm{e}.
\end{align}
Due to the asymptotic tightness of the upper bound in the right side
of \eqref{KL1 asymp.} (as we let $n \to \infty$), requiring that this upper
bound is not larger than $d \log \mathrm{e}$ is necessary and sufficient for
the satisfiability of \eqref{280519b} for all $n$ and $Q \in \set{P}_n(\rho)$.
This leads to the analytical solution $\rho \leq \rho_{\max}^{(1)}(d)$ with
(see Appendix~\ref{appendix: Lambert-W})
\begin{align}
\label{280519c}
\rho_{\max}^{(1)}(d)
:= \frac{W_{-1}\bigl(-\mathrm{e}^{-d-1}\bigr)}{W_0\bigl(-\mathrm{e}^{-d-1}\bigr)},
\end{align}
where $W_0$ and $W_{-1}$ denote, respectively, the principal and secondary
real branches of the Lambert $W$ function \cite{Corless96}.
Requiring the stronger condition where the right side of
\eqref{270519k} is not larger than $d \log \mathrm{e}$ leads to the sufficient
solution $\rho \leq \rho_{\max}^{(2)}$ with the simple expression
\begin{align}
\label{270519l}
\rho_{\max}^{(2)}(d) := 1 + \sqrt{8d}.
\end{align}
In comparison to $\rho_{\max}^{(1)}$ in \eqref{280519c}, $\rho_{\max}^{(2)}$
in \eqref{270519l} is more insightful; these values nearly coincide for small
values of $d>0$, providing in that case the same range of possible values of
$\rho$ for asserting the satisfiability of condition \eqref{280519b}. As it is
shown in Figure~\ref{figure:KL-rho vs. d}, for $d \leq 0.01$, the difference
between the maximal values of $\rho$ in \eqref{280519c} and \eqref{270519l} is
marginal, though in general $\rho_{\max}^{(1)}(d) > \rho_{\max}^{(2)}(d)$ for
all $d>0$.
\begin{figure}[ht]
\vspace*{-3.8cm}
\begin{center}
\centerline{\includegraphics[width=9cm]{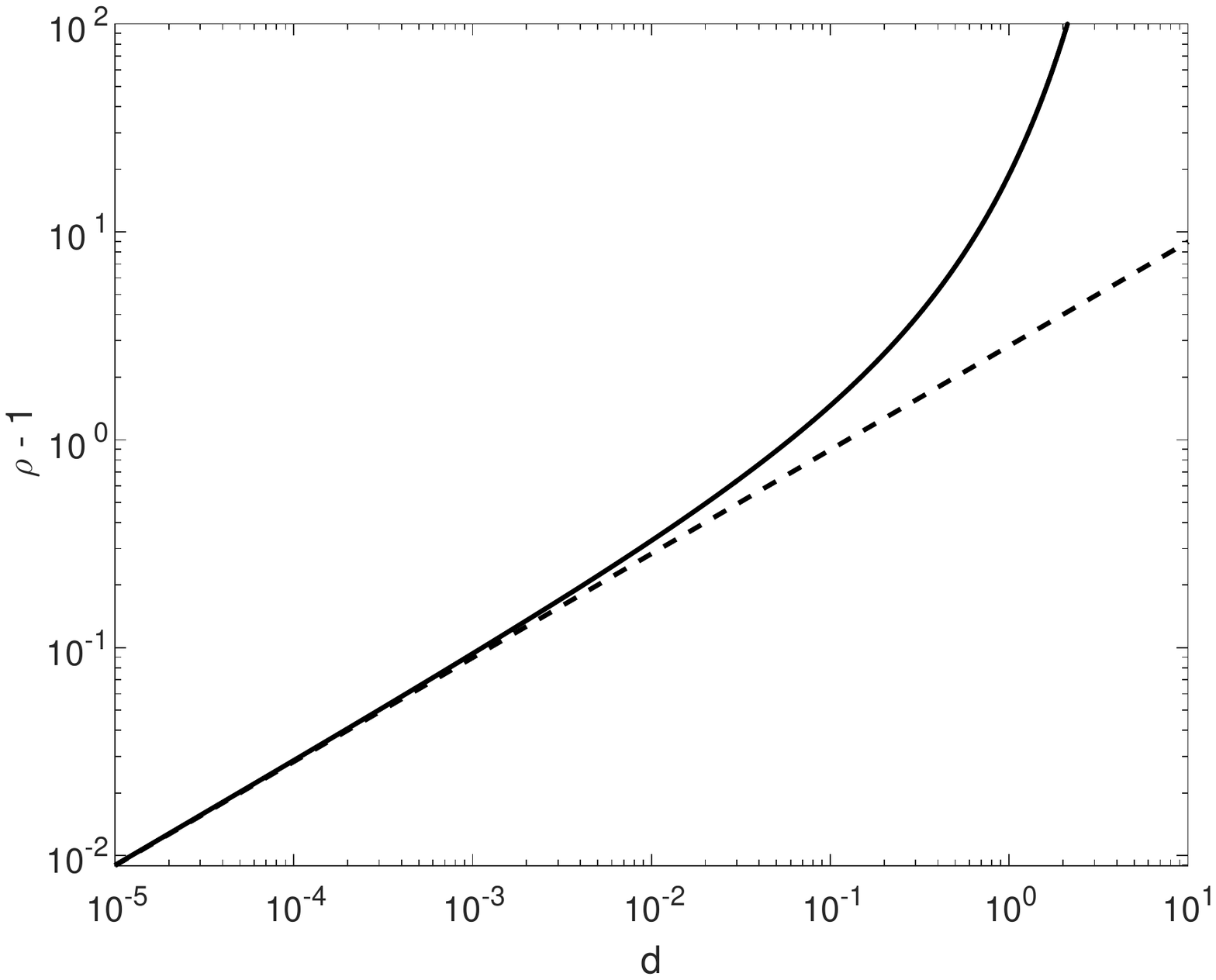}}
\end{center}
\vspace*{-3.7cm}
\caption{\label{figure:KL-rho vs. d}
A comparison of the maximal values of $\rho$ (minus 1) according to
\eqref{280519c} and \eqref{270519l},
asserting the satisfiability of the condition $D(Q \| U_n) \leq d \log \mathrm{e}$,
with an arbitrary $d>0$, for all integers $n \geq 2$ and probability mass functions
$Q$ supported on $\{1, \ldots, n\}$ with $\tfrac{q_{\max}}{q_{\min}} \leq \rho$. The
solid line refers to the necessary and sufficient condition which gives \eqref{280519c},
and the dashed line refers to a stronger condition which gives \eqref{270519l}.}
\end{figure}

\subsubsection{The subclass of $f$-divergences in Theorem~\ref{thm: f_alpha-divergence}}
This example refers to the subclass of $f$-divergences in
Theorem~\ref{thm: f_alpha-divergence}. For these $f_\alpha$-divergences, with
$\alpha \geq \mathrm{e}^{-\frac32}$, substituting $f := f_\alpha$ from \eqref{f_alpha}
into the right side of \eqref{d_f} gives that for all $\rho \geq 1$
\begin{align}
& \Phi(\alpha,\rho) \nonumber \\
\label{Phi def 0}
&:= d_{f_\alpha}(\rho) \\
\label{Phi def 1}
& \; = \lim_{n \to \infty} \max_{Q \in \set{P}_n(\rho)} D_{f_\alpha}(Q \| U_n) \\
& \; = \max_{x \in [0,1]}
\left\{ x \left(\alpha + \frac{\rho}{1+(\rho-1)x} \right)^2 \log\left(\alpha
+ \frac{\rho}{1+(\rho-1)x} \right) - (\alpha+1)^2 \, \log(\alpha+1) \right. \nonumber \\
\label{Phi 1: alpha-div.}
& \hspace*{1.9cm} \left. + \, (1-x) \left(\alpha + \frac1{1+(\rho-1)x} \right)^2
\log \left( \alpha + \frac1{1+(\rho-1)x} \right) \right\}.
\end{align}
The exact asymptotic expression in the right side of
\eqref{Phi 1: alpha-div.} is subject to numerical maximization.

We next provide two alternative closed-form upper bounds, based on
Theorems~\ref{thm: f_alpha-divergence} and~\ref{thm: LB/UB f-div}, and study their
tightness.
The two upper bounds, for all $\alpha \geq \mathrm{e}^{-\frac32}$ and
$\rho \geq 1$, are given by (see Appendix~\ref{appendix: Phi - UBs})
\begin{align}
\Phi(\alpha,\rho) & \leq \left[ \log(\alpha+1) + \tfrac32 \log \mathrm{e}
- \frac{\log \mathrm{e}}{\alpha+1} \right] \frac{(\rho-1)^2}{4 \rho} \nonumber \\[0.1cm]
\label{Phi - UB1}
& \hspace*{0.4cm} + \frac{\log \mathrm{e}}{81(\alpha+1)}
\left( \frac{(\rho-1)(2\rho+1)(\rho+2)}{\rho(\rho+1)} \right)^2,
\end{align}
and
\begin{align}
\label{Phi - UB2}
\Phi(\alpha,\rho) \leq \Bigl[ \log(\alpha+\rho) + \tfrac32 \log \mathrm{e} \Bigr]
\, \frac{(\rho-1)^2}{4 \rho}.
\end{align}

Suppose that we wish to assert that, for every integer $n \geq 2$ and for all
probability mass functions $Q \in \set{P}_n(\rho)$, the condition
\begin{align}
\label{10062019b1}
D_{f_\alpha}(Q \| U_n) \leq d \log \mathrm{e}
\end{align}
holds with a fixed $d>0$ and $\alpha \geq \mathrm{e}^{-\frac32}$. Due to
\eqref{Phi def 0}--\eqref{Phi def 1} and the left side inequality in
\eqref{convergence rate 1}, the satisfiability of the latter condition
is equivalent to the requirement that
\begin{align}
\label{10062019b2}
\Phi(\alpha, \rho) \leq d \log \mathrm{e}.
\end{align}
In order to obtain a sufficient condition for $\rho$ to satisfy
\eqref{10062019b2}, expressed as an explicit function of $\alpha$ and $d$,
the upper bound in the right side of \eqref{Phi - UB1} is slightly
loosened to
\begin{align}
\label{Phi - UB3}
\Phi(\alpha,\rho) \leq a (\rho-1)^2 + b \min\{\rho-1, (\rho-1)^2\},
\end{align}
where
\begin{align}
\label{10062019b3}
&a := \frac{4 \log \mathrm{e}}{81(\alpha+1)}, \\
\label{10062019b4}
&b := \tfrac14 \log(\alpha+1) + \tfrac38 \log \mathrm{e},
\end{align}
for all $\rho \geq 1$ and $\alpha \geq \mathrm{e}^{-\frac32}$.
The upper bounds in the right sides of \eqref{Phi - UB1}, \eqref{Phi - UB2}
and \eqref{Phi - UB3} are derived in Appendix~\ref{appendix: Phi - UBs}.

In comparison to \eqref{10062019b2}, the stronger requirement that the
right side of \eqref{Phi - UB3} is less than or equal to $d \log \mathrm{e}$
gives the sufficient condition
\begin{align}
\label{10062019b5}
\rho \leq \rho_{\max}(\alpha, d)
:= \max\bigl\{\rho_1(\alpha, d), \rho_2(\alpha, d) \bigr\},
\end{align}
with
\begin{align}
\label{10062019b6}
& \rho_1(\alpha, d) := 1 + \frac{\sqrt{b^2+4ad \log \mathrm{e}}-b}{2a}, \\[0.1cm]
\label{10062019b7}
& \rho_2(\alpha, d) := 1 + \sqrt{\frac{d \log \mathrm{e}}{a+b}}.
\end{align}

\begin{figure}[ht]
\vspace*{-4.2cm}
\begin{center}
\centerline{\includegraphics[width=10cm]{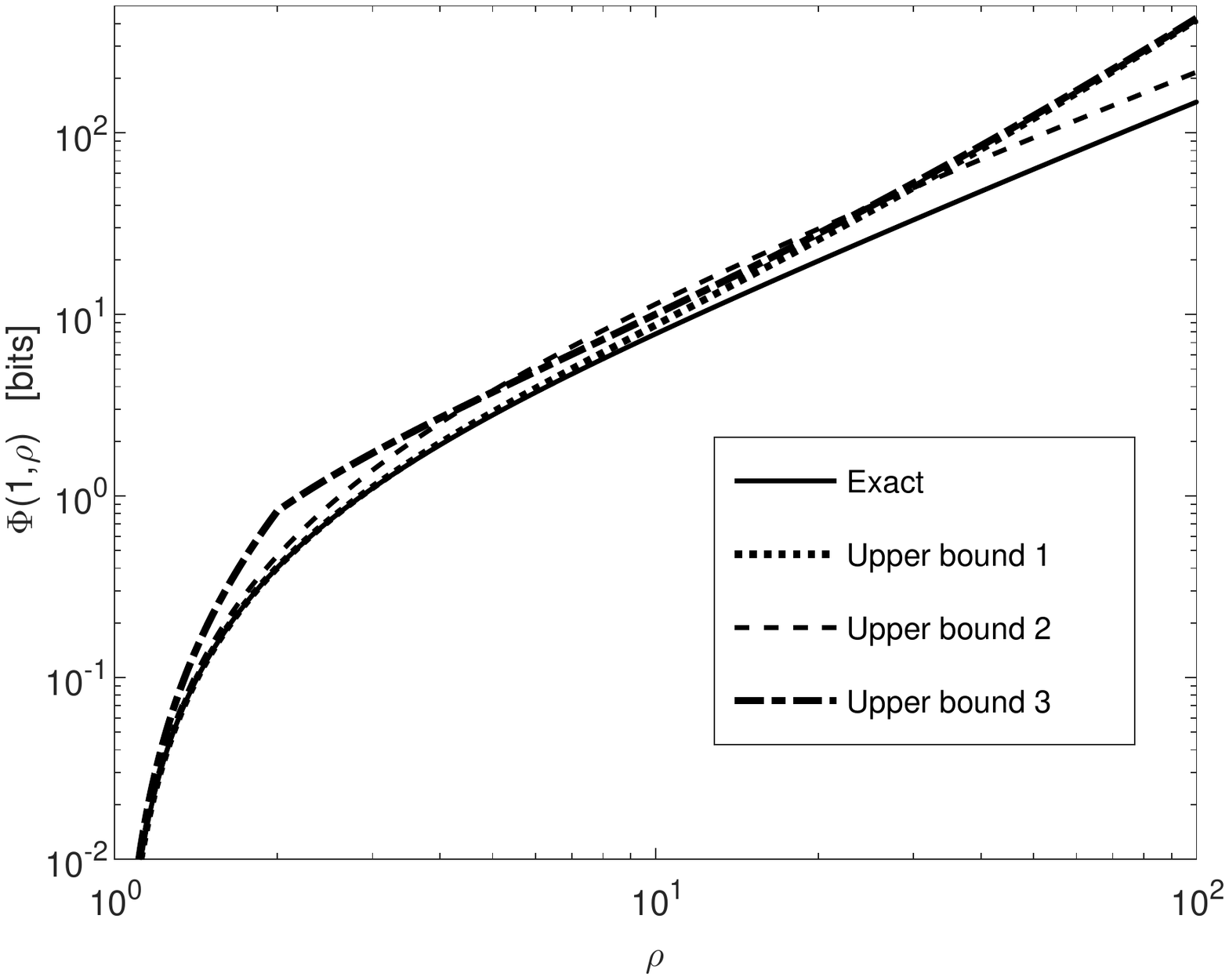}}
\end{center}
\vspace*{-4cm}
\caption{\label{figure:Phi}
A comparison of the exact expression of $\Phi(\alpha, \rho)$ in \eqref{Phi 1: alpha-div.},
with $\alpha=1$, and its three upper bounds in the right sides of \eqref{Phi - UB1},
\eqref{Phi - UB2} and \eqref{Phi - UB3} (called 'Upper bound~1' (dotted line),
'Upper bound~2' (thin dashed line), and 'Upper bound~3' (thick dashed line),
respectively).}
\end{figure}

Figure~\ref{figure:Phi} compares the exact expression in \eqref{Phi 1: alpha-div.}
with its upper bounds in \eqref{Phi - UB1}, \eqref{Phi - UB2} and \eqref{Phi - UB3}.
These bounds show good match with the exact value, and none of the bounds in
\eqref{Phi - UB1} and \eqref{Phi - UB2} is superseded by the other; the bound in
\eqref{Phi - UB3} is looser than \eqref{Phi - UB1}, and it is derived for obtaining
the closed-form solution in \eqref{10062019b5}--\eqref{10062019b7}. The bound in
\eqref{Phi - UB1} is tighter than the bound in \eqref{Phi - UB2} for small values
of $\rho \geq 1$, whereas the latter bound outperforms the first one for sufficiently
large values of $\rho$. It has been observed numerically that the tightness of the
bounds is improved by increasing the value of $\alpha$, and the range of parameters
of $\rho$ over which the bound in \eqref{Phi - UB1} outperforms the second bound in
\eqref{Phi - UB2} is enlarged when $\alpha$ is increased. It is also shown in
Figure~\ref{figure:Phi} that the bound in \eqref{Phi - UB1} and its loosened version
in \eqref{Phi - UB3} almost coincide for sufficiently small values of $\rho$ (i.e.,
for $\rho$ is close to~1), and also for sufficiently large values of $\rho$.

\subsection{An interpretation of $u_f(\cdot, \cdot)$ in Theorem~\ref{thm: LB/UB f-div}}
\label{subsection: u_f}

We provide here an interpretation of $u_f(n, \rho)$ in \eqref{def: u_f},
for $\rho > 1$ and an integer $n \geq 2$; note that $u_f(n, 1) \equiv 0$
since $\set{P}_n(1) = \{U_n\}$.
Before doing so, recall that \eqref{opt1 Df: beta} introduces an identity
which significantly simplifies the numerical calculation of $u_f(n, \rho)$,
and \eqref{LB/UB u_f} gives (asymptotically tight) upper and lower bounds.

The following result relies on the variational representation of $f$-divergences.
\begin{theorem} \label{thm: conjugate}
Let $f \colon (0, \infty) \to \Reals$ be convex with $f(1)=0$, and
let $\overline{f} \colon \Reals \to \Reals \cup \{\infty\}$ be the convex
conjugate function of $f$ (a.k.a. the Fenchel-Legendre transform of $f$), i.e.,
\begin{align} \label{conjugate}
\overline{f}(x) := \sup_{t>0} \bigl\{tx - f(t) \bigr\}, \quad x \in \Reals.
\end{align}
Let $\rho > 1$, and define $\set{A}_n := \{1, \ldots, n\}$ for an integer
$n \geq 2$. Then, the following holds:
\begin{enumerate}[a)]

\item \label{converse}
For every $P \in \set{P}_n(\rho)$, a random variable $X \sim P$,
and a function $g \colon \set{A}_n \to \Reals$,
\begin{align}
\label{eq: converse}
\expectation[g(X)] \leq u_f(n, \rho) + \frac1n \sum_{i=1}^n \overline{f}\bigl(g(i)\bigr).
\end{align}

\item \label{achievability}
There exists $P \in \set{P}_n(\rho)$ such that, for every $\varepsilon > 0$,
there is a function $g_\varepsilon \colon \set{A}_n \to \Reals$ which satisfies
\begin{align}
\label{eq: achievability}
\expectation[g_\varepsilon(X)] \geq u_f(n, \rho) +
\frac1n \sum_{i=1}^n \overline{f}\bigl(g_\varepsilon(i)\bigr) - \varepsilon,
\end{align}
with $X \sim P$.
\end{enumerate}
\end{theorem}

\begin{proof}
See Appendix~\ref{appendix: conjugate}.
\end{proof}

\begin{remark}
The proof suggests a constructive way to obtain, for an arbitrary $\varepsilon > 0$,
a function $g_\varepsilon$ which satisfies \eqref{eq: achievability}.
\end{remark}

\section{Applications in Information Theory and Statistics}
\label{section: applications}

\subsection{Bounds on the List Decoding Error Probability
with $f$-divergences}
\label{subsection: Fano - list decoder}

The minimum probability of error of a random variable $X$ given
$Y$, denoted by $\varepsilon_{X|Y}$,
can be achieved by a deterministic function
(\textit{maximum-a-posteriori} decision rule)
$\mathcal{L}^\ast \colon \set{Y} \to \set{X}$ (see \cite{ISSV18}):
\begin{align}
\varepsilon_{X|Y}
&= \min_{\mathcal{L} \colon \set{Y} \to \set{X}}
\mathbb{P} [ X \neq \mathcal{L} (Y) ] \label{20170904} \\
&= \mathbb{P} [ X \neq \mathcal{L}^\ast (Y) ] \label{eq:MAP}\\
&= 1- \mathbb{E} \left[ \max_{x \in \set{X}}
P_{X|Y}(x|Y) \right].
\label{eq1: cond. epsilon}
\end{align}
Fano's inequality \cite{Fano52} gives an upper bound on the
conditional entropy $H(X|Y)$ as a function of $\varepsilon_{X|Y}$
(or, otherwise, providing a lower bound on $\varepsilon_{X|Y}$ as
a function of $H(X|Y))$ when $X$ takes a finite number of possible values.

The list decoding setting, in which the hypothesis tester is allowed
to output a subset of given cardinality, and an error occurs if the
true hypothesis is not in the list, has great interest in
information theory. A generalization of Fano's inequality to list decoding,
in conjunction with the blowing-up lemma \cite[Lemma~1.5.4]{Csiszar_Korner},
leads to strong converse results in multi-user information theory. This
approach was initiated in \cite[Section~5]{Ahlswede_Gacs_Korner} (see also
\cite[Section~3.6]{RS_FnT19}). The main idea of the successful combination
of these two tools is that, given a code, it is possible to blow-up the
decoding sets in a way that the probability of decoding error can be as small
as desired for sufficiently large blocklengths; since the blown-up decoding
sets are no longer disjoint, the resulting setup is a list decoder with
sub-exponential list size (as a function of the block length).

In statistics, Fano's-type lower bounds on Bayes and minimax risks,
expressed in terms of $f$-divergences, are derived in \cite{ChenGZ16} and
\cite{Guntuboyina11}.

In this section, we further study the setup of list decoding, and derive
bounds on the average list decoding error probability. We first
consider the special case where the list size is fixed (see
Section~\ref{subsubsection: fixed-size list decoding}), and then
move to the more general case of a list size which depends on the channel
observation (see Section~\ref{subsubsection: variable-size list decoding}).

\subsubsection{Fixed-Size List Decoding}
\label{subsubsection: fixed-size list decoding}

A generalization of Fano's inequality for fixed-size list decoding is
given in \cite[(139)]{ISSV18}, expressed as a function of the conditional
Shannon entropy (strengthening \cite[Lemma~1]{Kim_Sutivong_Cover}).
A further generalization in this setup, which is expressed as a function
of the Arimoto-R\'enyi conditional entropy with an arbitrary positive order
(see Definition~\ref{definition: AR conditional entropy}), is provided
in \cite[Theorem~8]{ISSV18}.

The next result provides a generalized Fano's inequality for fixed-size
list decoding, expressed in terms of an arbitrary $f$-divergence. Some
earlier results in the literature are reproduced from the next result,
followed by its strengthening as an application of Theorem~\ref{thm: SDPI-IS}.

\begin{theorem}
\label{theorem: generalized Fano Df}
Let $P_{XY}$ be a probability measure defined on
$\set{X} \times \set{Y}$ with $|\set{X}|=M$. Consider
a decision rule $\set{L} \colon \set{Y} \to \binom{\set{X}}{L}$,
where $\binom{\set{X}}{L}$ stands for the set of subsets
of $\set{X}$ with cardinality $L$, and $L < M$ is fixed.
Denote the list decoding error probability by
$P_{\set{L}} := \prob \bigl[ X \notin \set{L}(Y) \bigr]$.
Let $U_M$ denote an equiprobable probability mass
function on $\set{X}$. Then, for every convex function
$f \colon (0, \infty) \to \Reals$ with $f(1)=0$,
\begin{align}
\label{generalized Fano Df}
\expectation\Bigl[D_f \bigl(P_{X|Y}(\cdot|Y) \, \| \, U_M \bigr) \Bigr]
\geq \frac{L}{M} \; f\biggl(\frac{M \, (1-P_{\set{L}})}{L} \biggr)
+ \biggl(1-\frac{L}{M}\biggr) \; f\biggl(\frac{M P_{\set{L}}}{M-L} \biggr).
\end{align}
\end{theorem}
\begin{proof}
See Appendix~\ref{appendix: generalized Fano Df}.
\end{proof}

\begin{remark}
The special case where $L=1$ (i.e., a decoder with a single
output) gives \cite[(5)]{Guntuboyina11}.
\end{remark}

As consequences of Theorem~\ref{theorem: generalized Fano Df}, we first
reproduce some earlier results as special cases.
\begin{corollary} \cite[(139)]{ISSV18}
\label{corollary: Fano - list}
Under the assumptions in Theorem~\ref{theorem: generalized Fano Df},
\begin{align}
\label{ISSV18 - Fano}
H(X|Y) \leq \log M - d\biggl(P_{\set{L}} \, \| \, 1-\frac{L}{M} \biggr)
\end{align}
where $d(\cdot \| \cdot) \colon [0,1] \times [0,1] \to [0, +\infty]$ denotes
the binary relative entropy, defined as the continuous extension of
$D([p, 1-p] \| [q, 1-q]) := p \log \frac{p}{q} + (1-p) \log \frac{1-p}{1-q}$
for $p,q \in (0,1)$.
\end{corollary}
\begin{proof}
The choice $f(t) := t \log t + (1-t) \log \mathrm{e}$, for all $t>0$, (so the
equality $f(t) = u_1(t) \log \mathrm{e}$ holds, for $t>0$, with $u_1(\cdot)$
defined in \eqref{f of Alpha-divergence}) gives
\begin{align}
\label{21062019a6}
\expectation\Bigl[D_f \bigl(P_{X|Y}(\cdot|Y) \, \| \, U_M \bigr) \Bigr]
& = \int_{\set{Y}} \mathrm{d}P_Y(y) \, D\bigl( P_{X|Y}(\cdot | y) \, \| U_M \bigr) \\
\label{21062019a7}
& = \int_{\set{Y}} \mathrm{d}P_Y(y) \, \bigl[ \log M - H(X|Y=y) \bigr] \\
\label{21062019a8}
& = \log M - H(X|Y),
\end{align}
and
\begin{align}
\label{21062019a9}
\frac{L}{M} \; f\biggl(\frac{M \, (1-P_{\set{L}})}{L} \biggr)
+ \biggl(1-\frac{L}{M}\biggr) \; f\biggl(\frac{M P_{\set{L}}}{M-L} \biggr)
= d\biggl(P_{\set{L}} \, \| \, 1-\frac{L}{M} \biggr).
\end{align}
Substituting \eqref{21062019a6}--\eqref{21062019a9} into
\eqref{generalized Fano Df} gives \eqref{ISSV18 - Fano}.
\end{proof}

Theorem~\ref{theorem: generalized Fano Df} enables to reproduce a result
in \cite{ISSV18} which generalizes Corollary~\ref{corollary: Fano - list}.
It relies on R\'{e}nyi information measures, and we first provide definitions
for a self-contained presentation.

\begin{definition} \cite{Renyientropy} \label{definition: Renyi entropy}
Let $P_X$  be a probability mass function defined on a discrete set $\set{X}$.
The \textit{R\'{e}nyi entropy of order} $\alpha \in (0,1) \cup (1, \infty)$
of $X$, denoted by $H_{\alpha}(X)$ or $H_{\alpha}(P_X)$, is defined as
\begin{align} \label{eq: Renyi entropy}
H_{\alpha}(X) & := \frac1{1-\alpha} \, \log
\sum_{x \in \set{X}} P_X^{\alpha}(x) \\
\label{eq2: Renyi entropy}
& \hspace*{0.1cm} = \frac{\alpha}{1-\alpha} \, \log \| P_X \|_{\alpha}.
\end{align}
The R\'enyi entropy is continuously
extended at orders $0$, $1$, and $\infty$; at order~1, it
coincides with the Shannon entropy $H(X)$.
\end{definition}

\begin{definition} \cite{Arimoto75}
\label{definition: AR conditional entropy}
Let $P_{XY}$ be defined on $\set{X} \times \set{Y}$,
where $X$ is a discrete random variable.
The \textit{Arimoto-R\'{e}nyi conditional entropy of order
$\alpha \in [ 0, \infty]$} of $X$ given $Y$ is defined as follows:
\begin{itemize}
\item
If $\alpha \in (0,1) \cup (1, \infty) $, then
\begin{align}
\label{eq1: Arimoto - cond. RE}
H_{\alpha}(X | Y) &= \frac{\alpha}{1-\alpha} \,
\log \, \mathbb{E} \left[
\left( \, \sum_{x \in \set{X}} P_{X|Y}^{\alpha}(x|Y)
\right)^{\frac1{\alpha}} \right] \\
&= \frac{\alpha}{1-\alpha} \, \log \mathbb{E}
\left[ \| P_{X|Y} ( \cdot | Y ) \|_\alpha \right] \\
\label{eq2: Arimoto - cond. RE}
&= \frac{\alpha}{1-\alpha} \, \log
\, \int_{\set{Y}} \mathrm{d}P_Y(y) \, \exp \left(
\frac{1-\alpha}{\alpha} \;
H_{\alpha}(X | Y=y) \right).
\end{align}
\item The Arimoto-R\'enyi conditional entropy is continuously
extended at orders $0$, $1$, and $\infty$; at order~1, it
coincides with the conditional Shannon entropy $H(X|Y)$.
\end{itemize}
\end{definition}

\begin{definition} \label{definition: binary RD} \cite{ISSV18}
For all $\alpha \in (0,1) \cup (1, \infty)$, the {\em binary R\'{e}nyi
divergence of order $\alpha$}, denoted by $d_{\alpha}(p\|q)$ for
$(p,q) \in [0,1]^2$, is defined as $D_{\alpha}([p, 1-p] \, \| \, [q, 1-q])$.
It is the continuous extension to $[0,1]^2$ of
\begin{align}
\label{eq1: binary RD}
d_\alpha (p \| q ) =\frac1{\alpha-1} \; \log \Bigl(p^{\alpha} q^{1-\alpha}
+ (1-p)^{\alpha} (1-q)^{1-\alpha} \Bigr).
\end{align}
For $\alpha=1$,
\begin{align}
d_1(p\|q) := \lim_{\alpha \to 1} d_\alpha(p\|q) = d(p\|q).
\end{align}
\end{definition}

The following result, generalizing Corollary~\ref{corollary: Fano - list},
is shown to be a consequence of Theorem~\ref{theorem: generalized Fano Df}.
It has been originally derived in \cite[Theorem~8]{ISSV18} in a different
way. The alternative derivation of this inequality relies on
Theorem~\ref{theorem: generalized Fano Df}, applied to the family of
Alpha-divergences (see \eqref{Alpha-divergence}) as a subclass of
the $f$-divergences.

\begin{corollary} \cite[Theorem~8]{ISSV18}
\label{corollary: generalized Fano-Renyi inequality}
Under the assumptions in Theorem~\ref{theorem: generalized Fano Df},
then for every $\alpha \in (0,1) \cup (1, \infty)$,
\begin{align}
\label{eq: generalized Fano-Renyi - list decoding}
H_{\alpha}(X | Y) & \leq \log M - d_{\alpha}\biggl( P_{\set{L}} \,
\| \, 1-\frac{L}{M} \biggr) \\
\label{eq2: generalized Fano-Renyi - list decoding}
& = \frac1{1-\alpha} \;
\log \Bigl( L^{1-\alpha} \, \bigl(1-P_{{\set{L}}}\bigr)^{\alpha}
+ (M-L)^{1-\alpha} \, P_{{\set{L}}}^{\alpha} \Bigr),
\end{align}
with equality in \eqref{eq: generalized Fano-Renyi - list decoding}
if and only if
\begin{align} \label{eq: tight Fano-Renyi - list decoding}
P_{X|Y}(x|y) =
\begin{dcases}
\frac{P_{\set{L}}}{M-L}, & \quad x \notin \set{L}(y), \\[0.2cm]
\frac{1-P_{\set{L}}}{L}, & \quad x \in \set{L}(y).
\end{dcases}
\end{align}
\end{corollary}

\begin{proof}
See Appendix~\ref{appendix: generalized Fano-Renyi inequality}.
\end{proof}

Another application of Theorem~\ref{theorem: generalized Fano Df} with
the selection $f(t) := |t-1|^s$, for $t \in [0, \infty)$ and a parameter
$s \geq 1$, gives the following result.
\begin{corollary}
\label{corollary: list dec. s>=1}
Under the assumptions in Theorem~\ref{theorem: generalized Fano Df},
for all $s \geq 1$,
\begin{align}
\label{29062019a4}
P_\set{L} \geq 1-\frac{L}{M} - \Bigl(L^{1-s}+(M-L)^{1-s} \Bigr)^{-\frac1s}
\left( \expectation \Biggl[ \, \underset{x \in \set{X}}{\sum}
\, \biggl|P_{X|Y}(x|Y) - \frac1M \biggr|^s \Biggr] \right)^{\frac1s},
\end{align}
where \eqref{29062019a4} holds with equality if $X$ and $Y$ are independent
with $X$ being equiprobable. For $s=1$ and $s=2$, \eqref{29062019a4} respectively
gives that
\begin{align}
\label{29062019a5}
& P_\set{L} \geq 1-\frac{L}{M} - \frac12 \,
\expectation \Biggl[ \, \underset{x \in \set{X}}{\sum}
\, \biggl|P_{X|Y}(x|Y) - \frac1M \biggr| \Biggr], \\[0.2cm]
\label{29062019a6}
& P_\set{L} \geq 1-\frac{L}{M} - \sqrt{\frac{L}{M}
\biggl(1-\frac{L}{M}\biggr) \bigl( M \, \expectation[P_{X|Y}(X|Y)] - 1 \bigr)}.
\end{align}
\end{corollary}

\vspace*{0.2cm}
The following refinement of the generalized Fano's inequality in
Theorem~\ref{theorem: generalized Fano Df} relies on the version
of the strong data-processing inequality in Theorem~\ref{thm: SDPI-IS}.

\begin{theorem}
\label{theorem: refined Fano's inequality}
Under the assumptions in Theorem~\ref{theorem: generalized Fano Df},
let $f \colon (0, \infty) \to \Reals$ be twice
differentiable, and assume that there exists a constant $m_f>0$ such that
\begin{align}
\label{m_f}
f''(t) \geq m_f, \quad \forall \,
t \in \set{I}(\xi_1^\ast, \xi_2^\ast),
\end{align}
where
\begin{align}
\label{28062019a1}
& \xi_1^\ast := M \inf_{(x,y) \in \set{X} \times \set{Y}} P_{X|Y}(x|y), \\
\label{28062019a2}
& \xi_2^\ast := M \sup_{(x,y) \in \set{X} \times \set{Y}} P_{X|Y}(x|y),
\end{align}
and the interval $\set{I}(\cdot, \cdot)$ is defined in \eqref{I_interval}.
Let $u^+ := \max\{u, 0\}$ for $u \in \Reals$. Then,
\begin{enumerate}[a)]
\item \label{Part a - refined Fano's inequality}
\begin{align}
\label{list dec.-26062019a}
\expectation\Bigl[D_f \bigl(P_{X|Y}(\cdot|Y) \, \| \, U_M \bigr) \Bigr]
& \geq \frac{L}{M} \; f\biggl(\frac{M \, (1-P_{\set{L}})}{L} \biggr)
+ \left(1-\frac{L}{M}\right) \; f\biggl(\frac{M P_{\set{L}}}{M-L} \biggr) \nonumber \\
& \hspace*{0.4cm} + \tfrac12 m_f \, M \left( \expectation\bigl[P_{X|Y}(X|Y)\bigr]
-\frac{1-P_{\set{L}}}{L} - \frac{P_{\set{L}}}{M-L} \right)^+.
\end{align}

\item \label{Part b - refined Fano's inequality}
If the list decoder selects the $L$ most probable elements from $\set{X}$,
given the value of $Y \in \set{Y}$, then \eqref{list dec.-26062019a} is
strengthened to
\begin{align}
\label{list dec.-26062019b}
\expectation\Bigl[D_f \bigl(P_{X|Y}(\cdot|Y) \, \| \, U_M \bigr) \Bigr]
& \geq \frac{L}{M} \; f\biggl(\frac{M \, (1-P_{\set{L}})}{L} \biggr)
+ \biggl(1-\frac{L}{M}\biggr) \; f\biggl(\frac{M P_{\set{L}}}{M-L} \biggr) \nonumber \\
& \hspace*{0.4cm} + \tfrac12 m_f \, M \left( \expectation\bigl[P_{X|Y}(X|Y)\bigr]
-\frac{1-P_{\set{L}}}{L} \right),
\end{align}
where the last term in the right side of \eqref{list dec.-26062019b} is
necessarily non-negative.
\end{enumerate}
\end{theorem}

\begin{proof}
See Appendix~\ref{appendix: refined Fano's inequality}.
\end{proof}

An application of Theorem~\ref{theorem: refined Fano's inequality} gives
the following tightened version of Corollary~\ref{corollary: Fano - list}.
\begin{corollary}
\label{corollary: Fano - list - strengthened}
Under the assumptions in Theorem~\ref{theorem: generalized Fano Df},
the following holds:
\begin{enumerate}[a)]
\item
Inequality \eqref{ISSV18 - Fano} is strengthened to
\begin{align}
H(X|Y) \leq & \log M - d\biggl(P_{\set{L}} \, \| \, 1-\frac{L}{M} \biggr) \nonumber \\
\label{29062019a1}
& -\frac{\log \mathrm{e}}{2} \;
\frac{\left( \expectation\bigl[P_{X|Y}(X|Y)\bigr]
-\frac{1-P_{\set{L}}}{L} - \frac{P_{\set{L}}}{M-L} \right)^+}{\underset{(x,y)
\in \set{X} \times \set{Y}}{\sup} P_{X|Y}(x|y)}.
\end{align}

\item
If the list decoder selects the $L$ most probable elements from $\set{X}$,
given the value of $Y \in \set{Y}$, then \eqref{29062019a1} is
strengthened to
\begin{align}
\label{29062019a2}
H(X|Y) \leq \log M - d\biggl(P_{\set{L}} \, \| \, 1-\frac{L}{M} \biggr)
-\frac{\log \mathrm{e}}{2} \cdot
\frac{\left(\expectation\bigl[P_{X|Y}(X|Y)\bigr]
-\frac{1-P_{\set{L}}}{L}\right)^{+}}{\underset{(x,y)
\in \set{X} \times \set{Y}}{\sup} P_{X|Y}(x|y)}.
\end{align}
\end{enumerate}
\end{corollary}
\begin{proof}
The choice $f(t) := t \log t + (1-t) \log \mathrm{e}$, for $t>0$, gives (see \eqref{I_interval}
and \eqref{m_f}--\eqref{28062019a2})
\begin{align}
m_f \, M &= M \inf_{t \in \set{I}(\xi_1^\ast, \xi_2^\ast)} f''t) \nonumber \\[0.1cm]
&= \frac{M \log \mathrm{e}}{\xi_2^\ast} \nonumber \\
\label{29062019a3}
& = \frac{\log \mathrm{e}}{\underset{(x,y)
\in \set{X} \times \set{Y}}{\sup} P_{X|Y}(x|y)}.
\end{align}
Substituting \eqref{21062019a6}--\eqref{21062019a9} and \eqref{29062019a3}
into \eqref{list dec.-26062019a} and \eqref{list dec.-26062019b} give, respectively,
\eqref{29062019a1} and \eqref{29062019a2}.
\end{proof}

\begin{remark}
Similarly to the bounds on $P_{\set{L}}$ in \eqref{ISSV18 - Fano} and
\eqref{eq: generalized Fano-Renyi - list decoding}, which tensorize
when $P_{X|Y}$ is replaced by a product probability measure
$P_{X^n|Y^n}(\underline{x} | \underline{y}) = \overset{n}{\underset{i=1}{\prod}} P_{X_i|Y_i}(x_i|y_i)$,
this is also the case with the new bounds in \eqref{29062019a1} and \eqref{29062019a2}.
\end{remark}

\begin{remark}
The ceil operation in the right side of \eqref{29062019a2} is redundant with
$P_{\set{L}}$ denoting the list decoding error probability (see
\eqref{09072019a1}--\eqref{09072019a7}).
However, for obtaining a lower bound on $P_{\set{L}}$ with \eqref{29062019a2},
the ceil operation assures that the bound is at least as good as the lower
bound which relies on the generalized Fano's inequality in \eqref{ISSV18 - Fano}.
\end{remark}

\begin{example}
\label{example: list-L fixed}
Let $X$ and $Y$ be random variables taking values in
$\set{X} = \{0, 1, \ldots, 8\}$ and $\set{Y} = \{0, 1\}$, respectively,
and let $P_{XY}$ be the joint probability mass function, given by
\begin{align} \label{P}
\bigl[P_{XY}(x,y)\bigr]_{(x,y) \in \set{X} \times \set{Y}}
= \frac1{512} \left( \begin{array}{rrrrrrrrr}
128   & 64   & 32   & 16    & 8    &  4    &  2    &  1     & 1\\
  2   &  2   &  2   &  2    & 8    & 16    & 32    & 64   & 128
\end{array}
\right)^{\mathrm{T}}.
\end{align}
Let the list decoder select the $L$ most probable elements from $\set{X}$,
given the value of $Y \in \set{Y}$. Table~\ref{table1} compares the list decoding error
probability $P_{\set{L}}$ with the lower bound which relies on the generalized Fano's
inequality in \eqref{ISSV18 - Fano}, its tightened version in \eqref{29062019a2}, and
the closed-form lower bound in \eqref{29062019a6} for fixed list sizes of $L=1, \ldots, 4$.
For $L=3$ and $L=4$, \eqref{29062019a2} improves the lower bound in \eqref{ISSV18 - Fano}
(see Table~\ref{table1}).
If $L=4$, then the generalized Fano's lower bound in \eqref{ISSV18 - Fano} and also
\eqref{29062019a6} are useless, whereas \eqref{29062019a2} gives a non-trivial
lower bound. It is shown here that none of the new lower bounds in \eqref{29062019a6}
and \eqref{29062019a2} is superseded by the other.

\begin{table}[h]
\renewcommand{\arraystretch}{1.5}
\begin{center}
\begin{tabular}{|r|c|c|c|c|}
\hline
$L$ & Exact $P_{\set{L}}$ & \eqref{ISSV18 - Fano} &\eqref{29062019a2} & \eqref{29062019a6} \\
\hline
   1 & 0.500 & 0.353 & 0.353 & 0.444 \\
   2 & 0.250 & 0.178 & 0.178 & 0.190 \\
   3 & 0.125 & 0.065 & 0.072 & $5.34 \cdot 10^{-5}$ \\
   4 & 0.063 & 0 & 0.016  & 0 \\
\hline
\end{tabular}
\vspace*{0.7cm}
\caption{\label{table1} The lower bounds on $P_{\set{L}}$ in \eqref{ISSV18 - Fano}, \eqref{29062019a6}
and \eqref{29062019a2}, and its exact value for fixed list size $L$ (see Example~\ref{example: list-L fixed}).}
\end{center}
\end{table}
\end{example}

\subsubsection{Variable-Size List Decoding}
\label{subsubsection: variable-size list decoding}
In the more general setting of list decoding where the size of the list
may depend on the channel observation, Fano's inequality has been
generalized as follows.
\begin{proposition} (\cite{AhlswedeK75} and \cite[Appendix~3.E]{RS_FnT19})
\label{prop: Fano-Ahlswede-Korner}
Let $P_{XY}$ be a probability measure defined on $\set{X} \times \set{Y}$
with $|\set{X}|=M$. Consider a decision rule $\set{L} \colon \set{Y} \to
2^{\set{X}}$, and let the (average) list decoding error probability be
given by $P_{\set{L}} := \prob \bigl[ X \notin \set{L}(Y) \bigr]$ with
$|\set{L}(y)| \geq 1$ for all $y \in \set{Y}$. Then,
\begin{align}
\label{Fano-Ahlswede-Korner 1}
H(X|Y) \leq h(P_\set{L}) + \expectation[\log |\set{L}(Y)|] + P_{\set{L}} \log M,
\end{align}
where $h \colon [0,1] \to [0, \log 2]$ denotes the binary entropy function.
If $|\set{L}(Y)| \leq N$ almost surely, then also
\begin{align}
\label{Fano-Ahlswede-Korner 2}
H(X|Y) \leq h(P_\set{L}) + (1-P_{\set{L}}) \log N + P_{\set{L}} \log M.
\end{align}
\end{proposition}

By relying on the data-processing inequality for $f$-divergences, we derive in
the following an alternative explicit lower bound on the average list decoding
error probability $P_{\set{L}}$. The derivation relies on the $E_\gamma$ divergence
(see, e.g., \cite{LCV17}), which forms a subclass of the $f$-divergences.

\begin{theorem}
\label{theorem: LB - variable list size}
Under the assumptions in \eqref{Fano-Ahlswede-Korner 1}, for every $\gamma \geq 1$,
\begin{align}
\label{LB - variable list size}
P_{\set{L}} \geq \frac{1+\gamma}{2} - \frac{\gamma \expectation[|\set{L}(Y)|]}{M}
- \frac12 \, \expectation \left[ \, \sum_{x \in \set{X}} \, \biggl| P_{X|Y}(x|Y)
- \frac{\gamma}{M} \biggr| \right].
\end{align}
Let $\gamma \geq 1$, and let $|\set{L}(y)| \leq \frac{M}{\gamma}$ for all
$y \in \set{Y}$. Then, \eqref{LB - variable list size} holds with equality if,
for every $y \in \set{Y}$, the list decoder selects the $|\set{L}(y)|$ most
probable elements in $\set{X}$ given $Y=y$; if $x_\ell(y)$ denotes
the $\ell$-th most probable element in $\set{X}$ given $Y=y$, where ties
in probabilities are resolved arbitrarily, then \eqref{LB - variable list size}
holds with equality if
\begin{align}
\label{02072019a19}
P_{X|Y}(x_\ell(y) \, | y) =
\begin{dcases}
\alpha(y), \quad & \forall \, \ell \in \bigl\{1, \ldots, |\set{L}(y)| \bigr\}, \\
\frac{1-\alpha(y) \, |\set{L}(y)|}{M-|\set{L}(y)|},
\quad & \forall \, \ell \in \bigl\{|\set{L}(y)|+1, \ldots, M\},
\end{dcases}
\end{align}
with $\alpha \colon \set{Y} \to [0,1]$ being an arbitrary function which satisfies
\begin{align}
\label{02072019a20}
\frac{\gamma}{M} \leq \alpha(y) \leq \frac1{|\set{L}(y)|},
\quad \forall \, y \in \set{Y}.
\end{align}
\end{theorem}

\begin{proof}
See Appendix~\ref{appendix: LB - variable list size}.
\end{proof}

\begin{remark}
By setting $\gamma=1$ and $|\set{L}(Y)| = L$ (i.e., a decoding list of fixed size $L$),
\eqref{LB - variable list size} is specialized to \eqref{29062019a5}.
\end{remark}

\begin{example}
\label{example: variable list size}
Let $X$ and $Y$ be random variables taking their values in
$\set{X} = \{0, 1, 2, 3, 4\}$ and $\set{Y} = \{0, 1\}$, respectively,
and let $P_{XY}$ be their joint probability mass function, which is given by
\begin{align}
\label{03072019a1}
\begin{dcases}
P_{XY}(0,0) = P_{XY}(1,0) = P_{XY}(2,0) = \tfrac18, \quad
& P_{XY}(3,0) = P_{XY}(4,0) = \tfrac1{16},  \\[0.1cm]
P_{XY}(0,1) = P_{XY}(1,1) = P_{XY}(2,1) = \tfrac1{24}, \quad
& P_{XY}(3,1) = P_{XY}(4,1) = \tfrac3{16}.
\end{dcases}
\end{align}
Let $\set{L}(0) := \{0,1,2\}$ and $\set{L}(1) := \{3,4\}$ be the lists in $\set{X}$,
given the value of $Y \in \set{Y}$. We get $P_Y(0) = P_Y(1) = \tfrac12$, so the
conditional probability mass function of $X$ given $Y$ satisfies
$P_{X|Y}(x|y) = 2 P_{XY}(x,y)$ for all $(x,y) \in \set{X} \times \set{Y}$.
It can be verified that, if $\gamma = \tfrac54$, then
$\max\{|\set{L}(0)|, |\set{L}(1)|\} = 3 \leq \frac{M}{\gamma}$, and also
\eqref{02072019a19} and \eqref{02072019a20} are satisfied
(here, $M:=|\set{X}|=5$, $\alpha(0) = \tfrac14 = \frac{\gamma}{M}$
and $\alpha(1) = \tfrac38 \in \bigl[\tfrac14, \tfrac12\bigr]$). By
Theorem~\ref{theorem: LB - variable list size}, it follows that
\eqref{LB - variable list size} holds in this case with equality,
and the list decoding error probability is equal to
$P_{\set{L}}=1-\expectation\bigl[ \alpha(Y) \, |\set{L}(Y)| \bigr]=\tfrac14$
(i.e., it coincides with the lower bound in the right side of
\eqref{LB - variable list size} with $\gamma = \tfrac54$).
On the other hand, the generalized Fano's inequality in
\eqref{Fano-Ahlswede-Korner 1} gives that $P_\set{L} \geq 0.1206$
(the left side of \eqref{Fano-Ahlswede-Korner 1} is
$H(X|Y) = \tfrac52 \, \log 2 - \tfrac14 \, \log 3 = 2.1038$~bits);
moreover, by letting $N := \underset{y \in \set{Y}}{\max} \, |\set{L}(y)| = 3$,
\eqref{Fano-Ahlswede-Korner 2} gives the looser bound
$P_\set{L} \geq 0.0939$. This exemplifies a case where the lower bound in
Theorem~\ref{theorem: LB - variable list size} is tight, whereas the
generalized Fano's inequalities in \eqref{Fano-Ahlswede-Korner 1} and
\eqref{Fano-Ahlswede-Korner 2} are looser.
\end{example}

\subsection{A Measure for the Approximation of Equiprobable Distributions
by Tunstall Trees}
\label{subsection: Tunstall}

The best possible approximation of equiprobable distributions, which one
can get by using tree codes has been considered in \cite{CicaleseGV06}.
The optimal solution is obtained by using Tunstall codes, which are
variable-to-fixed lossless compression codes (see
\cite[Section~11.2.3]{Bremaud17}, \cite{Tunstall67}). The main idea behind
Tunstall codes is parsing the source sequence into variable-length segments
of roughly the same probability, and then coding all these segments with
codewords of fixed length. This task is done by assigning the leaves of a
Tunstall tree, which correspond to segments of source symbols with a variable
length (according to the depth of the leaves in the tree), to codewords of
fixed length. The following result links Tunstall trees with majorization theory.

\begin{proposition} \cite[Theorem~1]{CicaleseGV06}
\label{prop: majorization Tunstall}
Let $P_\ell$ be the probability measure generated on the leaves by a Tunstall
tree $\set{T}$, and let $Q_\ell$ be the probability measure generated
by an arbitrary tree $\set{S}$ with the same number of leaves as of
$\set{T}$. Then, $P_\ell \prec Q_\ell$.
\end{proposition}

From Proposition~\ref{prop: majorization Tunstall}, and the
Schur-convexity of an $f$-divergence $D_f(\cdot \| U_n)$
(see \cite[Lemma~1]{CicaleseGV06}),
it follows that (see \cite[Corollary~1]{CicaleseGV06})
\begin{align}
\label{Tunstall-a1}
D_f(P_\ell \| U_n) \leq D_f(Q_\ell \| U_n),
\end{align}
where $n$ designates the joint number of leaves of the trees $\set{T}$
and $\set{S}$.

Before we proceed, it is worth noting that the strong data-processing
inequality in Theorem~\ref{thm: majorization Df} implies that if $f$ is
also twice differentiable, then \eqref{Tunstall-a1} can be strengthened to
\begin{align}
\label{Tunstall-a2}
D_f(P_\ell \| U_n) + n c_f(nq_{\min}, nq_{\max}) \bigl( \| Q_\ell \|_2^2
- \| P_\ell \|_2^2 \bigr) \leq D_f(Q_\ell \| U_n),
\end{align}
where $q_{\max}$ and $q_{\min}$ denote, respectively, the maximal and
minimal positive masses of $Q_\ell$ on the $n$ leaves of a tree
$\set{S}$, and $c_f(\cdot, \cdot)$ is given in \eqref{c_f}.

We next consider a measure which quantifies the quality of the
approximation of the probability mass function $P_\ell$,
induced by the leaves of a Tunstall tree, by an equiprobable
distribution $U_n$ over a set whose cardinality ($n$) is equal
to the number of leaves in the tree. To this end, consider the
setup of Bayesian binary hypothesis testing where a random
variable $X$ has one of the two probability distributions
\begin{align}
\label{BHT}
\begin{dcases}
\mathrm{H}_0: & X \sim P_\ell, \\
\mathrm{H}_1: & X \sim U_n,
\end{dcases}
\end{align}
with a-priori probabilities $\prob[\mathrm{H}_0] = \omega$,
and $\prob[\mathrm{H}_1] = 1-\omega$ for an arbitrary
$\omega \in (0,1)$. The measure being considered here is
equal to the difference between the minimum a-priori and
minimum a-posteriori error probabilities of the Bayesian
binary hypothesis testing model in \eqref{BHT}, which is
close to zero if the two distributions are sufficiently close.

The difference between the minimum a-priori and minimum
a-posteriori error probabilities of a general Bayesian binary
hypothesis testing model with the two arbitrary alternative
hypotheses $\mathrm{H}_0: \, X \sim P$ and
$\mathrm{H}_1: \, X \sim Q$ with a-priori probabilities
$\omega$ and $1-\omega$, respectively, is defined to be
the order-$\omega$ DeGroot statistical information
$\set{I}_\omega(P,Q)$ \cite{DeGroot62} (see also
\cite[Definition~3]{LieseV_IT2006}).
It can be expressed as an $f$-divergence:
\begin{align}
\label{DG1}
\set{I}_\omega(P, Q) = D_{\phi_\omega}(P \| Q),
\end{align}
where $\phi_\omega \colon [0, \infty) \to \Reals$ is the
convex function with $\phi_\omega(1)=0$, given by (see
\cite[(73)]{LieseV_IT2006})
\begin{align}
\label{DG2a}
\phi_\omega(t) := \min\{\omega, 1-\omega\}
- \min\{\omega, 1-\omega t\}, \quad t \geq 0.
\end{align}
The measure considered here for quantifying the closeness
of $P_\ell$ to the equiprobable distribution $U_n$ is
therefore given by
\begin{align}
\label{closeness to U_n}
d_{\omega,n}(P_\ell) := D_{\phi_\omega}(P_\ell \| U_n),
\quad \forall \, \omega \in (0,1),
\end{align}
which is bounded in the interval $\bigl[0, \min\{\omega, 1-\omega\}\bigr]$.

The next result partially relies on Theorem~\ref{thm: LB/UB f-div}.

\begin{theorem}
\label{theorem: closeness to equiprobable}
The measure in \eqref{closeness to U_n} satisfies the
following properties:
\begin{enumerate}[a)]
\item \label{thm: Tunstall-a}
It is the minimum of $D_{\phi_\omega}(P \| U_n)$  with respect
to all probability measures $P \in \set{P}_n$ that are
induced by an arbitrary tree with $n$ leaves.

\item \label{thm: Tunstall-b}
\begin{align}
\label{23062019a3}
d_{\omega,n}(P_\ell) & \leq
\max_{\beta \in \Gamma_n(\rho)} D_{\phi_\omega}(Q_\beta \| U_n),
\end{align}
with the function $\phi_\omega(\cdot)$ in \eqref{DG2a},
the interval $\Gamma_n(\rho)$ in \eqref{Gamma interval},
the probability mass function $Q_\beta$ in \eqref{Q_beta},
and $\rho := \frac1{p_{\min}}$ is the reciprocal of the
minimal probability of the source symbols.

\item \label{thm: Tunstall-c}
The following bound holds for every $n \in \naturals$, which is
the asymptotic limit of the right side of \eqref{23062019a3} as
we let $n \to \infty$:
\begin{align}
\label{23062019a5}
d_{\omega,n}(P_\ell) \leq \max_{x \in [0,1]}
\left\{ x \, \phi_\omega\biggl(\frac{\rho}{1+(\rho-1)x} \biggr)
+ (1-x) \, \phi_\omega\biggl(\frac1{1+(\rho-1)x} \biggr) \right\}.
\end{align}

\item \label{thm: Tunstall-d}
If $f \colon (0, \infty) \to \Reals$ is convex and twice differentiable,
continuous at zero and $f(1)=0$, then
\begin{align}
\label{int_rep.}
D_f(P_\ell \| U_n) = \int_0^1 \frac{d_{\omega,n}(P_\ell)}{\omega^3} \;
f''\left(\frac{1-\omega}{\omega}\right) \, \mathrm{d}\omega.
\end{align}
\end{enumerate}
\end{theorem}

\begin{proof}
See Appendix~\ref{appendix: Tunstall}.
\end{proof}

\begin{remark}
\label{remark: integral repr.}
The integral representation in \eqref{int_rep.} provides another
justification for quantifying the closeness of $P_\ell$ to an
equiprobable distribution by the measure in \eqref{closeness to U_n}.
\end{remark}

Figure~\ref{figure: Tunstall_de_Groot_measure} refers to the upper
bound on the closeness-to-equiprobable measure $d_{\omega,n}(P_\ell)$
in \eqref{23062019a5} for Tunstall trees with $n$ leaves. The bound
holds for all $n \in \naturals$, and it is shown as a
function of $\omega \in [0,1]$ for several values of $\rho \in [1, \infty]$.
In the limit where $\rho \to \infty$, the upper bound is equal to
$\min\{\omega, 1-\omega\}$ since the minimum a-posteriori error
probability of the Bayesian binary hypothesis testing model in \eqref{BHT}
tends to zero. On the other hand, if $\rho=1$, then the right side of
\eqref{23062019a5} is identically equal to zero (since $\phi_\omega(1) = 0$).

\begin{figure}[ht]
\vspace*{-4cm}
\begin{center}
\centerline{\includegraphics[width=10cm]{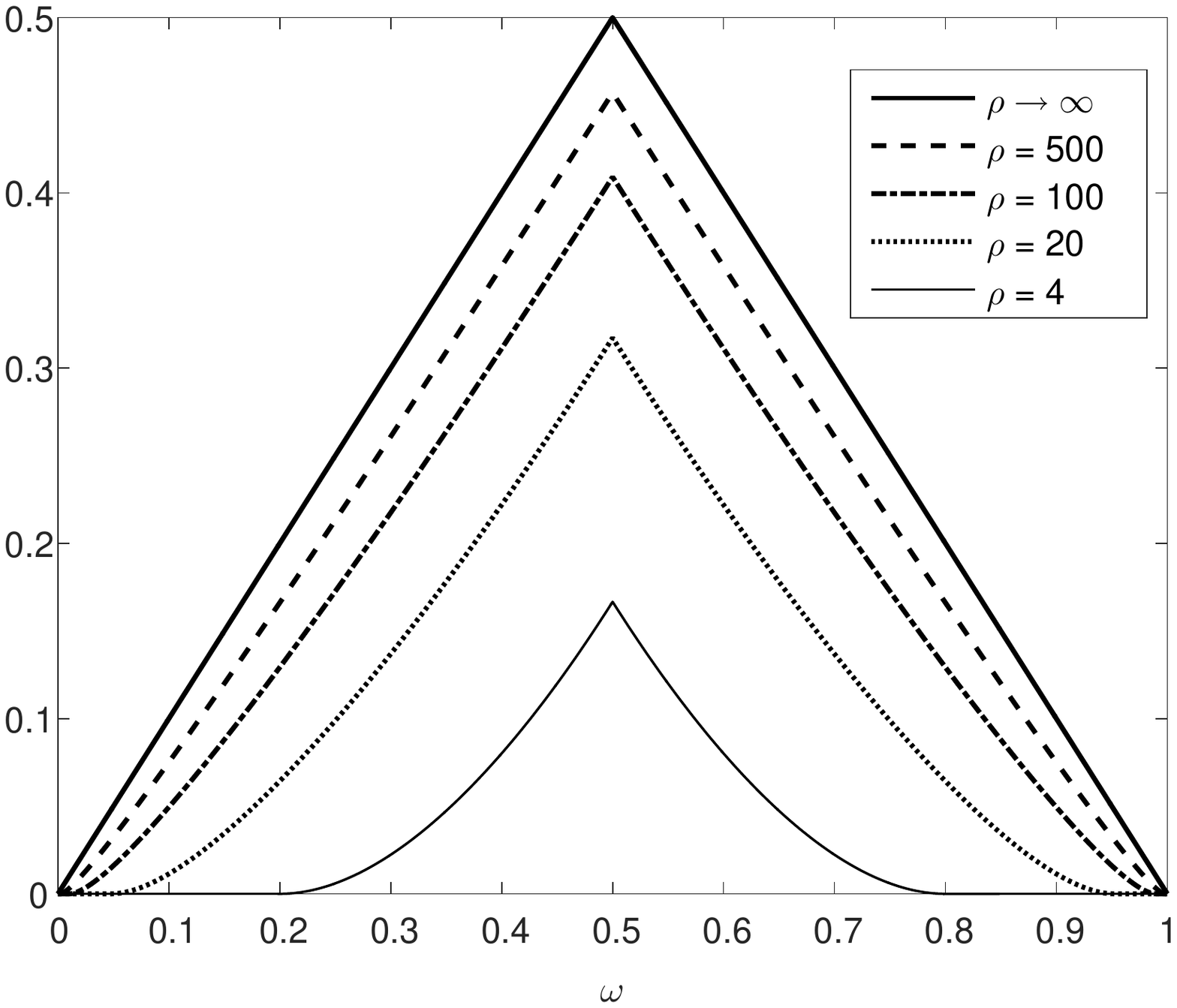}}
\end{center}
\vspace*{-4cm}
\caption{\label{figure: Tunstall_de_Groot_measure}
Curves of the upper bound on the measure $d_{\omega,n}(P_\ell)$
in \eqref{23062019a5}, valid for all $n \in \naturals$, as a
function of $\omega \in [0,1]$ for different values of $\rho := \frac1{p_{\min}}$.}
\end{figure}

Theorem~\ref{theorem: closeness to equiprobable} gives an upper
bound on the measure in \eqref{closeness to U_n}, for the closeness
of the probability mass function generated on the leaves by a Tunstall
tree to the equiprobable distribution, where this bound is expressed
as a function of the minimal probability mass of the source. The
following result, which relies on \cite[Theorem~4]{CicaleseGV18}
and our earlier analysis related to Theorem~\ref{thm: LB/UB f-div},
provides a sufficient condition on the minimal probability mass for
asserting the closeness of the compression rate to the Shannon entropy
of a stationary and memoryless discrete source.
\begin{theorem}
\label{theorem: p_min Tunstall}
Let $P$ be a probability mass function of a stationary and memoryless
discrete source, and let the emitted source symbols be from an alphabet
of size $D \geq 2$. Let $\set{C}$ be a Tunstall code which is used for
source compression; let $m$ and $\set{X}$ denote, respectively, the
fixed length and the alphabet of the codewords of $\set{C}$ (where
$|\set{X}| \geq 2$), referring to a Tunstall tree of $n$ leaves with
$n \leq |\set{X}|^m < n+(D-1)$. Let $p_{\min}$ be the minimal
probability mass of the source symbols, and let
\begin{align}
\label{13072019a1}
d = d(m, \varepsilon) :=
\begin{dcases}
\frac{m \varepsilon \, \log_{\mathrm{e}}|\set{X}|}{1+\varepsilon}
+ \log_{\mathrm{e}}\biggl(1-\frac{D-1}{|\set{X}|^m}\biggr), \quad \mbox{if $D > 2$}, \\
\frac{m \varepsilon \, \log_{\mathrm{e}}|\set{X}|}{1+\varepsilon},
\hspace*{4.1cm} \mbox{if $D=2$,}
\end{dcases}
\end{align}
with an arbitrary $\varepsilon > 0$ such that $d>0$. If
\begin{align}
\label{12072019a2}
p_{\min} \geq
\frac{W_0\bigl(-\mathrm{e}^{-d-1}\bigr)}{W_{-1}\bigl(-\mathrm{e}^{-d-1}\bigr)},
\end{align}
where $W_0$ and $W_{-1}$ denote, respectively, the principal and secondary
real branches of the Lambert $W$ function \cite{Corless96}, then the
compression rate of the Tunstall code is larger than the Shannon entropy
of the source by a factor which is at most $1+\varepsilon$.
\end{theorem}

\begin{proof}
See Appendix~\ref{appendix: Tunstall}.
\end{proof}

\begin{remark}
The condition in \eqref{12072019a2} can be replaced by the stronger
requirement that
\begin{align}
\label{12072019a3}
p_{\min} \geq \frac1{1+\sqrt{8d}}.
\end{align}
However, unless $d$ is a small fraction of unity, there is a significant
difference between the condition in \eqref{12072019a2} and the more
restrictive condition in \eqref{12072019a3} (see Figure~\ref{figure: Tunstall_p_min}).
\end{remark}

\begin{figure}[ht]
\vspace*{-4cm}
\begin{center}
\centerline{\includegraphics[width=9.3cm]{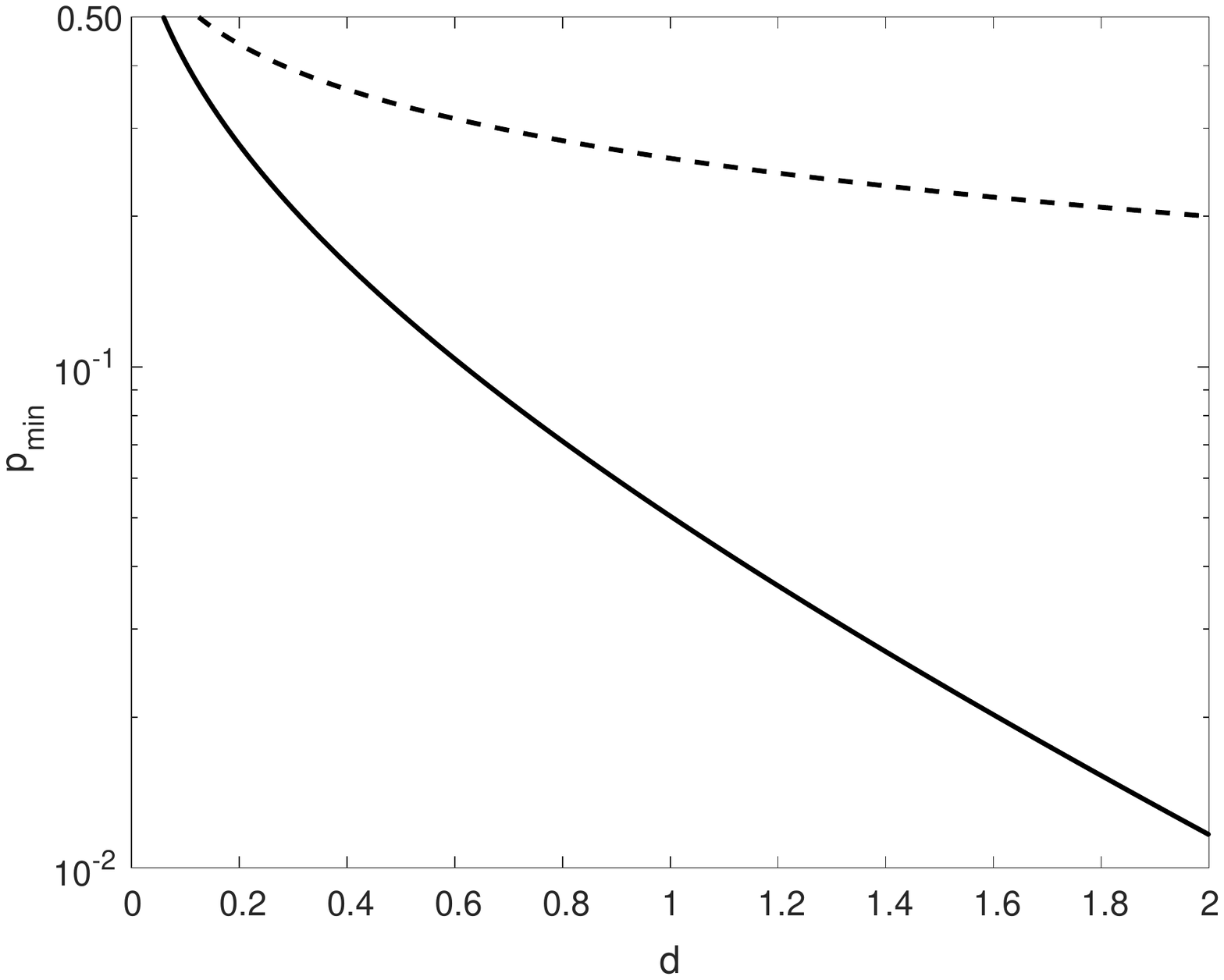}}
\end{center}
\vspace*{-3.8cm}
\caption{\label{figure: Tunstall_p_min}
Curves for the smallest values of $p_{\min}$, in the setup of
Theorem~\ref{theorem: p_min Tunstall}, according to the condition
in \eqref{12072019a2} (solid line) and the more restrictive
condition in \eqref{12072019a3} (dashed line) for binary Tunstall
codes which are used to compress memoryless and stationary binary
sources.}
\end{figure}

\begin{example}
\label{example: Tunstall}
Consider a memoryless and stationary binary source, and a binary Tunstall code
with codewords of length $m=10$ referring to a Tunstall tree with $n=2^m = 1024$
leaves. Letting $\varepsilon = 0.1$ in Theorem~\ref{theorem: p_min Tunstall},
it follows that if the minimal probability mass of the source satisfies
$p_{\min} \geq 0.0978$ (see \eqref{13072019a1}, and Figure~\ref{figure: Tunstall_p_min}
with $d = \frac{m \varepsilon \log_{\mathrm{e}} 2}{1+\varepsilon} = 0.6301$),
then the compression rate of the Tunstall code is at most
$10\%$ larger than the Shannon entropy of the source.
\end{example}

\vspace*{0.1cm}
\subsection*{\bf{Acknowledgments}}
The author wishes to thank the Guest Editor, Amos
Lapidoth, and the two anonymous reviewers for an
efficient process in reviewing and handling this paper.

\break
\appendices

\section{Proof of Theorem~\ref{thm: SDPI-IS}}
\label{appendix: SDPI-IS}

We start by proving Item~\ref{Th. 1.a}). By our assumptions
on $Q_X$ and $W_{Y|X}$,
\begin{align}
\label{P_X/Q_X pos.}
& P_X(x), Q_X(x) > 0, \quad
\hspace*{0.2cm} \forall \, x \in \set{X}, \\
\label{W_Y|X assumption}
& \sum_{x \in \set{X}} W_{Y|X}(y|x) > 0, \quad
\forall \, y \in \set{Y}, \\
\label{sum to 1}
& \sum_{y \in \set{Y}} W_{Y|X}(y|x) = 1, \quad
\hspace*{0.1cm} \forall \, x \in \set{X}, \\
\label{non-neg. W}
& W_{Y|X}(y|x) \geq 0, \quad
\hspace*{0.8cm} \forall \, (x,y) \in \set{X} \times \set{Y}.
\end{align}
From \eqref{transP}, \eqref{transQ}, \eqref{P_X/Q_X pos.},
\eqref{W_Y|X assumption} and \eqref{non-neg. W}, it follows that
\begin{align}
\label{P_Y pos.}
& P_Y(y) = \sum_{x \in \set{X}} P_X(x) W_{Y|X}(y|x) > 0,
\quad \forall \, y \in \set{Y}, \\
\label{Q_Y pos.}
& Q_Y(y) = \sum_{x \in \set{X}} Q_X(x) W_{Y|X}(y|x) > 0,
\quad \forall \, y \in \set{Y},
\end{align}
which imply that, for all $y \in \set{Y}$,
\begin{align}
\label{min/max ratio}
\inf_{x \in \set{X}} \frac{P_X(x)}{Q_X(x)}
\leq \frac{P_Y(y)}{Q_Y(y)}
\leq \sup_{x \in \set{X}} \frac{P_X(x)}{Q_X(x)}.
\end{align}
Since by assumption $P_X$ and $Q_X$ are supported on $\set{X}$, and
$P_Y$ and $Q_Y$ are supported on $\set{Y}$ (see \eqref{P_Y pos.}
and \eqref{Q_Y pos.}), it follows that the left side
inequality in \eqref{min/max ratio} is strict if
the infimum in the left side is equal to~0, and the right
side inequality in \eqref{min/max ratio} is strict if
the supremum in the right side is equal to~$\infty$.
Hence, due to \eqref{xi1}, \eqref{xi2} and \eqref{I_interval},
\begin{align}
\label{range P_Y/Q_Y}
\frac{P_X(x)}{Q_X(x)}, \frac{P_Y(y)}{Q_Y(y)} \in \set{I}(\xi_1, \xi_2),
\quad \forall \, (x,y) \in \set{X} \times \set{Y}.
\end{align}

Since by assumption $f \colon (0, \infty) \to \Reals$ is convex,
it follows that its right derivative $f'_+(\cdot)$ exists, and it is
monotonically non-decreasing and finite on $(0, \infty)$ (see,
e.g., \cite[Theorem~1.2]{RobertsV73} or \cite[Theorem~24.1]{Rockafellar96}).
A straightforward generalization of \cite[Theorem~1.1]{Collet19}
(see \cite[Remark~1]{Collet19}) gives
\begin{align}
\label{eq: Collet}
& D_f(P_X \| Q_X) - D_f(P_Y \| Q_Y)
= \sum_{(x,y) \in \set{X} \times \set{Y}} \left\{ Q_X(x) \, W_{Y|X}(y|x)
\, \Delta\biggl( \frac{P_X(x)}{Q_X(x)}, \frac{P_Y(y)}{Q_Y(y)} \biggr) \right\}
\end{align}
where
\begin{align}
\label{Delta}
\Delta(u,v) := f(u)-f(v)-f'_+(v) (u-v), \quad u,v>0.
\end{align}
In comparison to \cite[Theorem~1.1]{Collet19}, the
requirement that $f$ is differentiable on $(0, \infty)$ is relaxed
here, and the derivative of $f$ is replaced by its right-side
derivative. Note that if $f$ is differentiable, then
$\Delta\Bigl( \tfrac{P_X(x)}{Q_X(x)}, \frac{P_Y(y)}{Q_Y(y)} \Bigr)$
with $\Delta(\cdot, \cdot)$ as defined in \eqref{Delta} is Bregman's
divergence \cite{Bregman67}.
The following equality, expressed in terms of Lebesgue-Stieltjes
integrals, holds by \cite[Theorem~1]{LieseV_IT2006}:
\begin{align}
& \Delta\biggl( \frac{P_X(x)}{Q_X(x)}, \frac{P_Y(y)}{Q_Y(y)} \biggr)
\nonumber \\[0.1cm]
\label{integral1}
& = \begin{dcases}
\int 1\biggl\{ s \in \biggl( \frac{P_Y(y)}{Q_Y(y)}, \frac{P_X(x)}{Q_X(x)}
\biggr] \biggr\}
\, \biggl( \frac{P_X(x)}{Q_X(x)} - s \biggr) \, \mathrm{d}f'_{+}(s),
& \quad \mbox{if $\frac{P_X(x)}{Q_X(x)} \geq \frac{P_Y(y)}{Q_Y(y)}$,} \\[0.2cm]
\int 1\biggl\{ s \in \biggl( \frac{P_X(x)}{Q_X(x)}, \frac{P_Y(y)}{Q_Y(y)}
\biggr] \biggr\}
\, \biggl( s - \frac{P_X(x)}{Q_X(x)} \biggr) \, \mathrm{d}f'_{+}(s),
& \quad \mbox{if $\frac{P_X(x)}{Q_X(x)} < \frac{P_Y(y)}{Q_Y(y)}$.}
\end{dcases}
\end{align}
From \eqref{xi1}, \eqref{xi2}, \eqref{condition on c_f},
\eqref{range P_Y/Q_Y} and \eqref{integral1}, if
$\frac{P_X(x)}{Q_X(x)} \geq \frac{P_Y(y)}{Q_Y(y)}$, then
\begin{align}
\Delta\biggl( \frac{P_X(x)}{Q_X(x)}, \frac{P_Y(y)}{Q_Y(y)} \biggr)
& \geq 2 c_f(\xi_1, \xi_2) \int_{\frac{P_Y(y)}{Q_Y(y)}}^{\frac{P_X(x)}{Q_X(x)}}
\biggl( \frac{P_X(x)}{Q_X(x)} - s \biggr) \, \mathrm{d}s \nonumber \\
\label{integral2}
& = c_f(\xi_1, \xi_2) \biggl(\frac{P_X(x)}{Q_X(x)} - \frac{P_Y(y)}{Q_Y(y)} \biggr)^2,
\end{align}
and similarly, if $\frac{P_X(x)}{Q_X(x)} < \frac{P_Y(y)}{Q_Y(y)}$, then
\begin{align}
\label{integral3}
\Delta\biggl( \frac{P_X(x)}{Q_X(x)}, \frac{P_Y(y)}{Q_Y(y)} \biggr)
& \geq 2 c_f(\xi_1, \xi_2) \int_{\frac{P_X(x)}{Q_X(x)}}^{\frac{P_Y(y)}{Q_Y(y)}}
\biggl( s - \frac{P_X(x)}{Q_X(x)} \biggr) \, \mathrm{d}s \nonumber \\
& = c_f(\xi_1, \xi_2) \biggl(\frac{P_X(x)}{Q_X(x)} - \frac{P_Y(y)}{Q_Y(y)} \biggr)^2.
\end{align}
By combining \eqref{eq: Collet}, \eqref{integral2} and \eqref{integral3}, it follows that
\begin{align}
& D_f(P_X \| Q_X) - D_f(P_Y \| Q_Y) \nonumber \\
\label{diff f-div1}
& \geq c_f(\xi_1, \xi_2)
\sum_{(x,y) \in \set{X} \times \set{Y}} \biggl\{ Q_X(x) \, W_{Y|X}(y|x)
\, \biggl(\frac{P_X(x)}{Q_X(x)} - \frac{P_Y(y)}{Q_Y(y)} \biggr)^2 \biggr\},
\end{align}
and an evaluation of the sum in the right side of \eqref{diff f-div1} gives (see
\eqref{transP}, \eqref{transQ} and \eqref{sum to 1})
\begin{align}
& \sum_{(x,y) \in \set{X} \times \set{Y}} \biggl\{ Q_X(x) \, W_{Y|X}(y|x)
\, \biggl(\frac{P_X(x)}{Q_X(x)} - \frac{P_Y(y)}{Q_Y(y)} \biggr)^2 \biggr\} \nonumber \\
& = \sum_{x \in \set{X}} \biggl\{\frac{P_X^2(x)}{Q_X(x)} \,
\underbrace{\sum_{y \in \set{Y}} W_{Y|X}(y|x)}_{=1} \biggr\}
-2 \sum_{y \in \set{Y}} \biggl\{ \frac{P_Y(y)}{Q_Y(y)} \,
\underbrace{\sum_{x \in \set{X}} P_X(x) W_{Y|X}(y|x)}_{=P_Y(y)} \biggr\} \nonumber \\
& \hspace*{0.4cm} + \sum_{y \in \set{Y}} \biggl\{ \frac{P_Y^2(y)}{Q_Y^2(y)} \,
\underbrace{\sum_{x \in \set{X}} Q_X(x) W_{Y|X}(y|x)}_{=Q_Y(y)} \biggr\} \\
\label{27062019a17}
& = \sum_{x \in \set{X}} \frac{P_X^2(x)}{Q_X(x)}
- \sum_{y \in \set{Y}} \frac{P_Y^2(y)}{Q_Y(y)} \\
\label{27062019a18}
& = \sum_{x \in \set{X}} \frac{\bigl(P_X(x)-Q_X(x)\bigr)^2}{Q_X(x)}
- \sum_{y \in \set{Y}} \frac{\bigl(P_Y(y)-Q_Y(y)\bigr)^2}{Q_Y(y)} \\
\label{diff of chi^2 div.}
& = \chi^2(P_X \| Q_X) - \chi^2(P_Y \| Q_Y).
\end{align}
Combining \eqref{diff f-div1}--\eqref{diff of chi^2 div.} gives
\eqref{key}; \eqref{DPI1} is due to the data-processing
inequality for $f$-divergences (applied to the $\chi^2$-divergence),
and the non-negativity of $c_f(\xi_1, \xi_2)$ in \eqref{condition on c_f}.

The $\chi^2$-divergence is an $f$-divergence
with $f(t) = (t-1)^2$ for $t \geq 0$. The condition
in \eqref{condition on c_f} allows to set here
$c_f(\xi_1, \xi_2) \equiv 1$, implying that \eqref{key}
holds in this case with equality.

We next prove Item~\ref{Th. 1.b}). Let $f$ be twice differentiable
on $\set{I}:=\set{I}(\xi_1, \xi_2)$ (see \eqref{I_interval}), and
let $(u, v) \in \set{I} \times \set{I}$ with $v>u$.
Dividing both sides of \eqref{condition on c_f} by $v-u$, and letting
$v \to u^+$, yields $c_f(\xi_1, \xi_2) \leq \tfrac12 \, f''(u)$. Since
this holds for all $u \in \set{I}$, it follows that
$c_f(\xi_1, \xi_2) \leq \tfrac12 \, \underset{t \in \set{I}}{\inf} f''(t)$.
We next show that $c_f(\xi_1, \xi_2)$ in \eqref{c_f} fulfills the
condition in \eqref{condition on c_f}, and therefore it is the largest
possible value of $c_f$ to satisfy \eqref{condition on c_f}.
By the mean value theorem of Lagrange, for all $(u, v) \in \set{I} \times \set{I}$
with $v>u$, there exists an intermediate value $\xi \in (u,v)$ such that
$f'(v) - f'(u) = f''(\xi) \, (v-u)$; hence,
$$f'(v) - f'(u) \geq 2 c_f(\xi_1, \xi_2) \, (v-u),$$
so the condition in \eqref{condition on c_f} is indeed fulfilled
with $c_f := c_f(\xi_1, \xi_2)$ as given in \eqref{c_f}.

We next prove Item~\ref{Th. 1.c}). Let $f^\ast \colon (0, \infty) \to \Reals$
be the dual convex function which is given by $f^\ast(t) := t f\bigl(\tfrac1t)$
for all $t>0$ with $f^\ast(1) = f(1) = 0$. Since $P_X$, $P_Y$, $Q_X$ and $Q_Y$
are supported on $\set{X}$ (see \eqref{P_Y pos.} and \eqref{Q_Y pos.}),
we have
\begin{align}
\label{swap1}
& D_f(P_X \| Q_X) = D_{f^\ast}(Q_X \| P_X), \\
\label{swap2}
& D_f(P_Y \| Q_Y) = D_{f^\ast}(Q_Y \| P_Y), \\
\label{swap3}
& \xi_1^\ast := \inf_{x \in \set{X}} \frac{Q_X(x)}{P_X(x)}
= \biggl(\sup_{x \in \set{X}} \frac{P_X(x)}{Q_X(x)} \biggr)^{-1}
= \frac1{\xi_2}, \\[0.1cm]
\label{swap4}
& \xi_2^\ast := \sup_{x \in \set{X}} \frac{Q_X(x)}{P_X(x)}
= \biggl(\inf_{x \in \set{X}} \frac{P_X(x)}{Q_X(x)} \biggr)^{-1}
= \frac1{\xi_1}.
\end{align}
Consequently, it follows that
\begin{align}
D_f(P_X \| Q_X) - D_f(P_Y \| Q_Y)
\label{swap5}
& = D_{f^\ast}(Q_X \| P_X) - D_{f^\ast}(Q_Y \| P_Y) \\
\label{swap6}
& \geq c_{f^\ast}(\xi_1^\ast, \xi_2^\ast)
\bigl[ \chi^2(Q_X \| P_X) - \chi^2(Q_Y \| P_Y) \bigr] \\
\label{swap7}
& = c_{f^\ast}\bigl(\tfrac1{\xi_2}, \tfrac1{\xi_1}\bigr)
\bigl[ \chi^2(Q_X \| P_X) - \chi^2(Q_Y \| P_Y) \bigr]
\end{align}
where \eqref{swap5} holds due to \eqref{swap1} and \eqref{swap2};
\eqref{swap6} follows from \eqref{key} with $f$, $P_X$ and $Q_X$
replaced by $f^\ast$, $Q_X$ and $P_X$, respectively, which then
implies that $\xi_1$ and $\xi_2$ in \eqref{xi1} and \eqref{xi2}
are, respectively, replaced by $\xi_1^\ast$ and $\xi_2^\ast$
in \eqref{swap3} and \eqref{swap4}; finally, \eqref{swap7}
holds due to \eqref{swap3} and \eqref{swap4}. Since by assumption
$f$ is twice differentiable on $(0, \infty)$, so is $f^\ast$, and
\begin{align}
\label{2nd der. f^ast}
(f^\ast)''(t) = \frac1{t^3} \, f\biggl(\frac1t\biggr), \quad t>0.
\end{align}
Hence,
\begin{align}
\label{f to f^ast1}
c_{f^\ast}\bigl(\tfrac1{\xi_2}, \tfrac1{\xi_1}\bigr)
& = \tfrac12 \inf_{u \in \set{I}\bigl(\tfrac1{\xi_2}, \tfrac1{\xi_1} \bigr)}
(f^\ast)''(u) \\
\label{f to f^ast2}
& = \tfrac12 \inf_{u \in \set{I}\bigl(\tfrac1{\xi_2}, \tfrac1{\xi_1} \bigr)}
\left\{ \left(\frac1{u}\right)^3 \, f\biggl(\frac1u\biggr) \right\} \\
\label{f to f^ast3}
&= \tfrac12 \inf_{t \in \set{I}(\xi_1, \xi_2)} \bigl\{ t^3 f(t) \bigr\}
\end{align}
where \eqref{f to f^ast1} follows from \eqref{key} with $f$, $\xi_1$
and $\xi_2$ replaced by $f^\ast$, $\tfrac1{\xi_2}$ and $\tfrac1{\xi_1}$,
respectively; \eqref{f to f^ast2} holds due to \eqref{2nd der. f^ast}, and
\eqref{f to f^ast3} holds by substituting $t =: \tfrac1u$. This
proves \eqref{key-dual} and \eqref{c_f^*}, where \eqref{DPI2} is
due to the data-processing inequality for $f$-divergences, and the
non-negativity of $c_{f^\ast}(\cdot, \cdot)$.

Similarly to the condition for equality in \eqref{key}, equality
in \eqref{key-dual} is satisfied if $f^*(t) = (t-1)^2$ for all $t>0$,
or equivalently $f(t) = t f^\ast\bigl(\tfrac1t) = \tfrac{(t-1)^2}{t}$ for all
$t>0$. This $f$-divergence is Neyman's $\chi^2$-divergence where
$D_f(P\|Q) := \chi^2(Q\|P)$ for all $P$ and $Q$ with $c_{f^\ast} \equiv 1$
(due to \eqref{c_f^*}, and since $t^3 f''(t)=2$ for all $t>0$).

The proof of Item~\ref{Th. 1.d}) follows that same lines as the proof of
Items~\ref{Th. 1.a})--\ref{Th. 1.c}) by replacing the condition
in \eqref{condition on c_f} with a complementary condition of the form
\begin{align}
\label{condition on e_f}
f'_{+}(v) - f'_{+}(u) \leq 2 e_f(\xi_1, \xi_2) \; (v-u), \quad
\forall \, u,v \in \set{I}(\xi_1, \xi_2), \; u<v.
\end{align}

We finally prove Item~\ref{Th. 1.e}) by showing that
the lower and upper bounds in \eqref{key}, \eqref{key-dual}, \eqref{UB1}
and \eqref{UB2} are locally tight. More precisely, let $\{P_X^{(n)}\}$ be
a sequence of probability mass functions defined on $\set{X}$ and pointwise
converging to $Q_X$ which is supported on $\set{X}$, let $P_Y^{(n)}$
and $Q_Y$ be the probability mass functions defined on $\set{Y}$ via
\eqref{transP} and \eqref{transQ} with inputs $P_X^{(n)}$ and $Q_X$,
respectively, and let $\{\xi_{1,n}\}$ and $\{\xi_{2,n}\}$ be defined,
respectively, by \eqref{xi1} and \eqref{xi2} with $P_X$ being replaced by
$P_X^{(n)}$. By the assumptions in \eqref{13062019a1} and \eqref{13062019a2},
\begin{align}
\label{lim xi_1,n}
& \lim_{n \to \infty} \xi_{1,n} = \lim_{n \to \infty} \inf_{x \in \set{X}}
\frac{P_X^{(n)}(x)}{Q_X(x)} = 1, \\
\label{lim xi_2,n}
& \lim_{n \to \infty} \xi_{2,n} = \lim_{n \to \infty} \sup_{x \in \set{X}}
\frac{P_X^{(n)}(x)}{Q_X(x)} = 1.
\end{align}
Consequently, if $f$ has a continuous second derivative at unity, then
\eqref{key}, \eqref{c_f}, \eqref{e_f}, \eqref{UB1}, \eqref{lim xi_1,n}
and \eqref{lim xi_2,n} imply that
\begin{align}
& \lim_{n \to \infty} \frac{D_f(P_X^{(n)} \| Q_X) - D_f(P_Y^{(n)} \| Q_Y)}
{\chi^2(P_X^{(n)} \| Q_X) - \chi^2(P_Y^{(n)} \| Q_Y)} \nonumber \\[0.1cm]
& = \lim_{n \to \infty} c_f(\xi_{1,n}, \xi_{2,n})
= \lim_{n \to \infty} e_f(\xi_{1,n}, \xi_{2,n}) = \tfrac12 f''(1),
\end{align}
and similarly, from \eqref{key-dual}, \eqref{c_f^*}, \eqref{UB2}, \eqref{e_f^*},
\eqref{lim xi_1,n} and \eqref{lim xi_2,n},
\begin{align}
& \lim_{n \to \infty} \frac{D_f(P_X^{(n)} \| Q_X) - D_f(P_Y^{(n)} \| Q_Y)}
{\chi^2(Q_X \| P_X^{(n)}) - \chi^2(Q_Y \| P_Y^{(n)})} \nonumber \\[0.2cm]
& = \lim_{n \to \infty} c_{f^\ast}\biggl(\frac1{\xi_{2,n}}, \frac1{\xi_{1,n}} \biggr)
= \lim_{n \to \infty} e_{f^\ast}\biggl(\frac1{\xi_{2,n}}, \frac1{\xi_{1,n}} \biggr)
= \tfrac12 f''(1),
\end{align}
which, respectively, prove \eqref{tight} and \eqref{tight-dual}.

\section{Proof of Theorem~\ref{Thm: DMS-DMC}}
\label{appendix: DMS-DMC}

We start by proving Item~\ref{Th. 2.a.}). By the assumption that $P_{X_i}$
and $Q_{X_i}$ are supported on $\set{X}$ for all $i \in \{1, \ldots, n\}$,
it follows from \eqref{2DMS} that the probability mass functions $P_{X^n}$
and $Q_{X^n}$ are supported on $\set{X}^n$. Consequently, from \eqref{prod. RX},
also $R_{X^n}^{(\lambda)}$ is supported on $\set{X}^n$ for all $\lambda \in [0,1]$.
Due to the product forms of $Q_{X^n}$ and $R_{X^n}^{(\lambda)}$ in \eqref{2DMS}
and \eqref{prod. RX}, respectively, we get from \eqref{xi1_n} that
\begin{align}
\xi_1(n, \lambda) &= \prod_{i=1}^n \left(1 - \lambda +
\lambda \, \inf_{x \in \set{X}} \frac{P_{X_i}(x)}{Q_{X_i}(x)}\right) \nonumber \\
&= \prod_{i=1}^n \left( \inf_{x \in \set{X}}
\frac{\lambda P_{X_i}(x) + (1-\lambda) Q_{X_i}(x)}{Q_{X_i}(x)} \right) \nonumber \\
&= \inf_{\underline{x} \in \set{X}^n} \left\{
\frac{\underset{i=1}{\overset{n}{\prod}} \bigl(\lambda P_{X_i}(x_i)
+ (1-\lambda) Q_{X_i}(x_i) \bigr)}{\underset{i=1}{\overset{n}{\prod}}
Q_{X_i}(x_i)} \right\} \nonumber \\
\label{xi_1,n DMS}
&= \inf_{\underline{x} \in \set{X}^n} \frac{R_{X^n}^{(\lambda)}(\underline{x})}
{Q_{X^n}(\underline{x})} \in (0, 1],
\end{align}
and likewise, from \eqref{xi2_n},
\begin{align}
\label{xi_2,n DMS}
\xi_2(n, \lambda)
&= \sup_{\underline{x} \in \set{X}^n} \frac{R_{X^n}^{(\lambda)}(\underline{x})}
{Q_{X^n}(\underline{x})} \in [1, \infty)
\end{align}
for $\lambda \in [0,1]$.
In view of \eqref{key}, \eqref{c_f}, \eqref{xi_1,n DMS} and
\eqref{xi_2,n DMS}, replacing
$(P_X, \, P_Y, \, Q_X, \, Q_Y, \, \xi_1, \, \xi_2)$ in
\eqref{key} and \eqref{c_f} with
$(R_{X^n}^{(\lambda)}, \, R_{Y^n}^{(\lambda)}, \, Q_{X^n},
\, Q_{Y^n}, \, \xi_1(n,\lambda), \, \xi_2(n,\lambda)),$ we
obtain that, for all $\lambda \in [0,1]$,
\begin{align}
\label{key: modified}
& D_f(R_{X^n}^{(\lambda)} \, \| \, Q_{X^n})
- D_f(R_{Y^n}^{(\lambda)} \, \| \, Q_{Y^n}) \nonumber \\
&\geq c_f\bigl(\xi_1(n,\lambda), \xi_2(n,\lambda)\bigr)
\left[ \chi^2(R_{X^n}^{(\lambda)} \, \| \, Q_{X^n})
- \chi^2(R_{Y^n}^{(\lambda)} \, \| \, Q_{Y^n}) \right].
\end{align}
Due to the setting in \eqref{2DMS}--\eqref{MC3 in DMC},
for all $\underline{y} \in \set{Y}^n$ and $\lambda \in [0,1]$,
\begin{align}
R_{Y^n}^{(\lambda)}(\underline{y})
&= \sum_{\underline{x} \in \set{X}^n} R_{X^n}^{(\lambda)}(\underline{x})
\, W_{Y^n | X^n}(\underline{y} | \underline{x}) \nonumber \\
&= \sum_{\underline{x} \in \set{X}^n} \left\{ \prod_{i=1}^n
\bigl( \lambda P_{X_i}(x_i) + (1-\lambda) Q_{X_i}(x_i) \bigr) \,
\prod_{i=1}^n W_{Y_i|X_i}(y_i | x_i) \right\} \nonumber \\
&= \prod_{i=1}^n \left\{ \sum_{x_i \in \set{X}} \Bigl\{
\bigl( \lambda P_{X_i}(x_i) + (1-\lambda) Q_{X_i}(x_i) \bigr) \,
W_{Y_i|X_i}(y_i | x_i) \Bigr\} \right\} \nonumber \\
&= \prod_{i=1}^n \left\{ \lambda \sum_{x \in \set{X}} P_{X_i}(x)
W_{Y_i|X_i}(y_i | x) + (1-\lambda) \sum_{x \in \set{X}}
Q_{X_i}(x) W_{Y_i|X_i}(y_i | x)  \right\} \nonumber \\
\label{prod RY}
&= \prod_{i=1}^n \bigl( \lambda P_{Y_i}(y_i)
+ (1-\lambda) Q_{Y_i}(y_i) \bigr) \nonumber \\
&= \prod_{i=1}^n R_{Y_i}^{(\lambda)}(y_i)
\end{align}
with
\begin{align}
\label{RY}
R_{Y_i}^{(\lambda)}(y) := \lambda P_{Y_i}(y) + (1-\lambda) Q_{Y_i}(y), \quad
\forall \, i \in \{1, \ldots, n\}, \; y \in \set{Y}, \; \lambda \in [0,1],
\end{align}
and $R_{Y_i}^{(\lambda)}$ is the probability mass function at the channel
output at time instant $i$. In particular, setting $\lambda=0$ in \eqref{prod RY}
gives
\begin{align}
\label{prod QY}
Q_{Y^n}(\underline{y}) = \prod_{i=1}^n Q_{Y_i}(y_i), \quad \forall \,
\underline{y} \in \set{Y}^n.
\end{align}
Due to the tensorization property of the $\chi^2$ divergence,
since $R_{X^n}^{(\lambda)}$, $R_{Y^n}^{(\lambda)}$, $Q_{X^n}$ and $Q_{Y^n}$
are product probability measures (see \eqref{2DMS}, \eqref{prod. RX},
\eqref{prod RY} and \eqref{prod QY}), it follows that
\begin{align}
\label{1st tensorize chi2}
\chi^2(R_{X^n}^{(\lambda)} \, \| \, Q_{X^n}) = \prod_{i=1}^n
\Bigl( 1+\chi^2(R_{X_i}^{(\lambda)} \, \| \, Q_{X_i} \bigr) \Bigr) - 1,
\end{align}
and
\begin{align}
\label{2nd tensorize chi2}
\chi^2(R_{Y^n}^{(\lambda)} \, \| \, Q_{Y^n}) = \prod_{i=1}^n
\Bigl( 1+\chi^2(R_{Y_i}^{(\lambda)} \, \| \, Q_{Y_i} \bigr) \Bigr) - 1.
\end{align}
Substituting \eqref{1st tensorize chi2} and \eqref{2nd tensorize chi2} into
the right side of \eqref{key: modified} gives that, for all $\lambda \in [0,1]$,
\begin{align}
& D_f(R_{X^n}^{(\lambda)} \, \| \, Q_{X^n})
- D_f(R_{Y^n}^{(\lambda)} \, \| \, Q_{Y^n}) \label{eq: diff. Df DMS-DMC} \\
\nonumber
&\geq c_f\bigl(\xi_1(n,\lambda), \xi_2(n,\lambda)\bigr)
\left[ \, \prod_{i=1}^n \Bigl( 1+\chi^2(R_{X_i}^{(\lambda)} \, \| \, Q_{X_i} \bigr) \Bigr)
- \prod_{i=1}^n \Bigl( 1+\chi^2(R_{Y_i}^{(\lambda)} \, \| \, Q_{Y_i} \bigr) \Bigr) \right].
\end{align}
Due to \eqref{prod. RX} and \eqref{RY}, since
\begin{align}
& R_{X_i}^{(\lambda)} = \lambda P_{X_i} + (1-\lambda) Q_{X_i}, \\
& R_{Y_i}^{(\lambda)} = \lambda P_{Y_i} + (1-\lambda) Q_{Y_i},
\end{align}
and (see \cite[Lemma~5]{Sason18})
\begin{align}
\chi^2(\lambda P + (1-\lambda)Q \, \| \, Q)
= \lambda^2 \, \chi^2(P\|Q), \quad \forall \, \lambda \in [0,1]
\end{align}
for every pair of probability measures $(P,Q)$, it follows that
\begin{align}
\label{mixture1 chi^2}
& \chi^2(R_{X_i}^{(\lambda)} \, \| \, Q_{X_i} \bigr)
= \lambda^2 \, \chi^2(P_{X_i} \, \| \, Q_{X_i}), \\
\label{mixture2 chi^2}
& \chi^2(R_{Y_i}^{(\lambda)} \, \| \, Q_{Y_i} \bigr)
= \lambda^2 \, \chi^2(P_{Y_i} \, \| \, Q_{Y_i}).
\end{align}
Substituting \eqref{mixture1 chi^2} and \eqref{mixture2 chi^2} into the right side
of \eqref{eq: diff. Df DMS-DMC} gives \eqref{LB1 - DMC}.
For proving the looser bound \eqref{LB2 - DMC} from \eqref{LB1 - DMC}, and also for
later proving the result in Item~\ref{Th. 2.c.}), we rely on the following lemma.

\begin{lemma}
\label{lemma 1}
Let $\{a_i\}_{i=1}^n$ and $\{b_i\}_{i=1}^n$ be non-negative with $a_i \geq b_i$
for all $i \in \{1, \ldots, n\}$. Then,
\begin{enumerate}[a)]
\item \label{lemma 1.a}
For all $u \geq 0$,
\begin{align}
\label{lemma1-1}
\prod_{i=1}^n (1+a_i u) - \prod_{i=1}^n (1+b_i u) \geq \sum_{i=1}^n (a_i - b_i) u.
\end{align}
\item \label{lemma 1.b}
If $a_i > b_i$ for at least one index $i$, then
\begin{align}
\label{lemma1-2}
\prod_{i=1}^n (1+a_i u) - \prod_{i=1}^n (1+b_i u) = \sum_{i=1}^n (a_i - b_i) u + O(u^2).
\end{align}
\end{enumerate}
\end{lemma}
\begin{proof}
Let $g \colon [0, \infty) \to \Reals$ be defined as
\begin{align}
\label{aux. function g}
g(u) := \prod_{i=1}^n (1+a_i u) - \prod_{i=1}^n (1+b_i u), \quad \forall \, u \geq 0.
\end{align}
We have $g(0)=0$, and the first two derivatives of $g$ are given by
\begin{align}
\label{1st der. g}
g'(u) = \sum_{i=1}^n \Bigl\{ a_i \prod_{j \neq i} (1+a_j u)
- b_i \prod_{j \neq i} (1+b_j u) \Bigr\},
\end{align}
and
\begin{align}
\label{2nd der. g}
g''(u) = \sum_{i=1}^n \sum_{j \neq i} \Bigl\{ a_i a_j \prod_{k \neq i,j} (1+a_k u)
- b_i b_j \prod_{k \neq i,j} (1+b_k u) \Bigr\}.
\end{align}
Since by assumption $a_i \geq b_i \geq 0$ for all $i$, it follows from \eqref{2nd der. g}
that $g''(u) \geq 0$ for all $u \geq 0$, which asserts the convexity of $g$ on $[0, \infty)$.
Hence, for all $u \geq 0$,
\begin{align}
\label{LB g}
g(u) \geq g(0) + g'(0) u
= \sum_{i=1}^n (b_i - a_i) u
\end{align}
where the right-side equality in \eqref{LB g} is due to \eqref{aux. function g} and \eqref{1st der. g}.
This gives \eqref{lemma1-1}.

We next prove Item~\ref{lemma 1.b}) of Lemma~\ref{lemma 1}. By the Taylor series expansion of
the polynomial function $g$, we get
\begin{align}
g(u) &= g(0) + g'(0) u + \tfrac12 g''(0) u^2 + \ldots \nonumber \\
\label{Taylor with 2nd-order term}
&= \sum_{i=1}^n (b_i - a_i) u + \tfrac12 \sum_{i=1}^n \sum_{j \neq i} (a_i a_j - b_i b_j) u^2 + \ldots
\end{align}
for all $u \geq 0$. Since by assumption $a_i \geq b_i \geq 0$ for all $i$, and there exists an index
$i \in \{1, \ldots, n\}$ such that $a_i > b_i$, it follows that the coefficient of $u^2$ in the right
side of \eqref{Taylor with 2nd-order term} is positive. This yields \eqref{lemma1-2}.
\end{proof}

We obtain here \eqref{LB2 - DMC} from \eqref{LB1 - DMC} and Item~\ref{lemma 1.a}) of Lemma~\ref{lemma 1}.
To that end, for $i \in \{1, \ldots, n\}$, let
\begin{align}
\label{a,b,u}
a_i := \chi^2(P_{X_i} \| Q_{X_i}), \quad b_i := \chi^2(P_{Y_i} \| Q_{Y_i}), \quad u := \lambda^2
\end{align}
with $u \in [0,1]$ for every $\lambda \in [0,1]$. Since by \eqref{2DMS}, \eqref{DMC}, \eqref{MC2 in DMC}
and \eqref{MC3 in DMC},
\begin{align}
& P_{X_i} \to W_{Y_i | X_i} \to P_{Y_i}, \\
& Q_{X_i} \to W_{Y_i | X_i} \to Q_{Y_i},
\end{align}
it follows from the data-processing inequality for $f$-divergences, and their non-negativity, that
\begin{align}
\label{a>=b>=0}
a_i \geq b_i \geq 0, \quad \forall \, i \in \{1, \ldots, n\},
\end{align}
which yields \eqref{LB2 - DMC} from \eqref{LB1 - DMC}, \eqref{lemma1-1}, \eqref{a,b,u} and \eqref{a>=b>=0}.

We next prove Item~\ref{Th. 2.b.}) of Theorem~\ref{Thm: DMS-DMC}.
Similarly to the proof of \eqref{key: modified}, we get from \eqref{UB1}
(rather than \eqref{key}) that
\begin{align}
\label{key: modified2}
& D_f(R_{X^n}^{(\lambda)} \, \| \, Q_{X^n}) - D_f(R_{Y^n}^{(\lambda)} \, \| \, Q_{Y^n}) \nonumber \\
&\leq e_f\bigl(\xi_1(n,\lambda), \xi_2(n,\lambda)\bigr)
\left[ \chi^2(R_{X^n}^{(\lambda)} \, \| \, Q_{X^n}) - \chi^2(R_{Y^n}^{(\lambda)} \, \| \, Q_{Y^n}) \right].
\end{align}
Combining \eqref{1st tensorize chi2}, \eqref{2nd tensorize chi2}, \eqref{mixture1 chi^2},
\eqref{mixture2 chi^2} and \eqref{key: modified2} gives \eqref{UB1 - DMC}.

We finally prove Item~\ref{Th. 2.c.}) of Theorem~\ref{Thm: DMS-DMC}.
In view of \eqref{xi1_n} and \eqref{xi2_n}, and by the assumption that
$\underset{x \in \set{X}}{\sup} \frac{P_{X_i}(x)}{Q_{X_i}(x)} < \infty$
for all $i \in \{1, \ldots, n\}$, we get
\begin{align}
\label{lim1}
& \lim_{\lambda \to 0^+} \xi_1(n, \lambda) = 1, \\
\label{lim2}
& \lim_{\lambda \to 0^+} \xi_2(n, \lambda) = 1.
\end{align}
Since, by assumption $f$ has a continuous second derivative at unity,
\eqref{c_f}, \eqref{e_f}, \eqref{lim1} and \eqref{lim2} imply that
\begin{align}
\label{lim3}
& \lim_{\lambda \to 0^+} c_f\bigl(\xi_1(n, \lambda), \xi_2(n, \lambda)\bigr) = \tfrac12 f''(1), \\
\label{lim4}
& \lim_{\lambda \to 0^+} e_f\bigl(\xi_1(n, \lambda), \xi_2(n, \lambda)\bigr) = \tfrac12 f''(1).
\end{align}
From \eqref{a,b,u}, \eqref{a>=b>=0}, and Item~\ref{lemma 1.b})
of Lemma~\ref{lemma 1}, it follows that
\begin{align}
& \lim_{\lambda \to 0^+} \, \frac1{\lambda^2}
\left[ \, \prod_{i=1}^n \bigl(1 + \lambda^2 \, \chi^2(P_{X_i} \| Q_{X_i}) \bigr)
- \prod_{i=1}^n \bigl(1 + \lambda^2 \, \chi^2(P_{Y_i} \| Q_{Y_i}) \bigr) \right] \nonumber \\
\label{lim5}
&= \sum_{i=1}^n \bigl[ \chi^2(P_{X_i} \| Q_{X_i}) - \chi^2(P_{Y_i} \| Q_{Y_i}) \bigr].
\end{align}
The result in \eqref{lim - DMC} finally follows from \eqref{LB1 - DMC},
\eqref{UB1 - DMC} and \eqref{lim3}--\eqref{lim5}. This indeed shows
that the lower bounds in the right sides of \eqref{LB1 - DMC} and
\eqref{LB2 - DMC}, and the upper bound in the right side of \eqref{UB1 - DMC}
yield a tight result as we let $\lambda \to 0^+$, leading to the limit
in the right side of \eqref{lim - DMC}.

\section{Proof of Theorems~\ref{theorem: contraction coef} and~\ref{theorem: DMS/DMC - ver2}}
\label{appendix: contraction coef.}

\subsection{Proof of Theorem~\ref{theorem: contraction coef}}
We first obtain a lower bound on $D_f(P_X \| Q_X)$, and then obtain an upper bound
on $D_f(P_Y \| Q_Y)$.
\begin{align}
\label{13062019b-1}
D_f(P_X \| Q_X) &=
\sum_{x \in \set{X}} Q_X(x) \, f\hspace*{-0.1cm}\left(\frac{P_X(x)}{Q_X(x)}\right) \\
\label{13062019b0}
&= \sum_{x \in \set{X}} Q_X(x) \, \left[ \frac{P_X(x)}{Q_X(x)} \;
g\hspace*{-0.1cm}\left(\frac{P_X(x)}{Q_X(x)}\right) + f(0) \right] \\
\label{13062019b1}
&= f(0) + \sum_{x \in \set{X}} P_X(x) \; g\hspace*{-0.1cm}\left(\frac{P_X(x)}{Q_X(x)}\right) \\
\label{13062019b1.2}
&\geq f(0) + g\left(\,  \sum_{x \in \set{X}} \frac{P_X^2(x)}{Q_X(x)} \right) \\
\label{13062019b1.3}
&= f(0) + g\bigl(1 + \chi^2(P_X \| Q_X) \bigr) \\
\label{13062019b2}
&\geq f(0) + g(1) + g'(1) \, \chi^2(P_X \| Q_X) \\
\label{13062019b3}
&= g'(1) \, \chi^2(P_X \| Q_X) \\
\label{13062019b4}
&= \bigl( f'(1) + f(0) \bigr) \, \chi^2(P_X \| Q_X),
\end{align}
where \eqref{13062019b0} holds by the definition of $g$ in
Theorem~\ref{theorem: contraction coef} and the assumption that $f(0)<\infty$;
\eqref{13062019b1.2} is due to Jensen's inequality and the convexity of $g$;
\eqref{13062019b1.3} holds by the definition of the $\chi^2$-divergence;
\eqref{13062019b2} holds due to the convexity of $g$, and its differentiability
at~1 (due to the differentiability of $f$ at~1); \eqref{13062019b3} holds since
$f(0)+g(1)=f(1)=0$; finally, \eqref{13062019b4} holds since $f(1)=0$ implies that
$g'(1) = f'(1) + f(0)$.

By \cite[Theorem~5]{ISSV16}, it follows that
\begin{align}
\label{13062019b5}
D_f(P_Y \| Q_Y) \leq \kappa(\xi_1, \xi_2) \, \chi^2(P_Y \| Q_Y),
\end{align}
where $\kappa(\xi_1, \xi_2)$ is given in \eqref{def: kappa}.

Combining \eqref{13062019b-1}--\eqref{13062019b5} yields \eqref{13062019c1}.
Taking suprema on both sides of \eqref{13062019c1}, with respect to all
probability mass functions $P_X$ with $P_X \ll Q_X$ and $P_X \neq Q_X$,
gives \eqref{13062019c2} since by the definition of $\kappa(\xi_1, \xi_2)$
in \eqref{def: kappa}, it is monotonically decreasing in $\xi_1 \in [0,1)$
and monotonically increasing in $\xi_2 \in (1, \infty]$, while \eqref{xi1}
and \eqref{xi2} yield
\begin{align}
\xi_1 \geq 0, \quad \xi_2 \leq \frac1{\underset{x \in \set{X}}{\min} \, Q_X(x)}.
\end{align}

\begin{remark}
\label{remark: Raginsky16}
The proof in \eqref{13062019b-1}--\eqref{13062019b4} is conceptually
similar to the proof of \cite[Lemma~A.2]{Raginsky16}. However, the function
$g$ here is convex, and the derivation here involves the $\chi^2$-divergence.
\end{remark}

\begin{remark}
\label{remark: MakurZ18}
The proof of \cite[Theorem~8]{MakurZ18} (see Proposition~\ref{prop.: MakurZ18} in
Section~\ref{subsection: preliminaries} here) relies on \cite[Lemma~A.2]{Raginsky16},
where the function $g$ is required to be concave in \cite{MakurZ18, Raginsky16}.
This leads, in the proof of \cite[Theorem~8]{MakurZ18}, to an upper
bound on $D_f(P_Y \| Q_Y)$. One difference in the derivation of
Theorem~\ref{theorem: contraction coef} is that our requirement on the convexity of $g$
leads to a lower bound on $D_f(P_X \| Q_X)$, instead of an upper bound on $D_f(P_Y \| Q_Y)$.
Another difference between the proofs of Theorem~\ref{theorem: contraction coef}
and \cite[Theorem~8]{MakurZ18} is that we apply here the result in \cite[Theorem~5]{ISSV16}
to obtain an upper bound on $D_f(P_Y \| Q_Y)$, whereas the proof of \cite[Theorem~8]{MakurZ18}
relies on a Pinsker-type inequality (see \cite[Theorem~3]{Gilardoni10})
to obtain a lower bound on $D_f(P_X \| Q_X)$; the latter lower bound relies on the
condition on $f$ in \eqref{Gilardoni}, which is not necessary for the derivation of
the bound in Theorem~\ref{theorem: contraction coef}.
\end{remark}

\begin{remark}
\label{remark: ISSV16}
From \cite[Theorem~1~(b)]{ISSV16}, it follows that
\begin{align}
\label{tight const.}
\sup_{P \neq Q} \frac{D_f(P \| Q)}{\chi^2(P \| Q)} = \kappa(\xi_1, \xi_2),
\end{align}
with $\kappa(\xi_1, \xi_2)$ in the right side of \eqref{tight const.} as
given in \eqref{def: kappa}, and the supremum in the left side of \eqref{tight const.}
is taken over all probability measures $P$ and $Q$ such that $P \neq Q$. In
view of \cite[Theorem~1~(b)]{ISSV16}, the equality in \eqref{tight const.}
holds since the functions $\widetilde{f}, \widetilde{g} \colon (0, \infty) \to \Reals$,
defined as $\widetilde{f}(t) := f(t) + f'(1) (1-t)$ and $\widetilde{g}(t) := (t-1)^2$
for all $t>0$, satisfy
$$D_{\widetilde{f}}(P\|Q) = D_f(P\|Q), \quad D_{\widetilde{g}}(P\|Q) = \chi^2(P\|Q)$$
for all probability measures $P$ and $Q$,
and since $\widetilde{f}'(1) = \widetilde{g} \, '(1)=0$ and the function
$\widetilde{g}$ is strictly positive on $(0,1) \cup (1, \infty)$.
Furthermore, from the proof of \cite[Theorem~1~(b)]{ISSV16}, restricting $P$ and $Q$
to be probability mass functions which are defined over a binary alphabet, the ratio
$\frac{D_f(P \| Q)}{\chi^2(P \| Q)}$ can be made arbitrarily close to the supremum in
the left side of \eqref{tight const.}; such probability measures can be obtained as
the output distributions $P_Y$ and $Q_Y$ of an arbitrary non-degenerate stochastic
transformation $W_{Y|X} \colon \set{X} \to \set{Y}$, with $|\set{Y}|=2$, by a suitable
selection of probability input distributions $P_X$ and $Q_X$, respectively (see
(\ref{P_Y pos.}) and (\ref{Q_Y pos.})).
In the latter case where $|\set{Y}|=2$, this shows the optimality of the non-negative
constant $\kappa(\xi_1, \xi_2)$ in the right side of \eqref{13062019b5}.
\end{remark}

\subsection{Proof of Theorem~\ref{theorem: DMS/DMC - ver2}}
Combining \eqref{13062019b-1}--\eqref{13062019b4} gives that,
for all $\lambda \in [0,1]$,
\begin{align}
\label{13062019b6}
D_f\bigl(R_{X^n}^{(\lambda)} \, \| \, Q_{X^n}^{(\lambda)}\bigr) \geq
\bigl( f'(1) + f(0) \bigr) \, \chi^2\bigl(R_{X^n}^{(\lambda)} \, \| \, Q_{X^n}\bigr),
\end{align}
and from \eqref{13062019b5}
\begin{align}
\label{13062019b7}
D_f\bigl(R_{Y^n}^{(\lambda)} \, \| \, Q_{Y^n}\bigr) \leq
\kappa\bigl(\xi_1(n, \lambda), \xi_2(n, \lambda)\bigr)
\; \chi^2\bigl(R_{Y^n}^{(\lambda)} \, \| \, Q_{Y^n} \bigr).
\end{align}
From \eqref{1st tensorize chi2} and \eqref{mixture1 chi^2},
\begin{align}
\label{13062019b8}
\chi^2\bigl(R_{X^n}^{(\lambda)} \, \| \, Q_{X^n} \bigr) = \prod_{i=1}^n
\Bigl( 1+ \lambda^2 \chi^2(P_{X_i} \, \| \, Q_{X_i} \bigr) \Bigr) - 1,
\end{align}
and similarly, from \eqref{2nd tensorize chi2} and \eqref{mixture2 chi^2},
\begin{align}
\label{13062019b9}
\chi^2\bigl(R_{Y^n}^{(\lambda)} \, \| \, Q_{Y^n} \bigr) = \prod_{i=1}^n
\Bigl( 1+ \lambda^2 \chi^2(P_{Y_i} \, \| \, Q_{Y_i} \bigr) \Bigr) - 1.
\end{align}
Combining \eqref{13062019b6}--\eqref{13062019b9} yields \eqref{13062019d1}.

\section{Proof of Theorem~\ref{thm: f_alpha-divergence}}
\label{appendix: f_alpha-divergence}

The function $f_\alpha \colon [0, \infty) \to \Reals$ in
\eqref{f_alpha} satisfies $f_\alpha(1)=0$, and for all
$\alpha \geq \mathrm{e}^{-\frac32}$
\begin{align}
\label{14062019a1}
& f_\alpha''(t) = 2 \log(\alpha+t) + 3 \log \mathrm{e} > 0, \quad \forall \, t>0,
\end{align}
which yields the convexity of $f_\alpha(\cdot)$ on $[0, \infty)$.
This justifies the definition of the $f$-divergence
\begin{align}
\label{def f_alpha div.}
D_{f_\alpha}(P\|Q) &:= \sum_{x \in \set{X}} Q(x) \;
f_\alpha\biggl(\frac{P(x)}{Q(x)}\biggr)
\end{align}
for probability mass functions $P$ and $Q$, which are defined
on a finite or countably infinite set $\set{X}$, with $Q$
supported on $\set{X}$. In the general
alphabet setting, sums and probability mass functions are,
respectively, replaced by Lebesgue integrals and Radon-Nikodym
derivatives.

Differentiation of both sides of \eqref{def f_alpha div.}
with respect to $\alpha$ gives
\begin{align}
\label{1st derivative}
\frac{\partial}{\partial \alpha} \bigl\{ D_{f_\alpha}(P\|Q) \bigr\}
= \sum_{x \in \set{X}} Q(x) \; r_\alpha\biggl(\frac{P(x)}{Q(x)}\biggr)
\end{align}
where
\begin{align}
\label{r_alpha}
r_\alpha(t) &:= \frac{\partial f_\alpha(t)}{\partial \alpha} \\
\label{r_alpha 2}
&= 2(\alpha+t) \log(\alpha+t) - 2(\alpha+1) \log(\alpha+1)
+ (t-1) \log \mathrm{e}, \quad t>0.
\end{align}
The function $r_\alpha \colon (0, \infty) \to \Reals$ is convex since
\begin{align}
r_\alpha''(t) = \frac{2 \log \mathrm{e}}{\alpha + t} > 0, \quad \forall \, t>0,
\end{align}
and $r_\alpha(1)=0$.
Hence, $D_{r_\alpha}(\cdot \| \cdot)$ is an $f$-divergence, and it
follows from \eqref{1st derivative}--\eqref{r_alpha 2} that
\begin{align}
& \frac{\partial}{\partial \alpha} \bigl\{ D_{f_\alpha}(P\|Q) \bigr\} \nonumber \\
& = D_{r_\alpha}(P\|Q) \\
& = 2 \sum_{x \in \set{X}} \left\{ \bigl(\alpha Q(x) + P(x) \bigr)
\, \log \left(\alpha+\frac{P(x)}{Q(x)}\right) \right\} - 2(\alpha+1) \log(\alpha+1) \\
& = 2(\alpha+1) \sum_{x \in \set{X}} \frac{\alpha Q(x) + P(x)}{\alpha+1}
\, \log \left(\frac{\alpha Q(x)+P(x)}{(\alpha+1) \, Q(x)} \right) \\
& = 2(\alpha+1) \, D\biggl( \frac{\alpha Q + P}{\alpha+1} \, \| \, Q \biggr)
\geq 0,
\end{align}
which gives \eqref{1st partial der.}, so
$D_{f_\alpha}(\cdot \| \cdot)$ is monotonically
increasing in $\alpha$. Double differentiation of both sides of
\eqref{def f_alpha div.} with respect to $\alpha$ gives
\begin{align}
\label{2nd derivative}
\frac{\partial^2}{\partial \alpha^2} \bigl\{ D_{f_\alpha}(P\|Q) \bigr\}
= \sum_{x \in \set{X}} Q(x) \; v_\alpha\biggl(\frac{P(x)}{Q(x)}\biggr)
\end{align}
where
\begin{align}
\label{v: 2nd derivative}
v_\alpha(t) &:= \frac{\partial^2 f_\alpha(t)}{\partial \alpha^2} \\
\label{v2: 2nd derivative}
& \, = 2 \log(\alpha+t) - 2 \log(\alpha+1), \quad t>0.
\end{align}
The function $v_\alpha \colon (0, \infty) \to \Reals$ is concave,
and $v_\alpha(1)=0$. By referring to the
$f$-divergence $D_{-v_\alpha}(\cdot \| \cdot)$, it follows from
\eqref{2nd derivative}--\eqref{v2: 2nd derivative} that
\begin{align}
& \frac{\partial^2}{\partial \alpha^2}
\bigl\{ D_{f_\alpha}(P\|Q) \bigr\} \nonumber \\
& = - D_{-v_\alpha}(P\|Q) \\
& = -2 \sum_{x \in \set{X}} Q(x) \left[ \log(\alpha+1)
- \log\left(\alpha + \frac{P(x)}{Q(x)} \right) \right] \\
& = -2 \sum_{x \in \set{X}} Q(x) \log \left(
\frac{(\alpha+1) Q(x)}{\alpha Q(x) + P(x)} \right) \\
& = -2 \, D\biggl(Q \, \| \, \frac{\alpha Q + P}{\alpha+1} \biggr)
\leq 0,
\end{align}
which gives \eqref{2nd partial der.}, so
$D_{f_\alpha}(\cdot \| \cdot)$ is concave in $\alpha$
for $\alpha \geq \mathrm{e}^{-\frac32}$. Differentiation of
both sides of \eqref{v2: 2nd derivative} gives that
\begin{align}
\label{3rd derivative}
\frac{\partial^3 f_\alpha(t)}{\partial \alpha^3} &=
2 \left(\frac1{\alpha+t} - \frac1{\alpha+1} \right) \log \mathrm{e},
\end{align}
which implies that
\begin{align}
\frac{\partial^3}{\partial \alpha^3} \bigl\{ D_{f_\alpha}(P\|Q) \bigr\}
&= 2 \log \mathrm{e} \, \sum_{x \in \set{X}} Q(x) \left( \tfrac1{\alpha +
\tfrac{P(x)}{Q(x)}} - \tfrac1{\alpha+1} \right) \\
&= \frac{2 \log \mathrm{e}}{\alpha+1} \left[ \, \sum_{x \in \set{X}}
\frac{Q^2(x)}{\frac{\alpha Q(x) + P(x)}{\alpha+1}} - 1 \right] \\
&= \frac{2 \log \mathrm{e}}{\alpha+1} \cdot \chi^2\biggl(Q \, \|
\, \frac{\alpha Q + P}{\alpha+1} \biggr) \geq 0.
\end{align}
This gives \eqref{3rd partial der.}, and it completes the proof
of Item~\ref{Thm. f-1}).

We next prove Item~\ref{Thm. f-1b}). From Item~\ref{Thm. f-1}),
the result in \eqref{generalized} holds for $n=1, 2, 3$. We provide in the
following a proof of  \eqref{generalized} for all $n \geq 3$.
In view of \eqref{3rd derivative}, it can be verified that for $n \geq 3$,
\begin{align} \label{f_alpha - nth der.}
\frac{\partial^n f_\alpha(t)}{\partial \alpha^n} =
2 (-1)^{n-1} (n-3)! \left[ \frac1{(\alpha+t)^{n-2}}
- \frac1{(\alpha+1)^{n-2}} \right] \log \mathrm{e},
\end{align}
which, from \eqref{def f_alpha div.}, implies that
\begin{align} \label{n-th derivative}
(-1)^{n-1} \frac{\partial^n}{\partial \alpha^n} \bigl\{ D_{f_\alpha}(P\|Q) \bigr\}
= \sum_{x \in \set{X}} Q(x) \, g_{\alpha, n} \left(\frac{P(x)}{Q(x)}\right)
\end{align}
with
\begin{align} \label{g}
g_{\alpha, n}(t) &:= (-1)^{n-1} \, \frac{\partial^n f_\alpha(t)}{\partial \alpha^n} \\
&= 2 (n-3)! \left[ \frac1{(\alpha+t)^{n-2}} - \frac1{(\alpha+1)^{n-2}} \right] \log \mathrm{e},
\quad t>0.
\end{align}
The function $g_{\alpha, n} \colon (0, \infty) \to \Reals$ is convex for
$n \geq 3$, with $g_{\alpha, n}(1) = 0$. By referring to the
$f$-divergence $D_{g_{\alpha, n}}(\cdot \| \cdot)$, its non-negativity and
\eqref{n-th derivative} imply that for all $n \geq 3$
\begin{align}
(-1)^{n-1} \frac{\partial^n}{\partial \alpha^n} \bigl\{ D_{f_\alpha}(P\|Q) \bigr\}
& = D_{g_{\alpha, n}}(P\|Q) \geq 0.
\end{align}
Furthermore, we get the following explicit formula for $n$-th partial derivative
of $D_{f_\alpha}(P\|Q)$ with respect to $\alpha$ for $n \geq 3$:
\begin{align}
& \frac{\partial^n}{\partial \alpha^n} \bigl\{ D_{f_\alpha}(P\|Q) \bigr\} \nonumber \\
\label{RD1}
&= (-1)^{n-1} \, \sum_{x \in \set{X}} Q(x) \, g_{\alpha, n} \left(\frac{P(x)}{Q(x)}\right) \\
\label{RD2}
&= \frac{2 (-1)^{n-1} (n-3)! \, \log \mathrm{e}}{(\alpha+1)^{n-2}}
\left[ \sum_{x \in \set{X}} \left\{ Q(x) \left( \tfrac{\alpha+1}{\alpha + \tfrac{P(x)}{Q(x)}}
\right)^{n-2} \right\} - 1 \right] \\[0.1cm]
\label{RD3}
&= \frac{2 (-1)^{n-1} (n-3)! \, \log \mathrm{e}}{(\alpha+1)^{n-2}}
\left[ \, \sum_{x \in \set{X}} \frac{Q^{n-1}(x)}{\left( \frac{\alpha Q(x)+P(x)}{\alpha+1}
\right)^{n-2}} - 1 \right] \\[0.1cm]
\label{RD4}
&= \frac{2 (-1)^{n-1} (n-3)! \, \log \mathrm{e}}{(\alpha+1)^{n-2}}
\left[ \exp\biggl( (n-2) \, D_{n-1}\Bigl(Q \, \| \, \tfrac{\alpha Q + P}{\alpha+1} \Bigr)
\biggr) -1 \right]
\end{align}
where \eqref{RD1} holds due to \eqref{n-th derivative}; \eqref{RD2} follows from \eqref{g},
and \eqref{RD4} is satisfied by the definition of the R\'{e}nyi divergence \cite{Renyientropy}
which is given by
\begin{align} \label{def: RD}
D_\beta(P\|Q) := \frac1{\beta-1} \, \log \left( \sum_{x \in \set{X}} P^\beta(x) \, Q^{1-\beta}(x) \right),
\quad \forall \, \beta \in (0,1) \cup (1, \infty)
\end{align}
with $D_1(P\|Q) := D(P\|Q)$ by continuous extension of $D_\beta(\cdot \| \cdot)$ at $\beta=1$.
For $n=3$, the right side of \eqref{RD4} is simplified to the right side of $\eqref{3rd partial der.}$;
this holds due to the identity
\begin{align} \label{RD2-chi2}
D_2(P\|Q) = \log \bigl(1 + \chi^2(P\|Q) \bigr).
\end{align}

To prove Item~\ref{Thm. f-2}), from \eqref{f_alpha}, for all $t \geq 0$
\begin{align}
& f_\alpha'(t) = 2(\alpha+t) \log(\alpha+t) + (\alpha+t) \log \mathrm{e}, \\
& f_\alpha''(t) = 2 \log(\alpha+t) + 3 \log \mathrm{e},  \label{2nd der. f} \\
& f_\alpha^{(3)}(t) = \tfrac{2 \log \mathrm{e}}{\alpha+t}, \label{3rd der. f}
\end{align}
which implies by a Taylor series expansion of $f_\alpha(\cdot)$ that
\begin{align} \label{Taylor-3rd}
f_\alpha(t) = f_\alpha(1) + f'_\alpha(1) (t-1) + \tfrac12 f_\alpha''(1) (t-1)^2
+ \tfrac16 f_\alpha^{(3)}(\xi) (t-1)^3,
\quad \forall \, t \geq 0
\end{align}
where $\xi$ in the right side of \eqref{Taylor-3rd} is an intermediate value between~1
and~$t$. Hence, for $t \geq 0$,
\begin{align}
\label{LB0 on f}
f_\alpha(t) & \geq f'_\alpha(1) (t-1) + \tfrac12 f_\alpha''(1) (t-1)^2 +
\tfrac16 f_\alpha^{(3)}(0) (t-1)^3 \, 1\{ t \in [0,1]\} \\
\label{LB1 on f}
& \geq f'_\alpha(1) (t-1) + \bigl( \tfrac12 f_\alpha''(1) -
\tfrac16 f_\alpha^{(3)}(0) \bigr) (t-1)^2 \\
\label{LB on f}
& = f'_\alpha(1) (t-1) +  k(\alpha) \, (t-1)^2
\end{align}
where \eqref{LB0 on f} follows from \eqref{Taylor-3rd} since $f_\alpha(1)=0$ and
$f_\alpha^{(3)}(\cdot)$ is monotonically decreasing and positive (see \eqref{3rd der. f});
$1\{t \in [0,1]\}$ in the right side of \eqref{LB0 on f} denotes the indicator function
which is equal to~1 if the relation $t \in [0,1]$ holds, and it is otherwise equal to zero;
\eqref{LB1 on f} holds since $(t-1)^3 \, 1\{ t \in [0,1]\} \geq -(t-1)^2$ for all $t \geq 0$,
and $f_\alpha^{(3)}(0) > 0$; finally, \eqref{LB on f} follows by substituting \eqref{2nd der. f}
and \eqref{3rd der. f} into the right side of \eqref{LB1 on f}, which gives the equality
\begin{align}
\tfrac12 f_\alpha''(1) - \tfrac16 f_\alpha^{(3)}(0) = k(\alpha)
\end{align}
with $k(\cdot)$ as defined in \eqref{k_alpha}.
Since the first term in the right side of \eqref{LB on f} does not affect an $f$-divergence
(as it is equal to $c \, (t-1)$ for $t \geq 0$ and some constant $c$), and for an arbitrary positive
constant $k > 0$ and $g(t) := (t-1)^2$ for $t \geq 0$, we get $D_{kg}(P\|Q) = k \, \chi^2(P\|Q)$,
inequality \eqref{Df-chi2} follows from \eqref{LB0 on f} and \eqref{LB on f}. To that end,
note that $k = k(\alpha)$ defined in \eqref{k_alpha}
is monotonically increasing in $\alpha$, and therefore
$k(\alpha) \geq k(\mathrm{e}^{-\tfrac32}) > 0.2075$ for all $\alpha \geq \mathrm{e}^{-\tfrac32}$.
Due to the inequality (see, e.g., \cite[Theorem~5]{GibbsSu02}, followed by refined versions
in \cite[Theorem~20]{ISSV16} and \cite[Theorem~9]{Simic15})
\begin{align}
D(P\|Q) \leq \log\bigl(1+\chi^2(P\|Q)\bigr),
\end{align}
the looser lower bound on $D_{f_\alpha}(P\|Q)$ in the right side of \eqref{Df-KL}, expressed
as a function of the relative entropy $D(P\|Q)$, follows from \eqref{Df-chi2}. Hence,
if $P$ and $Q$ are not identical, then \eqref{lim Infty} follows from \eqref{Df-chi2}
since $\chi^2(P\|Q) > 0$ and $\underset{\alpha \to \infty}{\lim} k(\alpha) = \infty$.

We next prove Item~\ref{Thm. f-2b}). The Taylor series expansion of $f_\alpha(\cdot)$
implies that for all $t \geq 0$
\begin{align} \label{Taylor-4th}
\hspace*{-0.3cm} f_\alpha(t) = f_\alpha(1) + f'_\alpha(1) (t-1) + \tfrac12 f_\alpha''(1) (t-1)^2
+ \tfrac16 f_\alpha^{(3)}(1) (t-1)^3 + \tfrac1{24} f_\alpha^{(4)}(\xi) (t-1)^4
\end{align}
where $\xi$ in the right side of \eqref{Taylor-4th} is an intermediate value between~1
and~$t$. Consequently, since $f_\alpha^{(4)}(\xi) = -\tfrac{2 \log \mathrm{e}}{(\alpha+\xi)^2} < 0$
and $f_\alpha(1)=0$, it follows from \eqref{Taylor-4th} that, for all $t \geq 0$,
\begin{align}
& f_\alpha(t) \nonumber \\
\label{UB f_alpha1}
& \leq f'_\alpha(1) (t-1) + \tfrac12 f_\alpha''(1) (t-1)^2
+ \tfrac16 f_\alpha^{(3)}(1) (t-1)^3 \\
\label{UB f_alpha2}
& = f'_\alpha(1) (t-1) + \tfrac12 f_\alpha''(1) (t-1)^2
+ \tfrac16 f_\alpha^{(3)}(1) \bigl[t^3 - 3(t-1)^2 - 3(t-1) - 1] \\
\label{UB f_alpha3}
& = \bigl[ f'_\alpha(1) - \tfrac12 f_\alpha^{(3)}(1) \bigr] (t-1)
+ \tfrac12 \bigl[f_\alpha''(1) - f_\alpha^{(3)}(1) \bigr] (t-1)^2
+ \tfrac16 f_\alpha^{(3)}(1) \, (t^3-1).
\end{align}
Based on \eqref{UB f_alpha1}--\eqref{UB f_alpha3}, it follows that
\begin{align}
D_{f_\alpha}(P\|Q) & \leq \tfrac12 \bigl[f_\alpha''(1) - f_\alpha^{(3)}(1) \bigr] \chi^2(P\|Q)
+ \tfrac16 f_\alpha^{(3)}(1) \,
\sum_{x \in \set{X}} \Biggl\{ Q(x) \biggl[ \biggl(\frac{P(x)}{Q(x)}\biggr)^3 - 1 \biggr] \Biggr\}
\nonumber \\
& = \tfrac12 \bigl[f_\alpha''(1) - f_\alpha^{(3)}(1) \bigr] \chi^2(P\|Q)
+ \tfrac16 f_\alpha^{(3)}(1) \Biggl( -1 + \sum_{x \in \set{X}} \frac{P^3(x)}{Q^2(x)} \Biggr) \\[0.1cm]
\label{UB f_alpha4}
& = \tfrac12 \bigl[f_\alpha''(1) - f_\alpha^{(3)}(1) \bigr] \chi^2(P\|Q)
+ \tfrac16 f_\alpha^{(3)}(1) \Bigl[ \exp\bigl(2 D_3(P\|Q)\bigr) - 1 \Bigr],
\end{align}
where \eqref{UB f_alpha4} holds due to \eqref{def: RD} (with $\beta=3$). Substituting
\eqref{2nd der. f} and \eqref{3rd der. f} into the right side of \eqref{UB f_alpha4}
gives \eqref{Df_UB}.

We next prove Item~\ref{Thm. f-2c}).
Let $P$ and $Q$ be probability mass functions such that $D_3(P\|Q) < \infty$,
and let $\varepsilon > 0$ be arbitrarily small. Since the R\'{e}nyi divergence
$D_{\alpha}(P\|Q)$ is monotonically non-decreasing in $\alpha>0$ (see
\cite[Theorem~3]{ErvenH14}), it follows that $D_2(P\|Q) < \infty$, and therefore also
\begin{align}
\label{D2-chi^2}
\chi^2(P\|Q) = \exp\bigl(D_2(P\|Q)\bigr)-1 < \infty.
\end{align}
In view of \eqref{Df-chi2}, there exists
$\alpha_1 := \alpha_1(P,Q,\varepsilon)$ such that for all $\alpha > \alpha_1$
\begin{align}
D_{f_\alpha}(P\|Q) > \bigl(\log (\alpha+1)
+ \tfrac32 \log \mathrm{e} \bigr) \, \chi^2(P\|Q) - \varepsilon,
\end{align}
and, from \eqref{Df_UB}, there exists $\alpha_2 := \alpha_2(P,Q,\varepsilon)$
such that for all $\alpha > \alpha_2$
\begin{align}
D_{f_\alpha}(P\|Q) <
\bigl(\log (\alpha+1) + \tfrac32 \log \mathrm{e} \bigr) \, \chi^2(P\|Q) + \varepsilon.
\end{align}
Letting $\alpha^\ast := \max\{\alpha_1, \alpha_2\}$ gives the result in \eqref{asymp.}
for all $\alpha > \alpha^\ast$.

Item~\ref{Thm. f-3}) of Theorem~\ref{thm: f_alpha-divergence} is a direct consequence
of \cite[Lemma~4]{Sason18}, which relies on \cite[Theorem~3]{PardoV03}. Let
$g(t) := (t-1)^2$ for $t \geq 0$ (hence, $D_g(\cdot \| \cdot)$ is the $\chi^2$ divergence).
If a sequence $\{P_n\}$ converges to a probability measure $Q$ in the sense that
the condition in \eqref{eq: 1st condition} is satisfied, and $P_n \ll Q$ for all
sufficiently large $n$, then \cite[Lemma~4]{Sason18} yields
\begin{align}
\label{eq: IS18}
\lim_{n \to \infty} \frac{D_{f_\alpha}(P_n \| Q)}{\chi^2(P_n \| Q)} = \tfrac12 f''_\alpha(1),
\end{align}
which gives \eqref{eq: limit of ratio of f-div} from \eqref{2nd der. f} and \eqref{eq: IS18}.

We next prove Item~\ref{Thm. f-4}). Inequality \eqref{diff Df1} is trivial.
Inequality \eqref{diff Df1b} is obtained as follows:
\begin{align}
\label{diff Df1 - a}
D_{f_\alpha}(P\|Q) - D_{f_\beta}(P\|Q) &= \int_{\beta}^{\alpha}
\frac{\partial}{\partial u} \bigl\{ D_{f_u}(P\|Q) \bigr\} \, \mathrm{d}u \\
\label{diff Df1 - b}
&= \int_{\beta}^{\alpha}
2(u+1) \, D\Bigl( \tfrac{u Q + P}{u+1} \, \| \, Q \Bigr) \, \mathrm{d}u \\
\label{diff Df1 - c}
&\geq \int_{\beta}^{\alpha} 2(u+1) \, \mathrm{d}u \cdot
D\Bigl( \tfrac{\alpha Q + P}{\alpha+1} \, \| \, Q \Bigr) \\
\label{diff Df1 - d}
&= \bigl[(\alpha+1)^2 - (\beta+1)^2 \bigr] \,
D\Bigl( \tfrac{\alpha Q + P}{\alpha+1} \, \| \, Q \Bigr) \\
\label{diff Df1 - e}
&= (\alpha-\beta) (\alpha+\beta+2) \,
D\Bigl( \tfrac{\alpha Q + P}{\alpha+1} \, \| \, Q \Bigr)
\end{align}
where \eqref{diff Df1 - b} follows from \eqref{1st partial der.},
and \eqref{diff Df1 - c} holds since the function
$I \colon [0, \infty) \to [0, \infty)$ given by
\begin{align}
I(u) := D\Bigl( \tfrac{u Q + P}{u+1} \, \| \, Q \Bigr), \quad u \geq 0
\end{align}
is monotonically decreasing in $u$ (note that by increasing the value
of the non-negative variable $u$, the probability mass function
$\tfrac{u Q + P}{u+1}$ gets closer to $Q$). This gives \eqref{diff Df1b}.

For proving inequality \eqref{diff Df2}, we obtain two upper
bounds on $D_{f_\alpha}(P\|Q) - D_{f_\beta}(P\|Q)$ with
$\alpha > \beta \geq \mathrm{e}^{-\frac32}$. For the derivation of the
first bound, we rely on \eqref{1st derivative}.
From \eqref{r_alpha}--\eqref{r_alpha 2},
\begin{align} \label{r,s}
r_\alpha(t) = 2t \log t - s_\alpha(t), \quad t \geq 0
\end{align}
where $s_\alpha \colon (0, \infty) \to \Reals$ is given by
\begin{align}
s_\alpha(t) := 2t \log t - 2(\alpha+t) \log(\alpha+t)
+ (1-t) \log \mathrm{e} + 2(\alpha+1) \log(\alpha+1), \quad t \geq 0,
\end{align}
with the convention that $0 \log 0 = 0$ (by a continuous extension
of $t \log t$ at $t=0$). Since $s_\alpha(1)=0$, and
\begin{align}
s''_\alpha(t) = \frac{2 \alpha}{t(\alpha+t)} > 0, \quad \forall \, t > 0,
\end{align}
which implies that $s_\alpha(\cdot)$ is convex on $(0, \infty)$, we get
\begin{align}
\frac{\partial}{\partial \alpha} \bigl\{ D_{f_\alpha}(P\|Q) \bigr\}
&= D_{r_\alpha}(P\|Q) \label{u div.} \\
&= 2 D(P\|Q) - D_{s_\alpha}(P\|Q) \label{s div.} \\
&\leq 2 D(P\|Q) \label{non-neg.}
\end{align}
where \eqref{u div.} holds due to \eqref{1st derivative} (recall the convexity
of $r_\alpha \colon (0, \infty) \to \Reals$ with $r_\alpha(1)=0$); \eqref{s div.}
holds due to \eqref{r,s} and since $r(t) := t \log t$ for $t>0$ yields
$D_r(P\|Q) = D(P\|Q)$; finally, \eqref{non-neg.} follows from the non-negativity
of the $f$-divergence $D_{s_\alpha}(\cdot \| \cdot)$. Consequently, integration
over the interval $[\beta, \alpha]$ ($\alpha > \beta$) on the left side of
\eqref{u div.} and the right side of \eqref{non-neg.} gives
\begin{align}
\label{diff Df3-a}
D_{f_\alpha}(P\|Q) - D_{f_\beta}(P\|Q) \leq 2 (\alpha-\beta) \, D(P\|Q).
\end{align}

Note that the same reasoning of \eqref{diff Df1 - a}--\eqref{diff Df1 - e}
also implies that
\begin{align}
\label{diff Df3-b}
D_{f_\alpha}(P\|Q) - D_{f_\beta}(P\|Q)
&\leq (\alpha-\beta) (\alpha+\beta+2) \,
D\Bigl( \tfrac{\beta Q + P}{\beta+1} \, \| \, Q \Bigr),
\end{align}
which gives a second upper bound on the left side of \eqref{diff Df3-b}.
Taking the minimal value among the two upper bounds in the right
sides of \eqref{diff Df3-a} and \eqref{diff Df3-b} gives \eqref{diff Df2}
(see Remark~\ref{remark: none UB is better} at the end of the proof of
Theorem~\ref{thm: f_alpha-divergence}).

We finally prove Item~\ref{Thm. f-5}). From \eqref{f_alpha} and
\eqref{14062019a1}, the function $f_\alpha \colon [0, \infty) \to \Reals$
is convex for $\alpha \geq \mathrm{e}^{-\frac32}$ with $f_\alpha(1)=0$,
$f_\alpha(0) = \alpha^2 \log \alpha - (\alpha+1)^2 \log(\alpha+1) \in \Reals$,
and it is also differentiable at~1. It is left to prove that the
function $g_\alpha \colon (0, \infty) \to \Reals$, defined as
$g_\alpha(t) := \frac{f_\alpha(t) - f_\alpha(0)}{t}$ for $t>0$,
is convex. From \eqref{f_alpha}, the function $g_\alpha$ is given explicitly by
\begin{align}
\label{14062019a2}
g_\alpha(t) = \frac{(\alpha+t)^2 \log(\alpha+t) - \alpha^2 \log \alpha}{t}, \quad t>0,
\end{align}
and its second derivative is given by
\begin{align}
\label{14062019a3}
g''_\alpha(t) = \frac{w_\alpha(t)}{t^3}, \quad t>0,
\end{align}
with
\begin{align}
\label{14062019a4}
w_\alpha(t) := 2 \alpha^2 \log\left(1+\frac{t}{\alpha}\right) + t(t-2\alpha) \log \mathrm{e},
\quad t \geq 0.
\end{align}
Since $w_\alpha(0) = 0$, and
\begin{align}
\label{14062019a5}
w'_\alpha(t) = \frac{2t^2 \log \mathrm{e}}{\alpha+t} > 0, \quad \forall \, t>0,
\end{align}
it follows that $w_\alpha(t)>0$ for all $t>0$; hence, from \eqref{14062019a3},
$g''_\alpha(t)>0$ for $t \in (0, \infty)$, which yields the convexity of the
function $g_\alpha(\cdot)$ on $(0, \infty)$ for all $\alpha \geq 0$. This shows
that, for every $\alpha \geq \mathrm{e}^{-\frac32}$, the function
$f_\alpha \colon [0, \infty) \to \Reals$ satisfies all the required conditions
in Theorems~\ref{theorem: contraction coef} and~\ref{theorem: DMS/DMC - ver2}.
We proceed to calculate the function $\kappa_\alpha \colon [0,1) \times
(1, \infty) \to \Reals$ in \eqref{def: kappa}, which corresponds to
$f := f_\alpha$, i.e. (see \eqref{kappa_alpha 1}),
\begin{align}
\label{kappa_alpha 3}
& \kappa_\alpha(\xi_1, \xi_2) =
\sup_{t \in (\xi_1, 1) \cup (1, \xi_2)} z_\alpha(t),
\end{align}
with
\begin{align}
\label{def: z_alpha}
z_\alpha(t) :=
\begin{dcases}
\frac{f_\alpha(t) + f_\alpha'(1) \, (1-t)}{(t-1)^2},
& \quad t \in [0,1) \cup (1, \infty), \\
\tfrac32 \, \log \mathrm{e} + \log(\alpha+1), & \quad t=1,
\end{dcases}
\end{align}
where the definition of $z_\alpha(1)$ is obtained by continuous extension
of the function $z_\alpha(\cdot)$ at $t=1$ (recall that the function
$f_\alpha(\cdot)$ is given in \eqref{f_alpha}). Differentiation shows that
\begin{align}
\label{v_alpha0}
\frac{\partial \, z_\alpha(t)}{\partial t} = \frac{v_\alpha(t)}{(t-1)^4},
\quad t \in [0,1) \cup (1, \infty),
\end{align}
where, for $t \geq 0$,
\begin{align}
\label{v_alpha1}
& v_\alpha(t) := (2 \alpha + t + 1) (t-1)^2 \log \mathrm{e}
- 2(\alpha+1)(\alpha+t)(t-1) \log \frac{\alpha+t}{\alpha+1},
\end{align}
and
\begin{align}
\label{v_alpha2}
& v'_\alpha(t) = (t-1)^2 \log \mathrm{e} + 2 (\alpha+t)(t-1) \log \mathrm{e}
-2(\alpha+1)(2t+\alpha-1) \log \frac{\alpha+t}{\alpha+1}, \\
\label{v_alpha3}
& v''_\alpha(t) = 6(t-1) \log \mathrm{e} + \frac{2(\alpha+1)^2 \log \mathrm{e}}{\alpha+t}
- 4(\alpha+1) \log \frac{\alpha+t}{\alpha+1}, \\
\label{v_alpha4}
& v_\alpha^{(3)}(t) = \frac{2(t-1)(3t+4\alpha+1)}{(\alpha+t)^2}.
\end{align}
From \eqref{v_alpha4}, it follows that $v_\alpha^{(3)}(t) < 0$ if $t \in [0,1)$,
$v_\alpha^{(3)}(1) = 0$, and $v_\alpha^{(3)}(t) > 0$ if $t \in (1,\infty)$.
Since $v''_\alpha(\cdot)$ is therefore monotonically decreasing on
$[0,1]$ and it is monotonically increasing on $[1, \infty)$, \eqref{v_alpha3}
implies that
\begin{align} \label{v_alpha5}
v''_\alpha(t) \geq v''_\alpha(1) = 2(\alpha+1) \log \mathrm{e} > 0, \quad \forall \, t \geq 0.
\end{align}
Since $v'_\alpha(1)=0$ (see \eqref{v_alpha2}), and $v'_\alpha(\cdot)$
is monotonically increasing on $[0, \infty)$, it follows that $v'_\alpha(t)<0$
for all $t \in [0,1)$ and $v'_\alpha(t) > 0$ for all $t>1$. This implies that
$v_\alpha(t) \geq v_\alpha(1) = 0$ for all $t \geq 0$ (see \eqref{v_alpha1});
hence, from \eqref{v_alpha0}, the function $z_\alpha(\cdot)$ is monotonically
increasing on $[0, \infty)$, and it is continuous over this interval (see
\eqref{def: z_alpha}). It therefore follows from \eqref{kappa_alpha 3} that
\begin{align}
\kappa_\alpha(\xi_1, \xi_2) = z_\alpha(\xi_2),
\end{align}
for every $\xi_1 \in [0,1)$ and $\xi_2 \in (1, \infty)$ (independently of $\xi_1$),
which proves \eqref{kappa_alpha 2}.

\begin{remark}
\label{remark: none UB is better}
None of the upper bounds in the right sides of \eqref{diff Df3-a} and
\eqref{diff Df3-b} supersedes the other. For example, if $P$ and $Q$
correspond to $\text{Bernoulli}(p)$ and $\text{Bernoulli}(q)$,
respectively, and $(\alpha, \beta, p, q) = (2,1,\tfrac15,\tfrac25)$,
then the right sides of \eqref{diff Df3-a} and \eqref{diff Df3-b} are,
respectively, equal to $0.264 \log \mathrm{e}$ and $0.156 \log \mathrm{e}$.
If on the other hand $(\alpha, \beta, p, q) = (10,1,\tfrac15,\tfrac25)$,
then the right sides of \eqref{diff Df3-a} and \eqref{diff Df3-b} are,
respectively, equal to $2.377 \log \mathrm{e}$ and $3.646 \log \mathrm{e}$.
\end{remark}

\section{Proof of Theorem~\ref{thm: majorization Df}}
\label{appendix: majorization}

By assumption, $P \prec Q$ where the probability mass functions $P$ and $Q$
are defined on the set $\set{A} := \{1, \ldots, n\}$. The majorization
relation $P \prec Q$ is equivalent to the existence of a doubly-stochastic
transformation $W_{Y|X} \colon \set{A} \to \set{A}$ such that (see
Proposition~\ref{proposition: majorization and DP})
\begin{align} \label{Q to P}
Q \to W_{Y|X} \to P.
\end{align}
(See, e.g., \cite[Theorem~2.1.10]{Bhatia} or \cite[Theorem~2.B.2]{MarshallOA}
or \cite[pp.~195--204]{Steele}). Define
\begin{align} \label{input pmfs}
\set{X} = \set{Y} := \set{A}, \quad P_X := Q, \quad Q_X := U_n.
\end{align}
The probability mass functions given by
\begin{align} \label{output pmfs}
P_Y := P, \quad Q_Y := U_n
\end{align}
satisfy, respectively, relations \eqref{transP} and \eqref{transQ}.
The first one is obvious from \eqref{Q to P}--\eqref{output pmfs};
relation \eqref{transQ} holds due to the fact that
$W_{Y|X} \colon \set{A} \to \set{A}$ is a doubly stochastic
transformation, which implies that for all $y \in \set{A}$
\begin{align}
\sum_{x \in \set{A}} Q_X(x) P_{Y|X}(y|x)
&= \frac1n \sum_{x \in \set{A}} P_{Y|X}(y|x) \\
&= \frac1n = Q_Y(y).
\end{align}
Since (by assumption) $P_X$ and $Q_X$ are supported on $\set{A}$,
relations \eqref{transP} and \eqref{transQ} hold in the setting of
\eqref{Q to P}--\eqref{output pmfs}, and $f \colon (0, \infty) \to \Reals$ is
(by assumption) convex and twice differentiable, it is possible to apply
the bounds in Theorem~\ref{thm: SDPI-IS}~\ref{Th. 1.b}) and~\ref{Th. 1.d}).
To that end, from \eqref{xi1}, \eqref{xi2}, \eqref{input pmfs} and \eqref{output pmfs},
\begin{align} \label{updated xi's}
& \xi_1 = \min_{x \in \set{A}} \frac{Q(x)}{\tfrac1n} = n q_{\min}, \\[0.1cm]
& \xi_2 = \max_{x \in \set{A}} \frac{Q(x)}{\tfrac1n} = n q_{\max},
\end{align}
which, from \eqref{key}, \eqref{DPI1}, \eqref{UB1}, \eqref{input pmfs},
\eqref{output pmfs} and \eqref{updated xi's}, give that
\begin{align}
& e_f(n q_{\min}, n q_{\max})
\left[ \chi^2(Q \| U_n) - \chi^2(P \| U_n) \right] \nonumber \\
& \geq D_f(Q \| U_n) - D_f(P \| U_n) \label{270519a} \\
&\geq c_f(n q_{\min}, n q_{\max})
\left[ \chi^2(Q \| U_n) - \chi^2(P \| U_n) \right] \label{270519b} \\
&\geq 0.
\end{align}
The difference of the $\chi^2$ divergences in the left side of
\eqref{270519a} and the right side of \eqref{270519b} satisfies
\begin{align}
\chi^2(Q \| U_n) - \chi^2(P \| U_n) \nonumber
&= \sum_{x \in \set{A}} \frac{Q^2(x)}{\tfrac1n} -
\sum_{x \in \set{A}} \frac{P^2(x)}{\tfrac1n} \nonumber \\
&= n \bigl( \| Q \|_2^2 - \| P \|_2^2 \bigr), \label{270519c}
\end{align}
and the substitution of \eqref{270519c} into the bounds in
\eqref{270519a} and \eqref{270519b} give the result
in \eqref{UB diff Df - equiprob.} and \eqref{LB diff Df - equiprob.}.

Let $f(t) = (t-1)^2$ for $t>0$. From \eqref{c_f} and \eqref{e_f}, it
yields $c_f(\cdot, \cdot) = e_f(\cdot, \cdot) = 1$. Since $D_f(\cdot \| \cdot)
= \chi^2(\cdot \| \cdot)$, it follows from \eqref{270519c} that the upper
and lower bounds in the left side of \eqref{UB diff Df - equiprob.} and the
right side of \eqref{LB diff Df - equiprob.}, respectively, coincide
for the $\chi^2$-divergence; this therefore yields the tightness of these
bounds in this special case.

We next prove \eqref{bounds on diff. norms}. The following lower bound
on the second-order R\'{e}nyi entropy (a.k.a. the collision entropy)
holds (see \cite[(25)--(27)]{Sason18b}):
\begin{align} \label{collision entropy}
H_2(Q) := -\log \bigl( \| Q \|_2^2 \bigr) \geq \log \frac{4n\rho}{(1+\rho)^2},
\end{align}
where $\frac{q_{\max}}{q_{\min}} \leq \rho$. This gives
\begin{align} \label{270519d}
\| Q \|_2^2 = \exp\bigl(-H_2(Q)\bigr) \leq \frac{(1+\rho)^2}{4n\rho}.
\end{align}
By Cauchy-Schwartz inequality $\| P \|_2^2 \geq \tfrac1n$ which, together
with \eqref{270519d}, give
\begin{align} \label{270519e}
\| Q \|_2^2 - \| P \|_2^2 \leq \frac{(\rho-1)^2}{4n \rho}.
\end{align}
In view of the Schur-concavity of the R\'{e}nyi entropy (see
\cite[Theorem~13.F.3.a.]{MarshallOA}), the assumption $P \prec Q$
implies that
\begin{align} \label{270519f}
H_2(P) \geq H_2(Q),
\end{align}
and an exponentiation of both sides of \eqref{270519f} (see the left-side
equality in \eqref{collision entropy}) gives
\begin{align} \label{270519g}
\| Q \|_2^2 \geq \| P \|_2^2.
\end{align}
Combining \eqref{270519e} and \eqref{270519f} gives \eqref{bounds on diff. norms}.

\section{Proof of Theorem~\ref{thm: LB/UB f-div}}
\label{appendix: LB/UB f-div}

We prove Item~\ref{Thm. 5-a}), showing that the set $\set{P}_n(\rho)$
(with $\rho \geq 1$) is non-empty, convex and compact. Note that
$\set{P}_n(1) = \{U_n\}$ is a singleton, so the claim is trivial for $\rho=1$.

Let $\rho > 1$. The non-emptiness of $\set{P}_n(\rho)$ is trivial since
$U_n \in \set{P}_n(\rho)$. To prove the convexity of $\set{P}_n(\rho)$,
let $P_1, P_2 \in \set{P}_n(\rho)$, and let
$p_{\max}^{(1)}, \, p_{\max}^{(2)}, \, p_{\min}^{(1)}$ and $p_{\min}^{(2)}$
be the (positive) maximal and minimal probability masses of $P_1$ and $P_2$,
respectively. Then, $\frac{p_{\max}^{(1)}}{p_{\min}^{(1)}} \leq \rho$ and
$\frac{p_{\max}^{(2)}}{p_{\min}^{(2)}} \leq \rho$ yield
\begin{align}
\label{290519b1}
\frac{\lambda p_{\max}^{(1)} + (1-\lambda) p_{\max}^{(2)}}{\lambda p_{\min}^{(1)}
+ (1-\lambda) p_{\min}^{(2)}} \leq \rho, \quad \forall \, \lambda \in [0,1].
\end{align}
For every $\lambda \in [0,1]$,
\begin{align}
\label{290519b2}
& \min_{1 \leq i \leq n} \bigl\{ \lambda P_1(i) + (1-\lambda) P_2(i) \bigr\}
\geq \lambda \, p_{\min}^{(1)} + (1-\lambda) \, p_{\min}^{(2)}, \\[0.1cm]
\label{290519b3}
& \max_{1 \leq i \leq n} \bigl\{ \lambda P_1(i) + (1-\lambda) P_2(i) \bigr\}
\leq \lambda \, p_{\max}^{(1)} + (1-\lambda) \, p_{\max}^{(2)}.
\end{align}
Combining \eqref{290519b1}--\eqref{290519b3} implies that
\begin{align}
\frac{\underset{1 \leq i \leq n}{\max} \bigl\{ \lambda P_1(i)
+ (1-\lambda) P_2(i) \bigr\}}{\underset{1 \leq i \leq n}{\min}
\bigl\{ \lambda P_1(i) + (1-\lambda) P_2(i) \bigr\}} \leq \rho,
\end{align}
so $\lambda P_1 + (1-\lambda) P_2 \in \set{P}_n(\rho)$ for
all $\lambda \in [0,1]$. This proves the convexity of $\set{P}_n(\rho)$.

An alternative proof for Item~\ref{Thm. 5-a} relies on the observation
that, for $\rho \geq 1$,
\begin{align}
\set{P}_n(\rho) = \set{P}_n \bigcap \left\{ \bigcap_{i \neq j} \left\{P: \, P(i) - \rho P(j) \leq 0 \right\} \right\},
\end{align}
which yields the convexity and compactness of the set $\set{P}_n(\rho)$ for all $\rho \geq 1$.

The set of probability mass functions $\set{P}_n(\rho)$ is
clearly bounded; for showing its compactness, it is left to show
that $\set{P}_n(\rho)$ is closed. Let $\rho>1$,
and let $\{P^{(m)}\}_{m=1}^{\infty}$ be a sequence of probability
mass functions in $\set{P}_n(\rho)$ which pointwise converges
to $P$ over the finite set $\set{A}_n$. It is required to show that
$P \in \set{P}_n(\rho) \subseteq \set{P}_n$. As a limit of probability
mass functions, $P \in \set{P}_n$, and since by assumption $P^{(m)}
\in \set{P}_n(\rho)$ for all $m \in \naturals$, it follows that
$$
(n-1) \rho p_{\min}^{(m)} + p_{\min}^{(m)}
\geq (n-1) p_{\max}^{(m)} + p_{\min}^{(m)} \geq 1,
$$
which yields $p_{\min}^{(m)} \geq \frac1{(n-1)\rho+1}$ for all
$m$. Since $p_{\max}^{(m)} \leq \rho p_{\min}^{(m)}$ for every
$m$, it follows that also for the limiting probability mass
function $P$ we have $p_{\min} \geq \frac1{(n-1)\rho+1} > 0$,
and $p_{\max} \leq \rho p_{\min}$. This proves that
$P \in \set{P}_n(\rho)$, and therefore $\set{P}_n(\rho)$ is
a closed set.

The result in Item~\ref{Thm. 5-b}) holds in view of
Item~\ref{Thm. 5-a}), and due to the convexity and continuity
of $D_f(P\|Q)$ in $(P,Q) \in \set{P}_n(\rho) \times \set{P}_n(\rho)$
(where $p_{\min}, \, q_{\min} \geq \frac1{(n-1)\rho+1} > 0$).
This implication is justified by the statement that a convex
and continuous function over a non-empty convex and compact
set attains its supremum over this set (see, e.g.,
\cite[Theorem~7.42]{Beck} or
\cite[Theorem~10.1 and Corollary~32.3.2]{Rockafellar96}).

We next prove Item~\ref{Thm. 5-c}). If $Q \in \set{P}_n(\rho)$,
then $\frac1{1+(n-1)\rho} \leq q_{\min} \leq \frac1n$ where the
lower bound on $q_{\min}$ is attained when $Q$ is the probability
mass function with $n-1$ masses equal to $\rho q_{\min}$ and a
single smaller mass equal to $q_{\min}$, and the upper bound is
attained when $Q$ is the equiprobable distribution. For an arbitrary
$Q \in \set{P}_n(\rho)$, let $q_{\min} := \beta$ where $\beta$
can get any value in the interval $\Gamma_n(\rho)$ defined in
\eqref{Gamma interval}. By \cite[Lemma~1]{Sason18b},
$Q \prec Q_\beta$ and $Q_\beta \in \set{P}_n(\rho)$ where
$Q_\beta$ is given in \eqref{Q_beta}. The Schur-convexity
of $D_f( \cdot \| U_n)$ (see \cite[Lemma~1]{CicaleseGV06})
and the identity $D_f(U_n \| \cdot) = D_{f^\ast}(\cdot \| U_n)$
give that
\begin{align}
\label{300519c3}
D_f(Q \| U_n) \leq D_f(Q_\beta \| U_n), \quad D_f(U_n \| Q) \leq D_f(U_n \| Q_\beta)
\end{align}
for all $Q \in \set{P}_n(\rho)$ with $q_{\min} = \beta \in \Gamma_n(\rho)$; furthermore,
equalities hold in \eqref{300519c3} if $Q = Q_\beta \in \set{P}_n(\rho)$.
The maximization of $D_f(Q \| U_n)$ and $D_f(U_n \| Q)$ over all the probability
mass functions $Q \in \set{P}_n(\rho)$ can be therefore simplified
to the maximization of $D_f(Q_\beta \| U_n)$ and $D_f(U_n \| Q_\beta)$,
respectively, over the parameter $\beta$ which lies in the interval
$\Gamma_n(\rho)$ in \eqref{Gamma interval}. This proves \eqref{opt1 Df: beta}
and \eqref{opt2 Df: beta}.

We next prove Item~\ref{Thm. 5-e}), and then prove Item~\ref{Thm. 5-d}).
In view of Item~\ref{Thm. 5-c}), the maximum of $D_f(Q\|U_n)$ over all the
probability mass functions $Q \in \set{P}_n(\rho)$ is attained by $Q = Q_\beta$
with $\beta \in \Gamma_n(\rho)$ (see \eqref{Gamma interval}--\eqref{i_beta}).
From \eqref{Q_beta}, $Q_\beta$ can be expressed as the $n$-length probability vector
\begin{align}
\label{Q_beta vec}
Q_\beta = ( \, \underbrace{\rho \beta, \ldots, \rho \beta}_{i_\beta},
\, 1-(n+i_\beta \rho - i_\beta - 1) \beta, \, \underbrace{\beta, \ldots,
\beta}_{n-i_\beta-1} \,).
\end{align}
The influence of the $(i_\beta+1)$-th entry of the probability vector in
\eqref{Q_beta vec} on $D_f(Q_\beta \| U_n)$ tends to zero as we
let $n \to \infty$. This holds since the entries of the vector in
\eqref{Q_beta vec} are written in decreasing order, which implies that
for all $\beta \in \Gamma_n(\rho)$ (with $\rho \geq 1$)
\begin{align}
\label{300519e2}
n \bigl[1-(n+i_\beta \rho - i_\beta - 1) \bigr] \in
[n \beta, n \rho \beta] \subseteq \Bigl[\tfrac{n}{(n-1)\rho+1}, \rho \Bigr]
\subseteq \bigl[ \tfrac1\rho, \rho \bigr];
\end{align}
from \eqref{300519e2} and the convexity of $f$ on $(0, \infty)$ (so, $f$ attains
its finite maximum on every closed sub-interval of $(0, \infty)$), it follows  that
\begin{align}
& \Bigl| \bigl[ 1-(n+i_\beta \rho - i_\beta - 1) \beta \bigr] \,
f\bigl( n \bigl[1-(n+i_\beta \rho - i_\beta - 1) \bigr] \bigr) \Bigr| \nonumber \\
& \leq \Bigl| \bigl[ 1-(n+i_\beta \rho - i_\beta - 1) \beta \bigr] \Bigr| \,
\max_{u \in \bigl[\frac1\rho, \rho \bigr]} \, \bigl| f(u) \bigr| \nonumber \\
\label{300519e3}
& \leq \frac{\rho}{n} \, \max_{u \in \bigl[\frac1\rho, \rho \bigr]} \, \bigl| f(u) \bigr|
\underset{n \to \infty}{\longrightarrow} 0.
\end{align}
In view of \eqref{Q_beta vec} and \eqref{300519e3}, by letting $n \to \infty$,
the maximization of $D_f(Q_\beta \| U_n)$ over $\beta \in \Gamma_n(\rho)$ can be
replaced by a maximization of $D_f(\widetilde{Q}_m \| U_n)$ where
\begin{align}
\label{300519e4}
\widetilde{Q}_m := ( \, \underbrace{\rho \beta, \ldots, \rho \beta}_m,
\underbrace{\beta, \ldots, \beta}_{n-m} \, ) \in \set{P}_n(\rho)
\end{align}
with the free parameter $m \in \{0, \ldots, n\}$, and with $\beta := \frac1{n+(\rho-1)m}$
(the value of $\beta$ is determined so that the total mass of $\widetilde{Q}_m$ is~1).
Hence, we get
\begin{align}
\label{310519e1}
\lim_{n \to \infty} \max_{\beta \in \Gamma_n(\rho)} D_f(Q_\beta \| U_n)
= \lim_{n \to \infty} \max_{m \in \{0, \ldots, n\}} D_f(\widetilde{Q}_m \| U_n).
\end{align}
The $f$-divergence in the right side of \eqref{310519e1} satisfies
\begin{align}
\label{300519e5}
D_f(\widetilde{Q}_m \| U_n)
&= \frac1n \sum_{i=1}^n f\bigl(n \, \widetilde{Q}_m(i) \bigr) \\
\label{300519e6}
&= \frac{m}{n} \, f\biggl( \frac{\rho n}{n+(\rho-1)m} \biggr) + \biggl(1-\frac{m}{n}\biggr)
\, f\biggl(\frac{n}{n+(\rho-1)m}\biggr) \\
\label{300519e7}
&= g_f^{(\rho)}\Bigl(\frac{m}{n}\Bigr),
\end{align}
where \eqref{300519e7} holds by the definition of the function $g_f^{(\rho)}(\cdot)$
in \eqref{def: g_f}.
It therefore follows that
\begin{align}
& \lim_{n \to \infty} u_f(n, \rho) \nonumber \\
\label{310519e2}
&= \lim_{n \to \infty} \max_{m \in \{0, \ldots, n\}} g_f^{(\rho)}\Bigl(\frac{m}{n}\Bigr) \\
\label{310519e3}
&= \max_{x \in [0,1]} g_f^{(\rho)}(x)
\end{align}
where \eqref{310519e2} holds by combining \eqref{opt1 Df: beta} and
\eqref{310519e1}--\eqref{300519e7}; \eqref{310519e3} holds by the continuity of the
function $g_f^{(\rho)}(\cdot)$ on $[0,1]$, which follows from \eqref{def: g_f} and
the continuity of the convex function $f$ on $\bigl[\tfrac1\rho, \rho \bigr]$ for
$\rho \geq 1$ (recall that a convex function is continuous on every closed sub-interval
of its domain of region, and by assumption $f$ is convex on $(0, \infty)$).
This proves \eqref{asympt. max1}, by the definition of $g_f^{(\rho)}(\cdot)$ in
\eqref{def: g_f}.

Equality \eqref{asympt. max2} follows from \eqref{asympt. max1}
by replacing $g_f^{(\rho)}(\cdot)$ with $g_{f^\ast}^{(\rho)}(\cdot)$, with
$f^\ast \colon (0, \infty) \to \Reals$ as given in \eqref{dual f};
this replacement is justified by the equality $D_f(U_n \| Q) = D_{f^\ast}(Q \| U_n)$.

Once Item~\ref{Thm. 5-e}) is proved, we return to prove Item~\ref{Thm. 5-d}).
To that end, it is first shown that
\begin{align}
\label{monotonicity u_f}
& u_f(n, \rho) \leq u_f(2n, \rho), \\
\label{monotonicity v_f}
& v_f(n, \rho) \leq v_f(2n, \rho),
\end{align}
for all $\rho \geq 1$ and integers $n \geq 2$,
with the functions $u_f$ and $v_f$, respectively, defined in \eqref{def: u_f} and
\eqref{def: v_f}. Since $D_f(P\|Q) = D_{f^\ast}(Q\|P)$ for all $P, Q \in \set{P}_n$,
\eqref{def: u_f} and \eqref{def: v_f} give that
\begin{align}
\label{u_f, v_f}
v_f(n, \rho) = u_{f^\ast}(n, \rho),
\end{align}
so the monotonicity property in \eqref{monotonicity v_f} follows from
\eqref{monotonicity u_f} by replacing $f$ with $f^\ast$. To prove
\eqref{monotonicity u_f}, let $Q^\ast \in \set{P}_n(\rho)$ be a
probability mass function which attains the maximum at the right side
of \eqref{def: u_f}, and let $P^\ast$ be the probability mass function
supported on $\set{A}_{2n} = \{1, \ldots, 2n\}$, and defined as follows:
\begin{align}
\label{P^ast}
P^\ast(i) =
\begin{dcases}
\tfrac12 Q^\ast(i),   & \quad \mbox{if $i \in \{1, \ldots, n\}$}, \\
\tfrac12 Q^\ast(i-n), & \quad \mbox{if $i \in \{n+1, \ldots, 2n\}$}.
\end{dcases}
\end{align}
Since by assumption $Q^\ast \in \set{P}_n(\rho)$, it is easy to verify
from \eqref{P^ast} that $P^\ast \in \set{P}_{2n}(\rho)$. It therefore
follows that
\begin{align}
\label{310519d1}
u_f(2n, \rho) &= \max_{Q \in \set{P}_{2n}(\rho)} D_f(Q \| U_{2n}) \\
\label{310519d2}
&\geq D_f(P^\ast \| U_{2n}) \\
\label{310519d4}
&= \frac1{2n} \left[ \, \sum_{i=1}^n f\bigl(2n P^\ast(i) \bigr)
+ \sum_{i=n+1}^{2n} f\bigl(2n P^\ast(i) \bigr) \right] \\
\label{310519d5}
&= \frac1{n} \sum_{i=1}^n f\bigl(n Q^\ast(i) \bigr) \\
\label{310519d6}
&= D_f(Q^\ast \| U_n) \\
\label{310519d7}
&= \max_{Q \in \set{P}_n(\rho)} D_f(Q \| U_{2n}) \\
\label{310519d8}
&= u_f(n, \rho)
\end{align}
where \eqref{310519d1} and \eqref{310519d8} hold due to \eqref{def: u_f};
\eqref{310519d2} holds since $P^\ast \in \set{P}_{2n}(\rho)$; finally,
\eqref{310519d5} holds due to \eqref{P^ast},
which implies that the two sums in the right side of \eqref{310519d4}
are identical, and they equal to the sum in the right side of \eqref{310519d5}.
This gives \eqref{monotonicity u_f}, and likewise also \eqref{monotonicity v_f}
(see \eqref{u_f, v_f}).
\begin{align}
\label{310519d9}
u_f(n, \rho) &\leq \lim_{k \to \infty} u_f(2^k n, \rho) \\
\label{310519d10}
&= \lim_{n' \to \infty} u_f(n', \rho) \\
\label{310519d11}
&= \max_{x \in [0,1]} g_f^{(\rho)}(x)
\end{align}
where \eqref{310519d9} holds since, due to \eqref{monotonicity u_f},
the sequence $\{u_f(2^k n, \rho)\}_{k=0}^{\infty}$ is monotonically
increasing, which implies that the first term of this sequence is
less than or equal to its limit. Equality~\eqref{310519d10} holds
since the limit in its right side exists (in view of the above
proof of \eqref{asympt. max1}), so its limit coincides with the
limit of every subsequence; \eqref{310519d11} holds due to
\eqref{310519e2} and \eqref{310519e3}. A replacement of $f$ with
$f^\ast$ gives, from \eqref{u_f, v_f}, that
\begin{align}
\label{310519d12}
v_f(n, \rho) \leq \underset{x \in [0,1]}{\max} g_{f^\ast}^{(\rho)}(x).
\end{align}
Combining \eqref{310519d9}--\eqref{310519d12} gives the right-side
inequalities in \eqref{LB/UB u_f} and \eqref{LB/UB v_f}.

The left-side inequality in \eqref{LB/UB u_f} follows by combining
\eqref{def: u_f}, \eqref{300519e4} and \eqref{300519e5}--\eqref{300519e7},
which gives
\begin{align}
u_f(n, \rho) &= \max_{Q \in \set{P}_n(\rho)} D_f(Q \| U_n) \\
&\geq \max_{m \in \{0, \ldots, n\}} D_f(\widetilde{Q}_m \| U_n) \\
&= \max_{m \in \{0, \ldots, n\}} g_f^{(\rho)}\Bigl(\frac{m}{n}\Bigr).
\end{align}
Likewise, in view of \eqref{u_f, v_f}, the left-side inequality in
\eqref{LB/UB v_f} follows from the left-side inequality in \eqref{LB/UB u_f}
by replacing $f$ with $f^\ast$.

We next prove Item~\ref{Thm. 5-f}), providing an upper bound on the
convergence rate of the limit in \eqref{asympt. max1}; an analogous result
can be obtained for the convergence rate to the limit in \eqref{asympt. max2}
by replacing $f$ with $f^\ast$ in \eqref{dual f}. To prove \eqref{convergence rate 1},
in view of Items~\ref{Thm. 5-d}) and~\ref{Thm. 5-e}), we get that for
every integer $n \geq 2$
\begin{align}
\label{02062019a1}
0 & \leq \lim_{n' \to \infty} \left\{ u_f(n', \rho) \right\} - u_f(n, \rho) \\
\label{02062019a2}
& \leq \max_{x \in [0,1]} g_f^{(\rho)}(x) - \max_{m \in \{0, \ldots, n\}}
g_f^{(\rho)}\Bigl(\frac{m}{n}\Bigr) \\
\label{02062019a3}
&= \max_{x \in [0,1]} g_f^{(\rho)}(x) - \max_{m \in \{0, \ldots, n-1\}}
g_f^{(\rho)}\Bigl(\frac{m}{n}\Bigr) \\
\label{02062019a4}
&= \max_{m \in \{0, \ldots, n-1\}} \Biggl\{ \max_{x \in \bigl[
\frac{m}{n}, \frac{m+1}{n} \bigr]} g_f^{(\rho)}(x) \Biggr\}
- \max_{m \in \{0, \ldots, n-1\}} g_f^{(\rho)}\Bigl(\frac{m}{n}\Bigr) \\
\label{02062019a5}
&\leq \max_{m \in \{0, \ldots, n-1\}} \Biggl\{ \max_{x \in
\bigl[ \frac{m}{n}, \frac{m+1}{n} \bigr]} \Bigl\{ g_f^{(\rho)}(x)
- g_f^{(\rho)}\Bigl(\frac{m}{n}\Bigr) \Bigr\} \Biggr\}
\end{align}
where \eqref{02062019a1} holds due to monotonicity property in \eqref{monotonicity u_f},
and also due to the existence of the limit of $\{u_f(n', \rho)\}_{n' \in \naturals}$;
\eqref{02062019a2} holds due to \eqref{LB/UB u_f};
\eqref{02062019a3} holds since the function $g_f^{(\rho)} \colon [0,1] \to \Reals$
(as defined in \eqref{def: g_f}) satisfies $g_f^{(\rho)}(1)=g_f^{(\rho)}(0)= 0$ (recall
that by assumption $f(1)=0$); \eqref{02062019a4} holds since
$[0, 1] = \overset{n-1}{\underset{m=1}{\bigcup}} \bigl[\frac{m}{n}, \frac{m+1}{n}\bigr]$,
so the maximization of $g_f^{(\rho)}(\cdot)$ over the interval $[0,1]$ is the maximum over
the maximal values over the sub-intervals $\bigl[\frac{m}{n}, \frac{m+1}{n}\bigr]$ for
$m \in \{0, \ldots, n-1\}$; finally, \eqref{02062019a5} holds since the
maximum of a sum of functions is less than or equal to the sum of the maxima of these
functions.
If the function $g_f^{(\rho)} \colon [0,1] \to \Reals$ is differentiable on $(0,1)$, and
its derivative is upper bounded by $K_f(\rho) \geq 0$, then by the mean value theorem of
Lagrange, for every $m \in \{0, \ldots, n-1\}$,
\begin{align}
\label{02062019a6}
g_f^{(\rho)}(x) - g_f^{(\rho)}\Bigl(\frac{m}{n}\Bigr) \leq \frac{K_f(\rho)}{n}, \quad
\forall \, x \in \left[ \frac{m}{n}, \frac{m+1}{n} \right].
\end{align}
Combining \eqref{02062019a1}--\eqref{02062019a6} gives \eqref{convergence rate 1}.

We next prove Item~\ref{Thm. 5-g}). By definition, it readily follows that
$\set{P}_n(\rho_1) \subseteq \set{P}_n(\rho_2)$ if $1 \leq \rho_1 < \rho_2$.
By the definition in \eqref{def: u_f}, for a fixed integer $n \geq 2$,
it follows that the function $u_f(n, \cdot)$ is monotonically increasing
on $[1, \infty)$. The limit in the left side of \eqref{asympt. max3} therefore
exists. Since $D_f(Q \| U_n)$ is convex in $Q$, its maximum over the convex
set of probability mass functions $Q \in \set{P}_n$ is obtained at one of the
vertices of the simplex $\set{P}_n$. Hence, a maximum of $D_f(Q \| U_n)$
over this set is attained at $Q^\ast = (q_1^\ast, \ldots, q_n^\ast)$
with $q_i^\ast = 1$ for some $i \in \{1, \ldots, n\}$, and $q_j^\ast = 0$
for $j \neq i$. In the latter case,
\begin{align}
D_f(Q^\ast \| U_n) = \frac1n \sum_{k=1}^n f(n q_k^\ast)
= \frac1n \, \bigl[ (n-1) f(0) + f(n) \bigr].
\end{align}
Note that $Q^\ast \notin \underset{\rho \geq 1} \bigcup \set{P}_n(\rho)$
(since the union of $\{\set{P}_n(\rho)\}$, for all $\rho \geq 1$, includes
all the probability mass functions in $\set{P}_n$ which are {\em supported}
on $\set{A}_n = \{1, \ldots, n\}$, so $Q^\ast \in \set{P}_n$ is not an
element of this union); hence, it follows that
\begin{align}
\label{02062019b1}
\lim_{\rho \to \infty} u_f(n, \rho)
\leq \biggl(1 - \frac1n\biggr) f(0) + \frac{f(n)}{n}.
\end{align}
On the other hand, for every $\rho \geq 1$,
\begin{align}
\label{02062019b2}
u_f(n, \rho) &
\geq g_f^{(\rho)}\biggl(\frac1n\biggr) \\
\label{02062019b4}
&= \frac1n \, f\left(\frac{\rho n}{n+\rho-1}\right)
+ \left(1-\frac1n\right) \, f\left(\frac{n}{n+\rho-1}\right)
\end{align}
where \eqref{02062019b2} holds due to the left-side inequality
of \eqref{LB/UB u_f},
and \eqref{02062019b4} is due to \eqref{def: g_f}. Combining
\eqref{02062019b2}--\eqref{02062019b4}, and the continuity of
$f$ at zero (by the continuous extension of the convex function
$f$ at zero), yields (by letting $\rho \to \infty$)
\begin{align}
\label{02062019b5}
\lim_{\rho \to \infty} u_f(n, \rho)
\geq \biggl(1 - \frac1n\biggr) f(0) + \frac{f(n)}{n}.
\end{align}
Combining \eqref{02062019b1} and \eqref{02062019b5} gives \eqref{asympt. max3}
for every integer $n \geq 2$. In order to get an upper bound on the convergence
rate in \eqref{asympt. max3}, suppose that $f(0) < \infty$, $f$ is differentiable
on $(0,n)$, and $ K_n := \underset{t \in (0,n)}{\sup} \, \bigl| f'(t) \bigr| < \infty$.
For every $\rho \geq 1$, we get
\begin{align}
\label{02062019c1}
0 & \leq \lim_{\rho' \to \infty} \left\{ u_f(n, \rho') \right\} - u_f(n, \rho) \\[0.1cm]
\label{02062019c2}
& \leq \frac1n \biggl[ f(n) - f\biggl( \frac{\rho n}{n+\rho-1} \biggr) \biggr]
+ \biggl(1 - \frac1n \biggr) \left[f(0) - f\biggl(\frac{n}{n+\rho-1}\biggr) \right] \\[0.1cm]
\label{02062019c3}
& \leq \frac{K_n}{n} \left(n - \frac{\rho n}{n+\rho-1} \right) +
\left(1-\frac1n\right) \frac{K_n \, n}{n+\rho-1} \\[0.1cm]
\label{02062019c4}
&= \frac{2K_n \; (n-1)}{n+\rho-1},
\end{align}
where \eqref{02062019c1} holds since the sets $\{\set{P}_n(\rho)\}_{\rho \geq 1}$
are monotonically increasing in $\rho$; \eqref{02062019c2} follows from
\eqref{02062019b1}--\eqref{02062019b4}; \eqref{02062019c3} holds by the
assumption that $\bigl|f'(t)\bigr| \leq K_n$ for all $t \in (0,n)$, by the
mean value theorem of Lagrange, and since
$0 < \frac{n}{n+\rho-1} \leq \frac{\rho n}{n+\rho-1} \leq n$
for all $\rho \geq 1$ and $n \in \naturals$. This proves \eqref{convergence rate 2}.

We next prove Item~\ref{Thm. 5-h}). Setting $P := U_n$ yields $P \prec Q$
for every probability mass function $Q$ which is supported on $\{1, \ldots, n\}$.
Since $q_{\min} + (n-1) q_{\max} \geq 1$ and also $(n-1) q_{\min} + q_{\max} \leq 1$,
and since by assumption $\frac{q_{\max}}{q_{\min}} \leq \rho$, it follows that
\begin{align}
\label{intervals}
[nq_{\min}, \, nq_{\max}] &
\subseteq \biggl[ \frac{n}{1+(n-1)\rho}, \; \frac{\rho n}{n-1+\rho} \biggr]
\subseteq \biggl[ \frac1\rho, \, \rho \biggr].
\end{align}
Combining the assumption in \eqref{f'' bounded} with \eqref{intervals}
implies that
\begin{align} \label{bounded f''}
m \leq f''(t) \leq M, \quad \forall \, t \in [nq_{\min}, \, nq_{\max}].
\end{align}
Hence, \eqref{c_f}, \eqref{e_f} and \eqref{bounded f''} yield
\begin{align}
\label{bounds c_f and e_f}
\tfrac12 \, m \leq c_f(n q_{\min}, n q_{\max}) \leq
e_f(n q_{\min}, n q_{\max}) \leq \tfrac12 \, M.
\end{align}
The lower bound on $D_f(Q \| U_n)$ in the left side of \eqref{LB Df-m} follows
from a combination of \eqref{LB diff Df - equiprob.}, the left-side inequality
in \eqref{bounds c_f and e_f}, and $\|P\|_2^2 = \frac1n$. Similarly, the
upper bound on $D_f(Q \| U_n)$ in the right side of \eqref{UB0 Df-M} follows
from a combination of \eqref{UB diff Df - equiprob.}, the right-side inequality
in \eqref{bounds c_f and e_f}, and the equality $\|P\|_2^2 = \frac1n$.
The looser upper bound on $D_f(Q \| U_n)$ in the right side of \eqref{UB Df-M},
expressed as a function of $M$ and $\rho$, follows by combining
\eqref{UB diff Df - equiprob.}, \eqref{bounds on diff. norms}, and the right-side
inequality in \eqref{bounds c_f and e_f}.

The tightness of the lower bound in the left side of \eqref{LB Df-m}
and the upper bound in the right side of \eqref{UB0 Df-M} for
the $\chi^2$ divergence is clear from the fact that $M=m=2$
if $f(t)=(t-1)^2$ for all $t>0$; in this case,
$\chi^2(Q \| U_n) = n \| Q \|_2^2 - 1$.

To prove Item~\ref{Thm. 5-i}), suppose that the second derivative
of $f$ is upper bounded on $(0, \infty)$ with $f''(t) \leq M_f
\in (0, \infty)$ for all $t>0$, and there is a need to assert
that $D_f(Q \| U_n) \leq d$ for an arbitrary $d>0$. Condition
\eqref{UB rho} follows from \eqref{UB Df-M} by solving the inequality
$\frac{M_f \, (\rho-1)^2}{8 \rho} \leq d$,
with the variable $\rho \geq 1$, for given $d > 0$ and
$M_f > 0$ (note that $M_f$ does not depend on $\rho$).

\section{Proof of Theorem~\ref{thm: bounds Tsallis}}
\label{appendix: bounds Tsallis}

The proof of Theorem~\ref{thm: bounds Tsallis} relies on Theorem~\ref{thm: majorization Df}.
For an arbitrary $\alpha \in (0,1) \cup (1, \infty)$, let $u_\alpha \colon (0,\infty) \to \Reals$
be the non-negative and convex function given by (see, e.g., \cite[(2.1)]{LieseV_book87}
or \cite[(17)]{LieseV_IT2006})
\begin{align}
\label{02062019d1}
u_\alpha(t) := \frac{t^\alpha - \alpha (t-1) - 1}{\alpha (\alpha-1)}, \quad t > 0,
\end{align}
and let $u_1 \colon (0, \infty) \to \Reals$ be the convex function given by
\begin{align}
\label{02062019d2}
u_1(t) := \lim_{\alpha \to 1} u_\alpha(t) = t \log_\mathrm{e} t + 1-t, \quad t > 0.
\end{align}
Let $P$ and $Q$ be probability mass functions which are supported on a finite set;
without loss of generality, let their support be given by
$\set{A}_n := \{1, \ldots, n\}$. Then,
\begin{align}
& D_{u_\alpha}(Q \| U_n) - D_{u_\alpha}(P \| U_n) \nonumber \\[0.1cm]
& = \frac1n \sum_{i=1}^n u_\alpha\bigl(n Q(i) \bigr) -
\frac1n \sum_{i=1}^n u_\alpha\bigl(n P(i) \bigr) \nonumber \\[0.1cm]
& = \frac{n^{\alpha-1}}{\alpha (\alpha-1)} \left[ \,
\sum_{i=1}^n Q^\alpha(i) - \sum_{i=1}^n P^\alpha(i) \right] \nonumber \\[0.1cm]
\label{02062019d3}
& = \frac{n^{\alpha-1} \bigl[S_\alpha(P) - S_\alpha(Q) \bigr]}{\alpha},
\end{align}
where
\begin{align}
\label{Tsallis entropy}
S_\alpha(P) :=
\begin{dcases}
\frac1{1-\alpha} \left( \sum_{i=1}^n P^\alpha(i) - 1 \right),
& \quad \alpha \in (0,1) \cup (1, \infty), \\
-\sum_{i=1}^n  P(i) \, \log_{\mathrm{e}} P(i), & \quad \alpha = 1.
\end{dcases}
\end{align}
designates the order-$\alpha$ Tsallis entropy of a probability mass $P$
defined on the set $\set{A}_n$.
Equality \eqref{02062019d3} also holds for $\alpha=1$ by continuous extension.

In view of \eqref{c_f} and \eqref{e_f}, since $u_\alpha''(t) = t^{\alpha-2}$
for all $t > 0$, it follows that
\begin{align}
\label{02062019d4}
c_{u_\alpha}(n q_{\min}, \, n q_{\max}) &=
\begin{dcases}
\tfrac12 \, n^{\alpha-2} \, q_{\max}^{\alpha-2},
& \quad \mbox{if $\alpha \in (0,2]$}, \\[0.1cm]
\tfrac12 \, n^{\alpha-2} \, q_{\min}^{\alpha-2},
& \quad \mbox{if $\alpha \in (2, \infty)$},
\end{dcases}
\end{align}
and
\begin{align}
\label{02062019d5}
e_{u_\alpha}(n q_{\min}, \, n q_{\max}) &=
\begin{dcases}
\tfrac12 \, n^{\alpha-2} \, q_{\min}^{\alpha-2},
& \quad \mbox{if $\alpha \in (0,2]$}, \\[0.1cm]
\tfrac12 \, n^{\alpha-2} \, q_{\max}^{\alpha-2},
& \quad \mbox{if $\alpha \in (2, \infty)$}.
\end{dcases}
\end{align}
The combination of \eqref{UB diff Df - equiprob.} and \eqref{LB diff Df - equiprob.}
under the assumption that $P$ and $Q$ are supported on $\set{A}_n$ and
$P \prec Q$, together with \eqref{02062019d3}, \eqref{02062019d4} and
\eqref{02062019d5} gives \eqref{bounds Tsallis}--\eqref{UB Tsallis}.
Furthermore, the left and right-side inequalities in \eqref{bounds Tsallis}
hold with equality if $c_{u_\alpha}(\cdot, \cdot)$ in \eqref{02062019d4}
and $e_{u_\alpha}(\cdot, \cdot)$ in \eqref{02062019d5} coincide, which implies that
the upper and lower bounds in \eqref{UB diff Df - equiprob.} and
\eqref{LB diff Df - equiprob.} are tight in that case. Comparing
$c_{u_\alpha}(\cdot, \cdot)$ in \eqref{02062019d4} and $e_{u_\alpha}(\cdot, \cdot)$
in \eqref{02062019d5} shows that they coincide if $\alpha=2$.

To prove Item~\ref{Thm. 6.b}) of Theorem~\ref{thm: bounds Tsallis},
let  $P_\varepsilon$ and $Q_\varepsilon$ be probability mass functions
supported on $\set{A} = \{0, 1\}$ where
$P_\varepsilon(0) = \tfrac12 + \varepsilon$,
$Q_\varepsilon(0) = \tfrac12 + \beta \varepsilon$,
and $\beta > 1$ and $0 < \varepsilon < \frac1{2 \beta}$. This yields
$P_\varepsilon \prec Q_\varepsilon$. The result in \eqref{inf,sup}
is proved by showing that, for all $\alpha > 0$,
\begin{align}
\label{limit 1 - Tsallis}
& \lim_{\varepsilon \to 0^+} \frac{S_\alpha(P_\varepsilon) -
S_\alpha(Q_\varepsilon)}{L(\alpha, P_\varepsilon, Q_\varepsilon)} = 1, \\[0.1cm]
\label{limit 2 - Tsallis}
& \lim_{\varepsilon \to 0^+} \frac{S_\alpha(P_\varepsilon) -
S_\alpha(Q_\varepsilon)}{U(\alpha, P_\varepsilon, Q_\varepsilon)} = 1,
\end{align}
which shows that the infimum and supremum in \eqref{inf,sup} can be even
restricted to the binary alphabet setting.
For every $\alpha \in (0,1) \cup (1, \infty)$,
\begin{align}
S_\alpha(P_\varepsilon) - S_\alpha(Q_\varepsilon)
& = \frac1{1-\alpha} \left( \sum_i P_\varepsilon^\alpha(i)
- \sum_i Q_\varepsilon^\alpha(i) \right) \nonumber \\
& = \frac1{1-\alpha} \Biggl[ \left(\tfrac12 + \varepsilon \right)^\alpha
+ \left(\tfrac12 - \varepsilon \right)^\alpha
- \left(\tfrac12 + \beta \varepsilon \right)^\alpha
- \left(\tfrac12 - \beta \varepsilon \right)^\alpha \Biggr] \nonumber \\
\label{03062019a1}
& = \alpha 2^{2-\alpha} (\beta^2 - 1) \varepsilon^2 + O\bigl(\varepsilon^4\bigr),
\end{align}
where \eqref{03062019a1} follows from a Taylor series expansion
around $\varepsilon=0$, and the passage in the limit where $\alpha \to 1$ shows
that \eqref{03062019a1} also holds at $\alpha=1$ (due to the continuous extension
of the order-$\alpha$ Tsallis entropy at $\alpha=1$). This implies that
\eqref{03062019a1} holds for all $\alpha>0$. We now calculate the lower and
upper bounds on $S_\alpha(P_\varepsilon) - S_\alpha(Q_\varepsilon)$ in
\eqref{LB Tsallis} and \eqref{UB Tsallis}, respectively.
\begin{enumerate}[1)]
\item
For $\alpha \in (0,2]$,
\begin{align}
L(\alpha, P_\varepsilon, Q_\varepsilon)
& = \tfrac12 \, \alpha q_{\max}^{\alpha-2} \, \bigl( \| Q_\varepsilon \|_2^2
- \| P_\varepsilon \|_2^2 \bigr) \nonumber \\
& = \tfrac12 \, \alpha \left(\tfrac12 + \beta \varepsilon \right)^{\alpha-2}
\left[ \left( \tfrac12 + \beta \varepsilon \right)^2 + \left( \tfrac12
- \beta \varepsilon \right)^2 - \left( \tfrac12 + \varepsilon \right)^2
- \left( \tfrac12 - \varepsilon \right)^2 \right] \nonumber \\
\label{03062019a2}
& = \alpha 2^{2-\alpha} (\beta^2-1) (1+2\beta \varepsilon)^{\alpha-2}.
\end{align}

\item
For $\alpha \in (2, \infty)$,
\begin{align}
L(\alpha, P_\varepsilon, Q_\varepsilon)
& = \tfrac12 \, \alpha q_{\min}^{\alpha-2} \, \bigl( \| Q_\varepsilon \|_2^2
- \| P_\varepsilon \|_2^2 \bigr) \nonumber \\
\label{03062019a3}
& = \alpha 2^{2-\alpha} (\beta^2-1) (1-2\beta \varepsilon)^{\alpha-2}.
\end{align}

\item
Similarly, for $\alpha \in (0,2]$,
\begin{align}
U(\alpha, P_\varepsilon, Q_\varepsilon)
& = \tfrac12 \, \alpha q_{\min}^{\alpha-2} \, \bigl( \| Q_\varepsilon \|_2^2
- \| P_\varepsilon \|_2^2 \bigr) \nonumber \\
\label{03062019a4}
& = \alpha 2^{2-\alpha} (\beta^2-1) (1-2\beta \varepsilon)^{\alpha-2},
\end{align}
and, for $\alpha \in (2, \infty)$,
\begin{align}
U(\alpha, P_\varepsilon, Q_\varepsilon)
& = \tfrac12 \, \alpha q_{\max}^{\alpha-2} \, \bigl( \| Q_\varepsilon \|_2^2
- \| P_\varepsilon \|_2^2 \bigr) \nonumber \\
\label{03062019a5}
& = \alpha 2^{2-\alpha} (\beta^2-1) (1+2\beta \varepsilon)^{\alpha-2}.
\end{align}
\end{enumerate}
The combination of \eqref{03062019a1}--\eqref{03062019a3} yields
\eqref{limit 1 - Tsallis}; similarly, the combination of
\eqref{03062019a1}, \eqref{03062019a4} and \eqref{03062019a5}
yields \eqref{limit 2 - Tsallis}.

\section{Proof of Theorem~\ref{theorem: Delta_alpha} and Corollary~\ref{corollary: Delta}}
\label{appendix: Delta_alpha}

\subsection{Proof of Theorem~\ref{theorem: Delta_alpha}}
\label{Proof of Theorem - Delta}
The proof of the convexity property of $\Delta(\cdot, \rho)$ in \eqref{Delta 2: alpha-div.},
with $\rho>1$, over the real line $\Reals$ relies on \cite[Theorem~2.1]{Simic07} which
states that if $W$ is a non-negative random variable, then
\begin{align}
\label{eq: log-convex lambda}
\lambda_\alpha :=
\begin{dcases}
\frac{\bigl(\mathbb{E}[W^\alpha] - \mathbb{E}^\alpha[W] \bigr)
\, \log \mathrm{e}}{\alpha(\alpha-1)}, & \quad \alpha \neq 0, 1 \\[0.1cm]
\log \bigl( \mathbb{E}[W] \bigr) - \mathbb{E}[\log W], & \quad \alpha=0 \\[0.1cm]
\mathbb{E}[W \log W] - \mathbb{E}[W] \, \log \bigl(\mathbb{E}[W]\bigr),
& \quad \alpha=1
\end{dcases}
\end{align}
is log-convex in $\alpha \in \Reals$. This property has been used to derive
$f$-divergence inequalities (see, e.g., \cite[Theorem~20]{ISSV16}, \cite{Simic07}
and \cite{Simic15}).

Let $Q \ll P$, and let $W := \frac{\mathrm{d}Q}{\mathrm{d}P}$ be the Radon-Nikodym
derivative ($W$ is a non-negative random variable). Let the expectations in the right
side of \eqref{eq: log-convex lambda} be taken with respect to $P$. In view of the
above statement from \cite[Theorem~2.1]{Simic07}, this gives the log-convexity of
$D_{\mathrm{A}}^{(\alpha)}(Q\|P)$ in $\alpha \in \Reals$.
Since log-convexity yields convexity, it follows that $D_{\mathrm{A}}^{(\alpha)}(Q\|P)$
is convex in $\alpha$ over the real line. Let $P := U_n$, and let $Q \in \set{P}_n(\rho)$;
since $Q \ll P$, it follows that $D_{\mathrm{A}}^{(\alpha)}(Q\|U_n)$ is convex in
$\alpha \in \Reals$. The pointwise maximum of a set of convex functions is a convex function,
which implies that $\underset{Q \in \set{P}_n(\rho)}{\max} D_{\mathrm{A}}^{(\alpha)}(Q\|U_n)$
is convex in $\alpha \in \Reals$ for every integer $n \geq 2$. Since the pointwise limit of
a convergent sequence of convex functions is a convex function, it follows that
$\underset{n \to \infty}{\lim} \, \underset{Q \in \set{P}_n(\rho)}{\max}
D_{\mathrm{A}}^{(\alpha)}(Q\|U_n)$ is convex in $\alpha$. This, by definition, is equal
to $\Delta(\alpha, \rho)$ (see \eqref{Delta def 1}), which proves the convexity of this
function in $\alpha$ over the real line.

From \eqref{Delta 2: alpha-div.}, for all $\rho > 1$,
\begin{align}
\vspace*{0.1cm}
\Delta(1+\alpha, \rho)
&= \frac1{(\alpha+1) \alpha} \left[ \frac{(-\alpha)^\alpha
\bigl(\rho^{1+\alpha}-1\bigr)^{1+\alpha} \bigl(\rho -
\rho^{1+\alpha}\bigr)^{-\alpha}}{(\rho-1)(1+\alpha)^{1+\alpha}} - 1 \right] \nonumber \\[0.1cm]
&= \frac1{(-\alpha)(-\alpha-1)} \left[ \frac{(1+\alpha)^{-\alpha-1} \bigl(\rho^{1+\alpha}-1\bigr)^{1+\alpha}
\bigl(\rho - \rho^{1+\alpha} \bigr)^{-\alpha}}{(\rho-1)(-\alpha)^{-\alpha}} - 1 \right] \nonumber \\[0.1cm]
&= \frac1{(-\alpha)(-\alpha-1)} \left[ \frac{(1+\alpha)^{-\alpha-1}
\bigl(\rho^\alpha (\rho - \rho^{-\alpha}) \bigr)^{1+\alpha} \bigl(\rho^{1+\alpha} (\rho^{-\alpha} - 1)
\bigr)^{-\alpha}}{(\rho-1)(-\alpha)^{-\alpha}} - 1 \right] \nonumber \\
&= \frac1{(-\alpha)(-\alpha-1)} \left[ \frac{(1+\alpha)^{-\alpha-1} \bigl(\rho - \rho^{-\alpha}\bigr)^{1+\alpha}
\bigl(\rho^{-\alpha} - 1 \bigr)^{-\alpha}}{(\rho-1)(-\alpha)^{-\alpha}} - 1 \right] \nonumber \\
&= \Delta(-\alpha, \rho),
\end{align}
which proves the symmetry property of $\Delta(\alpha, \rho)$ around $\alpha = \tfrac12$ for all $\rho > 1$.
The convexity in $\alpha$ over the real line, and the symmetry around $\alpha = \tfrac12$ implies that
$\Delta(\alpha, \rho)$ gets its global minimum at $\alpha = \tfrac12$, which is equal to
$\frac{4(\sqrt[4]{\rho}-1)^2}{\sqrt{\rho}+1}$ for all $\rho > 1$.

Inequalities \eqref{mon. 1} and \eqref{mon. 2} follow from \cite[Proposition~2.7]{LieseV_book87};
this proposition implies that, for every integer $n \geq 2$ and for all probability mass functions
$Q$ defined on $\set{A}_n := \{1, \ldots, n\}$,
\begin{align}
\label{mon. 1-b}
& \alpha \, D_{\mathrm{A}}^{(\alpha)}(Q\|U_n) \leq \beta \, D_{\mathrm{A}}^{(\beta)}(Q\|U_n),
\hspace*{3.5cm} 0 < \alpha \leq \beta < \infty, \\
\label{mon. 2-b}
& (1-\beta) \, D_{\mathrm{A}}^{(1-\beta)}(Q\|U_n) \leq (1-\alpha) \,
D_{\mathrm{A}}^{(1-\alpha)}(Q\|U_n), \quad -\infty < \alpha \leq \beta < 1.
\end{align}
Inequalities \eqref{mon. 1} and \eqref{mon. 2} follow, respectively, by maximizing both
sides of \eqref{mon. 1-b} or \eqref{mon. 2-b} over $Q \in \set{P}_n(\rho)$, and letting
$n$ tend to infinity.

For every $\alpha \in \Reals$, the function $\Delta(\alpha, \rho)$ is monotonically increasing
in $\rho \in (1, \infty)$ since (by definition) the set of probability mass functions
$\{\set{P}_n(\rho)\}_{\rho \geq 1}$ is monotonically increasing (i.e., $\set{P}_n(\rho_1)
\subseteq \set{P}_n(\rho_2)$ if $1 \leq \rho_1 < \rho_2 < \infty$), and therefore the maximum
of $D_{\mathrm{A}}^{(\alpha)}(Q \| U_n)$ over $Q \in \set{P}_n(\rho)$ is a monotonically
increasing function of $\rho \in [1, \infty)$; the limit of this maximum, as we let $n \to \infty$,
is equal to $\Delta(\alpha, \rho)$ in \eqref{Delta 2: alpha-div.} for all $\rho > 1$, which
is therefore monotonically increasing in $\rho$ over the interval $(1, \infty)$. The continuity of
$\Delta(\alpha, \rho)$ in both $\alpha$ and $\rho$ is due to its expression in \eqref{Delta 2: alpha-div.}
with its continuous extension at $\alpha=0$ and $\alpha=1$ in \eqref{continuous extension of Delta}. Since
$\set{P}_n(1) = \{ U_n\}$, it follows from the continuity of $\Delta(\alpha, \rho)$ that
$$\underset{\rho \to 1^+}{\lim} \Delta(\alpha, \rho) = D_{\mathrm{A}}^{(\alpha)}(U_n \| U_n) = 0.$$

\subsection{Proof of Corollary~\ref{corollary: Delta}}
\label{Proof of Corollary - Delta}
For all $\alpha \in \Reals$ and $\rho > 1$,
\begin{align}
& \lim_{n \to \infty} \max_{Q \in \set{P}_n(\rho)} D_{\mathrm{A}}^{(\alpha)}(U_n \| Q) \nonumber \\
\label{05062019-b1}
&= \lim_{n \to \infty} \max_{Q \in \set{P}_n(\rho)} D_{\mathrm{A}}^{(1-\alpha)}(Q \| U_n) \\
\label{05062019-b2}
&= \Delta(1-\alpha, \rho) \\
\label{05062019-b3}
&= \Delta(\alpha, \rho),
\end{align}
where \eqref{05062019-b1} holds due to the symmetry property in
\cite[p.~36]{LieseV_book87}, which states that
\begin{align}
D_{\mathrm{A}}^{(\alpha)}(P\|Q) = D_{\mathrm{A}}^{(1-\alpha)}(Q\|P),
\end{align}
for every $\alpha \in \Reals$ and probability mass functions $P$ and $Q$;
\eqref{05062019-b2} is due to \eqref{Delta def 1}; finally, \eqref{05062019-b3}
holds due to the symmetry property of $\Delta(\cdot, \rho)$ around $\tfrac12$
in Theorem~\ref{theorem: Delta_alpha}~\ref{Thm. 7-a}).

\section{Proof of \eqref{280519c}}
\label{appendix: Lambert-W}
In view of \eqref{10062019c1} and \eqref{04062019b}, it follows that the condition in
\eqref{10062019a} is satisfied if and only if $\rho \leq \rho^\ast$ where
$\rho^\ast \in (1, \infty)$ is the solution of the equation
\begin{align}
\label{10062019d1}
\frac{\rho^\ast \log \rho^\ast}{\rho^\ast-1} -
\log \left( \frac{\mathrm{e} \rho^\ast \log_{\mathrm{e}}
\rho^\ast}{\rho^\ast-1} \right) = d \log \mathrm{e}.
\end{align}
with a fixed $d>0$. The substitution
\begin{align}
\label{10062019d2}
x := \frac{\rho^\ast \log_{\mathrm{e}} \rho^\ast}{\rho^\ast-1}
\end{align}
leads to the equation
\begin{align}
\label{10062019d3}
x - \log_{\mathrm{e}} x = d+1.
\end{align}
Negation and exponentiation of both sides of \eqref{10062019d3} gives
\begin{align}
\label{10062019d4}
(-x) \mathrm{e}^{-x} = -\mathrm{e}^{-d-1}.
\end{align}
Since $\rho^\ast > 1$ implies by \eqref{10062019d2} that $x>1$, the
proper solution for $x$ is given by
\begin{align}
\label{10062019d5}
x = -W_{-1}\bigl(-\mathrm{e}^{-d-1}\bigr), \quad d > 0,
\end{align}
where $W_{-1}$ denotes the secondary real branch of the Lambert $W$ function
\cite{Corless96}; otherwise, the replacement of $W_{-1}$ in the right side of
\eqref{10062019d5} with the principal real branch $W_0$ yields $x \in (0,1)$.

We next proceed to solve $\rho^\ast$ as a function of $x$. From
\eqref{10062019d2}, letting $u := \frac1{\rho^\ast}$ gives the equation
$u = \mathrm{e}^{(u-1)x}$, which is equivalent to
\begin{align}
\label{10062019d7}
(-ux) \mathrm{e}^{-ux} &= -x \mathrm{e}^{-x} \\
\label{10062019d8}
&= -\mathrm{e}^{-d-1},
\end{align}
where \eqref{10062019d8} follows from \eqref{10062019d5} and by the definition of
the Lambert $W$ function (i.e., $t = W(u)$ if and only if $t \mathrm{e}^t = u$).
The solutions of \eqref{10062019d7} are given by
\begin{align}
\label{10062019d9}
-ux = W_{-1/ 0}\bigl(-\mathrm{e}^{-d-1}\bigr),
\end{align}
which (from \eqref{10062019d5}) correspond, respectively, to $u=1$ and
\begin{align}
\label{10062019d10}
u = \frac{W_{0}\bigl(-\mathrm{e}^{-d-1}\bigr)}{W_{-1}\bigl(-\mathrm{e}^{-d-1}\bigr)} \in (0,1).
\end{align}
Since $\rho^\ast \in (1, \infty)$ is equal to $\frac1u$, the reciprocal of the right side of
\eqref{10062019d10} gives the proper solution for $\rho^\ast$ (denoted by
$\rho_{\max}^{(1)}(d)$ in \eqref{280519c}).

\section{Proof of \eqref{Phi - UB1}, \eqref{Phi - UB2} and \eqref{Phi - UB3}}
\label{appendix: Phi - UBs}

We first derive the upper bound on $\Phi(\alpha, \rho)$ in \eqref{Phi - UB1}
for $\alpha \geq \mathrm{e}^{-\frac32}$ and $\rho \geq 1$. For every
$Q \in \set{P}_n(\rho)$, with an integer $n \geq 2$,
\begin{align}
D_{f_\alpha}(Q \| U_n) & \leq \Bigl[ \log(\alpha+1) + \tfrac32 \log \mathrm{e}
- \frac{\log \mathrm{e}}{\alpha+1} \Bigr] \, \chi^2(Q \| U_n) \nonumber \\
& \hspace*{0.4cm} + \frac{\log \mathrm{e}}{3(\alpha+1)}
\Bigl[ \exp\bigl(2 D_3(Q \| U_n)\bigr) - 1 \Bigr]  \label{06062019a1} \\
& \leq \Bigl[ \log(\alpha+1) + \tfrac32 \log \mathrm{e}
- \frac{\log \mathrm{e}}{\alpha+1} \Bigr] \, \frac{(\rho-1)^2}{4 \rho} \nonumber \\
& \hspace*{0.4cm} + \frac{\log \mathrm{e}}{3(\alpha+1)}
\Bigl[ \exp\bigl(2 D_3(Q \| U_n)\bigr) - 1 \Bigr]  \label{06062019a2}
\end{align}
where \eqref{06062019a1} follows from \eqref{Df_UB}, and \eqref{06062019a1}
holds due to \eqref{chi2 asymp.}. By upper bounding the second term in the
right side of \eqref{06062019a2}, for all $Q \in \set{P}_n(\rho)$,
\begin{align}
\label{06062019a3}
D_3(Q \| U_n) &= \tfrac12 \log \bigl(1 + 6 D_{\mathrm{A}}^{(3)}(Q \| U_n) \bigr) \\
\label{06062019a4}
&\leq \tfrac12 \log \bigl(1 + 6 \Delta(3,\rho) \bigr) \\
\label{06062019a5}
&= \tfrac12 \log \left( \frac{4 (\rho^3-1)^3}{27 (\rho-1)
(\rho - \rho^3)^2} \right) \\[0.1cm]
\label{06062019a6}
&= \tfrac12 \log \left( \frac{4 (\rho^2+\rho+1)^3}{27 \rho^2 (\rho+1)^2} \right)
\end{align}
where \eqref{06062019a3} holds by setting $\alpha=3$ in \eqref{1-1 Renyi and Alpha div.};
\eqref{06062019a4} follows from \eqref{UB - finite n}, \eqref{Alpha-divergence}
and \eqref{Delta def 0}; \eqref{06062019a5} holds by setting $\alpha=3$ in
\eqref{Delta 2: alpha-div.}; finally, \eqref{06062019a6} follows from the
factorizations $$ (\rho^3-1)^3 = (\rho-1)^3 (\rho^2 + \rho + 1)^3, \quad
(\rho-1)(\rho-\rho^3)^2 = (\rho-1)^3 \rho^2 (\rho+1)^2.$$
Substituting the bound in the right side of \eqref{06062019a6} into the
second term of the bound on the right side of \eqref{06062019a2} implies
that, for all $Q \in \set{P}_n(\rho)$,
\begin{align}
D_{f_\alpha}(Q \| U_n) & \leq \left[ \log(\alpha+1) + \tfrac32 \log \mathrm{e}
- \frac{\log \mathrm{e}}{\alpha+1} \right] \, \frac{(\rho-1)^2}{4 \rho} \nonumber \\
\label{06062019a7}
& \hspace*{0.4cm} + \frac{\log \mathrm{e}}{3(\alpha+1)}
\left[ \frac{4 (\rho^2+\rho+1)^3}{27 \rho^2 (\rho+1)^2} - 1 \right] \\
&= \left[ \log(\alpha+1) + \tfrac32 \log \mathrm{e}
- \frac{\log \mathrm{e}}{\alpha+1} \right] \frac{(\rho-1)^2}{4 \rho} \nonumber \\[0.1cm]
& \hspace*{0.4cm} + \frac{\log \mathrm{e}}{81(\alpha+1)}
\left( \frac{(\rho-1)(2\rho+1)(\rho+2)}{\rho(\rho+1)} \right)^2,
\end{align}
which therefore gives \eqref{Phi - UB1} by maximizing the left side of \eqref{06062019a7}
over $Q \in \set{P}_n(\rho)$, and letting $n$ tend to infinity (see \eqref{Phi def 1}).

We next derive the upper bound in \eqref{Phi - UB2}.
The second derivative of the convex function $f_\alpha \colon (0, \infty) \to \Reals$
in \eqref{f_alpha} is upper bounded over the interval $\bigl[\tfrac1\rho, \rho \bigr]$
by the positive constant $M = 2 \log(\alpha + \rho) + 3 \log \mathrm{e}$.
From \eqref{UB Df-M}, it follows that for all $Q \in \set{P}_n(\rho)$ (with
$\rho \geq 1$ and an integer $n \geq 2$) and $\alpha \geq \mathrm{e}^{-\frac32}$,
\begin{align}
\label{270519i}
D_{f_\alpha}(Q \| U_n) \leq
\Bigl[ \log(\alpha+\rho) + \tfrac32 \log \mathrm{e} \Bigr] \, \frac{(\rho-1)^2}{4 \rho},
\end{align}
which, from \eqref{Phi def 1}, yields \eqref{Phi - UB2}.

We finally derive the upper bound in \eqref{Phi - UB3} by loosening the bound in
\eqref{Phi - UB1}. The upper bound in the right side of \eqref{Phi - UB1} can be
rewritten as
\begin{align}
\Phi(\alpha, \rho) & \leq \Bigl[ \tfrac14 \, \log(\alpha+1) +
\tfrac38 \, \log \mathrm{e} \Bigr] \, \frac{(\rho-1)^2}{\rho} \nonumber \\
\label{10062019b8}
& \hspace*{0.4cm} + \frac{\log \mathrm{e}}{\alpha+1}
\left[ \frac1{81} \left(2 + \frac2{\rho} + \frac1{1+\rho} \right)^2
- \frac1{4\rho} \right] (\rho-1)^2.
\end{align}
For all $\rho \geq 1$,
\begin{align}
\label{10062019b9}
\frac1{81} \left(2 + \frac2{\rho} + \frac1{1+\rho} \right)^2 - \frac1{4\rho}
\leq \frac4{81},
\end{align}
which can be verified by showing that the left side of \eqref{10062019b9} is
monotonically increasing in $\rho$ over the interval $[1, \infty)$, and it
tends to $\tfrac4{81}$ as we let $\rho \to \infty$. Furthermore, for all
$\rho \geq 1$,
\begin{align}
\label{10062019b10}
\frac{(\rho-1)^2}{\rho} \leq \min\bigl\{\rho-1, (\rho-1)^2\bigr\}.
\end{align}
In view of inequalities \eqref{10062019b9} and \eqref{10062019b10}, one gets
\eqref{Phi - UB3} from \eqref{10062019b8} (where the latter is an equivalent
form of \eqref{Phi - UB1}).

\section{Proof of Theorem~\ref{thm: conjugate}}
\label{appendix: conjugate}

We start by proving Item~\ref{converse}).
In view of the variational representation of $f$-divergences (see
\cite[Theorem~2.1]{Keziou03}, and \cite[Lemma~1]{NWJ10}),
if $f \colon (0, \infty) \to \Reals$ is convex with $f(1)=0$, and
$P$ and $Q$ are probability measures defined on a set $\set{A}$, then
\begin{align}
\label{var. rep. f-div}
D_f(P\|Q) = \sup_{g \colon \set{A} \to \Reals} \Bigl( \expectation\bigl[g(X)\bigr]
- \expectation[ \, \overline{f}\bigl(g(Y)\bigr)\bigr] \Bigr),
\end{align}
where $X \sim P$ and $Y \sim Q$, and the supremum is taken over
all measurable functions $g$ under which the expectations are finite.

Let $P \in \set{P}_n(\rho)$, with $\rho > 1$, and let $Q := U_n$;
these probability mass functions are defined on the set
$\set{A}_n := \{1, \ldots, n\}$, and it follows that
\begin{align}
\label{08062019a1}
u_f(n, \rho) &\geq D_f(P \| U_n) \\
\label{08062019a2}
&\geq \expectation\bigl[g(X)\bigr]
- \frac1n \sum_{i=1}^n \overline{f}\bigl(g(i)\bigr),
\end{align}
where \eqref{08062019a1} holds by the definition in \eqref{def: u_f};
\eqref{08062019a2} holds due to \eqref{var. rep. f-div} with $X \sim P$,
and $Y$ being an equiprobable random variable over $\set{A}_n$. This
gives \eqref{eq: converse}.

We next prove Item~\ref{achievability}). As above, let
$f \colon (0, \infty) \to \Reals$ be a convex function with
$f(1)=0$. Let $\beta^\ast \in \Gamma_n(\rho)$ be a maximizer of
the right side of \eqref{opt1 Df: beta}. Then,
\begin{align}
\label{08062019a3}
u_f(n, \rho) &= D_f(Q_{\beta^\ast} \| U_n) \\
\label{08062019a4}
&=  \frac1n \sum_{i=1}^n f \bigl(n Q_{\beta^\ast}(i) \bigr).
\end{align}
Let $\varepsilon > 0$ be selected arbitrarily. We have
$\overline{(\overline{f})} \equiv f$  (i.e., repeating twice the
convex conjugate operation (see \eqref{conjugate}) on a convex
function $f$, returns $f$ itself). From the convexity of $f$, it
therefore follows that, for all $t > 0$, there exists $x \in \Reals$
such that
\begin{align}
\label{08062019a5}
f(t) \leq tx - \overline{f}(x) + \varepsilon.
\end{align}
Let
\begin{align}
\label{08062019a6}
t_i := n Q_{\beta^\ast}(i), \quad \forall \, i \in \set{A}_n,
\end{align}
let $x := x_i(\varepsilon) \in \Reals$ be selected to satisfy
\eqref{08062019a5} with $t := t_i$, and let the function
$g_\varepsilon \colon \set{A}_n \to \Reals$ be defined as
\begin{align}
\label{08062019a7}
g_\varepsilon(i) = x_i(\varepsilon), \quad \forall \, i \in \set{A}_n.
\end{align}
Consequently, it follows from \eqref{08062019a5}--\eqref{08062019a7}
that for all such $i$
\begin{align}
\label{08062019a8}
f\bigl(n Q_{\beta^\ast}(i)\bigr) \leq n Q_{\beta^\ast}(i) \, g_\varepsilon(i)
- \overline{f}\bigl(g_\varepsilon(i)\bigr) + \varepsilon.
\end{align}
Let $P := Q_{\beta^\ast} \in \set{P}_n(\rho)$ (see \eqref{Q_beta}), and $X \sim P$.
Then,
\begin{align}
\label{08062019a9}
u_f(n,\rho) &= \frac1n \sum_{i=1}^n f \bigl(n Q_{\beta^\ast}(i) \bigr) \\
\label{08062019a10}
&\leq \sum_{i=1}^n Q_{\beta^\ast}(i) \, g_\varepsilon(i) -
\frac1n \sum_{i=1}^n \overline{f}\bigl(g_\varepsilon(i)\bigr) + \varepsilon \\
\label{08062019a11}
&= \expectation\bigl[ g_\varepsilon(X) \bigr] -
\frac1n \sum_{i=1}^n \overline{f}\bigl(g_\varepsilon(i)\bigr) + \varepsilon
\end{align}
where \eqref{08062019a9} holds due to \eqref{08062019a3} and \eqref{08062019a4};
\eqref{08062019a10} follows from \eqref{08062019a8}; \eqref{08062019a11} holds
since by assumption $P_X = Q_{\beta^\ast}$. This gives \eqref{eq: achievability}.

\section{Proof of Theorem~\ref{theorem: generalized Fano Df}}
\label{appendix: generalized Fano Df}
For $y \in \set{Y}$, let the $L$-size list of the decoder be given by
$\set{L}(y) = \{x_1(y), \ldots, x_L(y) \}$ with $L < M$. Then,
the (average) list decoding error probability is given by
\begin{align}
\label{P err}
P_{\set{L}} = \expectation\bigl[P_{\set{L}}(Y)\bigr]
\end{align}
where the conditional list decoding error probability, given that
$Y=y \in \set{Y}$, is equal to
\begin{align}
\label{cond. P err}
P_{\set{L}}(y) = 1 - \sum_{\ell=1}^L P_{X|Y}\bigl( x_\ell(y) \, | \, y \bigr).
\end{align}
For every $y \in \set{Y}$,
\begin{align}
& D_f\bigl( P_{X|Y}(\cdot | y) \, \| \, U_M \bigr) \nonumber \\
\label{21062019a1}
& \geq D_f \left( \left[ \, \sum_{\ell=1}^L P_{X|Y}\bigl(x_\ell(y) \, | \, y \bigr),
\; 1 - \sum_{\ell=1}^L P_{X|Y}\bigl(x_\ell(y) \, | \, y) \right] \, \|
\, \left[\frac{L}{M}, \, 1 - \frac{L}{M}\right] \right) \\[0.1cm]
\label{21062019a2}
& = D_f \left( \bigl[1-P_{\set{L}}(y), \, P_{\set{L}}(y) \bigr] \, \|
\, \left[\frac{L}{M}, \, 1 - \frac{L}{M}\right] \right),
\end{align}
where \eqref{21062019a1} holds by the data-processing inequality
for $f$-divergences, and since for every $y \in \set{Y}$
\begin{align}
\label{30062019a12}
\sum_{\ell=1}^L U_M\bigl(x_\ell(y)\bigr)
= \sum_{\ell=1}^L \frac1M = \frac{L}{M};
\end{align}
\eqref{21062019a2} is due to \eqref{cond. P err}.
Hence, it follows that
\begin{align}
\label{21062019a2.1}
& \expectation \Bigl[ D_f\bigl( P_{X|Y}(\cdot | Y) \, \|
\, U_M \bigr) \Bigr] \nonumber \\[0.1cm]
& \geq \expectation \biggl[ D_f \biggl( \bigl[1-P_{\set{L}}(Y),
\, P_{\set{L}}(Y) \bigr] \, \| \, \left[\frac{L}{M}, \,
1 - \frac{L}{M}\right] \biggr) \biggr] \\[0.1cm]
\label{21062019a3}
& = \frac{L}{M} \; \expectation \left[ f\biggl( \frac{M
(1-P_{\set{L}}(Y))}{L} \biggr) \right]
+ \biggl(1 - \frac{L}{M}\biggr) \; \expectation \biggl[
f\biggl( \frac{M P_{\set{L}}(Y)}{M-L} \biggr) \biggr] \\[0.1cm]
\label{21062019a4}
& \geq \frac{L}{M} \; f\biggl( \frac{M \, \expectation
[ 1-P_{\set{L}}(Y)]}{L} \biggr)
+ \biggl(1 - \frac{L}{M}\biggr) \, f\biggl(\frac{M \,
\expectation[P_{\set{L}}(Y)]}{M-L} \biggr) \\[0.1cm]
\label{21062019a5}
&= \frac{L}{M} \; f\biggl( \frac{M \bigl(1-P_{\set{L}}\bigr)}{L} \biggr) +
\biggl(1 - \frac{L}{M}\biggr) \, f\biggl(\frac{M P_{\set{L}}}{M-L} \biggr),
\end{align}
where \eqref{21062019a2.1} holds by taking expectations in
\eqref{21062019a1}--\eqref{21062019a2} with respect to $Y$;
\eqref{21062019a3} holds by the definition of $f$-divergence,
and the linearity of expectation operator; \eqref{21062019a4}
follows from the convexity of $f$ and Jensen's inequality;
finally, \eqref{21062019a5} holds by \eqref{P err}.

\section{Proof of Corollary~\ref{corollary: generalized Fano-Renyi inequality}}
\label{appendix: generalized Fano-Renyi inequality}

Let $\alpha \in (0,1) \cup (1, \infty)$, and let $y \in \set{Y}$.
The proof starts by applying Theorem~\ref{theorem: generalized Fano Df}
in the setting where $Y=y$ is deterministic, and the convex function
$f \colon (0, \infty) \to \Reals$ is given by $f := u_\alpha$
in \eqref{f of Alpha-divergence}, i.e.,
\begin{align}
\label{30062019a1}
f(t) = \frac{t^\alpha-\alpha(t-1)-1}{\alpha(\alpha-1)}, \quad t \geq 0.
\end{align}
In this setting, \eqref{generalized Fano Df} is specialized to
\begin{align}
\label{30062019a2}
D_f \bigl(P_{X|Y}(\cdot|y) \, \| \, U_M \bigr)
\geq \frac{L}{M} \; f\biggl(\frac{M \, (1-P_{\set{L}}(y))}{L} \biggr)
+ \biggl(1-\frac{L}{M}\biggr) \; f\biggl(\frac{M P_{\set{L}}(y)}{M-L} \biggr),
\end{align}
where $P_{\set{L}}(y)$ is the conditional list decoding error probability given
that $Y=y$. Substituting \eqref{30062019a1} into the right side of \eqref{30062019a2}
gives
\begin{align}
& \frac{L}{M} \; f\biggl(\frac{M \, (1-P_{\set{L}}(y))}{L} \biggr)
+ \biggl(1-\frac{L}{M}\biggr) \; f\biggl(\frac{M P_{\set{L}}(y)}{M-L} \biggr) \nonumber \\[0.1cm]
&= \frac1{\alpha(\alpha-1)} \left[ P_{\set{L}}^\alpha(y) \, \biggl(1-\frac{L}{M}\biggr)^{1-\alpha}
+ \bigl(1-P_{\set{L}}(y)\bigr)^\alpha \, \biggl(\frac{L}{M}\biggr)^{1-\alpha} - 1 \right] \\[0.1cm]
\label{30062019a3}
&= \frac1{\alpha(\alpha-1)} \left[ \exp\Biggl( (\alpha-1) \, d_\alpha\biggl(P_{\set{L}}(y) \, \| \,
1-\frac{L}{M} \biggr) \Biggr) - 1 \right],
\end{align}
where \eqref{30062019a3} follows from \eqref{eq1: binary RD}.
Substituting \eqref{30062019a1} into the left side of \eqref{30062019a2} gives
\begin{align}
& D_f \bigl(P_{X|Y}(\cdot|y) \, \| \, U_M \bigr) \nonumber \\
&= \frac1{M\alpha(\alpha-1)} \sum_{x \in \set{X}} \Bigl[ \bigl(M P_{X|Y}(x|y)\bigr)^\alpha
- \alpha \bigl( M P_{X|Y}(x|y)-1 \bigr) - 1 \Bigr] \\
&= \frac1{M\alpha(\alpha-1)} \Biggl[ M^\alpha \sum_{x \in \set{X}} P_{X|Y}^\alpha(x|y)
- \alpha \underbrace{\sum_{x \in \set{X}} \bigl( M P_{X|Y}(x|y)-1 \bigr)}_{=0
\; \; (|\set{X}|=M)} -M \Biggr] \\
&= \frac1{\alpha(\alpha-1)} \Biggl[ M^{\alpha-1}
\sum_{x \in \set{X}} P_{X|Y}^\alpha(x|y) - 1 \Biggr] \\
\label{30062019a4}
&= \frac1{\alpha(\alpha-1)} \biggl[ \exp \Bigl((\alpha-1) \,
\bigl[\log M - H_\alpha(X | Y=y) \bigr] \Bigr) - 1 \biggr].
\end{align}
Substituting \eqref{30062019a3} and \eqref{30062019a4} into the right and left sides
of \eqref{30062019a2}, and rearranging terms while relying on the monotonicity property
of an exponential function gives
\begin{align}
\label{30062019a5}
H_\alpha(X|Y=y) \leq
\log M - d_\alpha\biggl(P_{\set{L}}(y) \, \| \, 1-\frac{L}{M} \biggr).
\end{align}

We next obtain an upper bound on the Arimoto-R\'{e}nyi conditional entropy.
\begin{align}
& H_\alpha(X|Y) \nonumber \\
\label{30062019a6}
&= \frac{\alpha}{1-\alpha} \, \log
\, \int_{\set{Y}} \mathrm{d}P_Y(y) \, \exp \left(
\frac{1-\alpha}{\alpha} \;
H_{\alpha}(X | Y=y) \right) \\
\label{30062019a7}
&\leq \frac{\alpha}{1-\alpha} \, \log
\, \int_{\set{Y}} \mathrm{d}P_Y(y) \, \exp \Biggl(
\frac{1-\alpha}{\alpha} \; \biggl[ \log M - d_\alpha
\biggl(P_{\set{L}}(y) \, \| \, 1-\frac{L}{M} \biggr) \biggr] \Biggr) \\
\label{30062019a8}
&= \log M + \frac{\alpha}{1-\alpha} \, \log
\int_{\set{Y}} \mathrm{d}P_Y(y)
\left[ P_{\set{L}}^{\alpha}(y) \biggl(1-\frac{L}{M}\biggr)^{1-\alpha}
+ \bigl(1-P_{\set{L}}(y)\bigr)^\alpha \biggl(\frac{L}{M}\biggr)^{1-\alpha}
\right]^{\frac1\alpha}
\end{align}
where \eqref{30062019a6} holds due to \eqref{eq2: Arimoto - cond. RE};
\eqref{30062019a7} follows from \eqref{30062019a5}, and \eqref{30062019a8}
follows from \eqref{eq1: binary RD}. By \cite[Lemma~1]{ISSV18}, it
follows that the integrand in the right side of \eqref{30062019a8}
is convex in $P_{\set{L}}(y)$ if $\alpha > 1$; furthermore, it is
concave in $P_{\set{L}}(y)$ if $\alpha \in (0,1)$. Invoking Jensen's
inequality therefore yields (see \eqref{P err})
\begin{align}
\label{30062019a9}
H_\alpha(X|Y) & \leq \log M + \frac{\alpha}{1-\alpha} \, \log \left(
\left[ P_{\set{L}}^{\alpha} \biggl(1-\frac{L}{M}\biggr)^{1-\alpha}
+ \bigl(1-P_{\set{L}}\bigr)^\alpha \biggl(\frac{L}{M}\biggr)^{1-\alpha}
\right]^{\frac1\alpha} \right) \\
\label{30062019a10}
&= \log M - \frac{1}{\alpha-1} \, \log
\left( P_{\set{L}}^{\alpha} \biggl(1-\frac{L}{M}\biggr)^{1-\alpha}
+ \bigl(1-P_{\set{L}}\bigr)^\alpha \biggl(\frac{L}{M}\biggr)^{1-\alpha}
\right) \\
\label{30062019a11}
&= \log M - d_\alpha\biggl(P_{\set{L}} \, \| \, 1-\frac{L}{M} \biggr),
\end{align}
where \eqref{30062019a9} follows from Jensen's inequality, and
\eqref{30062019a11} follows from \eqref{eq1: binary RD}. This
proves \eqref{eq: generalized Fano-Renyi - list decoding} and
\eqref{eq2: generalized Fano-Renyi - list decoding} for all
$\alpha \in (0,1) \cup (1,\infty)$. The necessary and sufficient
condition for \eqref{eq: generalized Fano-Renyi - list decoding}
to hold with equality, as given in \eqref{eq: tight Fano-Renyi - list decoding},
follows from the proof of \eqref{30062019a2} (see
\eqref{21062019a1}--\eqref{30062019a12}), and from the use of
Jensen's inequality in \eqref{30062019a9}.

\section{Proof of Theorem~\ref{theorem: refined Fano's inequality}}
\label{appendix: refined Fano's inequality}

The proof of Theorem~\ref{theorem: refined Fano's inequality} relies on
Theorem~\ref{thm: SDPI-IS}, and the proof of Theorem~\ref{theorem: generalized Fano Df}.

Let $\set{Z} = \{0,1\}$ and, without any loss of generality, let
$\set{X} = \{1, \ldots, M \}$. For every $y \in \set{Y}$, define a
deterministic transformation from $\set{X}$ to $\set{Z}$ such that
every $x \in \set{L}(y)$ is mapped to $z=0$, and every
$x \notin \set{L}(y)$ is mapped to $z=1$. This corresponds to a conditional
probability mass function, for every $y \in \set{Y}$, where
$W_{Z|X}^{(y)}(z|x) = 1$ if $x \in \set{L}(y)$ and $z=0$, or if
$x \notin \set{L}(y)$ and $z=1$; otherwise, $W_{Z|X}^{(y)}(z|x) = 0$.
Let $\set{L}(y) := \{x_1(y), \ldots, x_L(y) \}$ with $L < M$. Then, for every
$y \in \set{Y}$, a conditional probability mass function $P_{X|Y}(\cdot|y)$
implies that
\begin{align}
\label{P_Z}
P_Z^{(y)}(z) := \sum_{x \in \set{X}} P_{X|Y}(x | y) \, W_{Z|X}^{(y)}(z|x), \quad \forall \, z \in \{0,1\},
\end{align}
satisfies (see \eqref{cond. P err})
\begin{align}
\label{P_Z 2}
& P_Z^{(y)}(0) = \sum_{\ell = 1}^L P_{X|Y}(x_\ell(y) | y) = 1 - P_{\set{L}}(y), \\
& P_Z^{(y)}(1) = P_{\set{L}}(y).
\end{align}
Under the deterministic transformation $W_{Z|X}^{(y)}$ as above, the equiprobable distribution $Q_X^{(y)} = U_M$
(independently of $y \in \set{Y}$) is mapped to a Bernoulli distribution over the two-elements
set $\set{Z}$ where
\begin{align}
\label{Q_Z}
Q_Z^{(y)} = \left[\frac{L}{M}, ~1 - \frac{L}{M}\right], \quad \forall \, y \in \set{Y}.
\end{align}
Given $Y=y \in \set{Y}$, applying Theorem~\ref{thm: SDPI-IS} with the
transformation $W_{Z|X}^{(y)}$ as above gives that
\begin{align}
& D_f\bigl( P_{X|Y}(\cdot | y) \, \| \, U_M \bigr) \nonumber \\
\label{27062019a1}
& \geq D_f\bigl( P_Z^{(y)} \| Q_Z^{(y)} \bigr)
+ c_f\bigl(\xi_1(y), \xi_2(y)\bigr)
\left[ \chi^2\bigl( P_{X|Y}(\cdot | y) \, \| \, U_M \bigr)
- \chi^2\bigl( P_Z^{(y)} \| Q_Z^{(y)} \bigr) \right]
\end{align}
where, from \eqref{xi1} and \eqref{xi2},
\begin{align}
\label{27062019a2}
& \xi_1(y) = \min_{x \in \set{X}} \frac{P_{X|Y}(x|y)}{U_M(x)}
= M \min_{x \in \set{X}} P_{X|Y}(x|y), \\
\label{27062019a3}
& \xi_2(y) = \max_{x \in \set{X}} \frac{P_{X|Y}(x|y)}{U_M(x)}
= M \max_{x \in \set{X}} P_{X|Y}(x|y).
\end{align}
Since, from \eqref{28062019a1}, \eqref{28062019a2}, \eqref{27062019a2}
and \eqref{27062019a3},
\begin{align}
\label{27062019a4}
\inf_{y \in \set{Y}} \xi_1(y)
= M \inf_{(x,y) \in \set{X} \times \set{Y}} P_{X|Y}(x|y) = \xi_1^\ast, \\
\label{27062019a5}
\sup_{y \in \set{Y}} \xi_2(y)
= M \sup_{(x,y) \in \set{X} \times \set{Y}} P_{X|Y}(x|y) = \xi_2^\ast,
\end{align}
it follows from the definition of $c_f(\cdot, \cdot)$ in \eqref{c_f} that
for every $y \in \set{Y}$
\begin{align}
\label{27062019a6}
c_f\bigl(\xi_1(y), \xi_2(y)\bigr) & \geq c_f\bigl(\xi_1^\ast, \xi_2^\ast\bigr) \\
\label{27062019a7}
& = \tfrac12 \inf_{t \in \set{I}(\xi_1^\ast, \xi_2^\ast)} f''(t) \\
\label{27062019a8}
& \geq \tfrac12 \, m_f
\end{align}
where the last inequality holds by the assumption in \eqref{m_f}.
Combining \eqref{27062019a1} and \eqref{27062019a6}--\eqref{27062019a8}
implies that, for every $y \in \set{Y}$,
\begin{align}
\nonumber
& D_f\bigl( P_{X|Y}(\cdot | y) \, \| \, U_M \bigr) \\
\label{27062019a9}
& \geq D_f\bigl( P_Z^{(y)} \| Q_Z^{(y)} \bigr)
+ \tfrac12 m_f
\Bigl[ \chi^2\bigl( P_{X|Y}(\cdot | y) \, \| \, U_M \bigr)
- \chi^2\bigl( P_Z^{(y)} \| Q_Z^{(y)} \bigr) \Bigr].
\end{align}
Hence,
\begin{align}
& \expectation \bigl[ D_f\bigl( P_{X|Y}(\cdot | Y) \, \| \, U_M \bigr) \bigr] \nonumber \\
\label{27062019a10}
& \geq \expectation \bigl[ D_f\bigl( P_Z^{(Y)} \| Q_Z^{(Y)} \bigr) \bigr]
+ \tfrac12 \, m_f \,
\expectation \Bigl[ \chi^2\bigl( P_{X|Y}(\cdot | Y) \, \| \, U_M \bigr)
- \chi^2\bigl( P_Z^{(Y)} \| Q_Z^{(Y)} \bigr) \Bigr]
\end{align}
where \eqref{27062019a10} holds by taking expectations with respect to
$Y$ on both sides of \eqref{27062019a9}.

Referring to the first term in the right side of \eqref{27062019a10} gives
\begin{align}
\expectation \bigl[ D_f\bigl( P_Z^{(Y)} \| Q_Z^{(Y)} \bigr) \bigr]
\label{27062019a11}
& = \expectation \biggl[ D_f \left( \bigl[1-P_{\set{L}}(Y),
\, P_{\set{L}}(Y) \bigr] \, \| \, \left[\frac{L}{M}, \,
1 - \frac{L}{M}\right] \right) \biggr] \\[0.1cm]
\label{27062019a12}
&\geq \frac{L}{M} \; f\biggl( \frac{M \bigl(1-P_{\set{L}}\bigr)}{L} \biggr) +
\biggl(1 - \frac{L}{M}\biggr) \; f\biggl(\frac{M P_{\set{L}}}{M-L} \biggr),
\end{align}
where \eqref{27062019a11} follows from \eqref{P_Z 2}--\eqref{Q_Z}, and
\eqref{27062019a12} holds due to \eqref{21062019a3}--\eqref{21062019a5}.

Referring to the second term in the right side of \eqref{27062019a10} gives
\begin{align}
& \expectation \Bigl[ \chi^2\bigl( P_{X|Y}(\cdot | Y) \, \| \, U_M \bigr)
- \chi^2\bigl( P_Z^{(Y)} \| Q_Z^{(Y)} \bigr) \Bigr] \nonumber \\[0.1cm]
\label{27062019a13}
&= \expectation \biggl[ \chi^2\bigl( P_{X|Y}(\cdot | Y) \, \| \, U_M \bigr)
- \chi^2\left( \bigl[1-P_{\set{L}}(Y),
\, P_{\set{L}}(Y) \bigr] \, \| \, \left[\frac{L}{M}, \,
1 - \frac{L}{M}\right] \right) \biggr] \\[0.1cm]
\label{27062019a14}
&= \expectation \Biggl[ M \sum_{x \in \set{X}} P_{X|Y}^2(x|Y) -
\frac{M \bigl(1-P_{\set{L}}(Y)\bigr)^2}{L}
- \frac{M P_{\set{L}}^2(Y)}{M-L} \Biggr] \\[0.1cm]
\label{27062019a15}
&= M \, \expectation \Biggl[ \, \sum_{x \in \set{X}} P_{X|Y}^2(x|Y) \Biggr]
- \frac{M}{L} + \frac{2M}{L} \cdot \expectation\bigl[P_{\set{L}}(Y)\bigr] \nonumber \\[0.1cm]
& \hspace*{0.4cm} - \left( \frac{M}{L} + \frac{M}{M-L} \right)
\expectation\bigl[P_{\set{L}}^2(Y)\bigr] \\[0.1cm]
\label{27062019a16}
&= M \, \expectation \Biggl[ \, \sum_{x \in \set{X}} P_{X|Y}^2(x|Y) \Biggr]
- \frac{M \bigl(1 - 2 P_{\set{L}}\bigr)}{L}
- \frac{M^2 \, \expectation \bigl[P_{\set{L}}^2(Y)\bigr]}{L(M-L)},
\end{align}
where \eqref{27062019a13} follows from \eqref{P_Z}--\eqref{Q_Z};
\eqref{27062019a14} follows from \eqref{27062019a17}--\eqref{diff of chi^2 div.};
\eqref{27062019a16} is due to \eqref{P err}. Furthermore, we get (since
$P_{\set{L}}(Y) \in [0,1]$)
\begin{align}
\label{27062019a19}
& \expectation \bigl[P_{\set{L}}^2(Y)\bigr]
\leq \expectation\bigl[P_{\set{L}}(Y)\bigr] = P_{\set{L}}, \\[0.1cm]
\label{27062019a20}
& \expectation \bigl[P_{\set{L}}^2(Y)\bigr]
\geq \expectation^2\bigl[P_{\set{L}}(Y)\bigr] = P_{\set{L}}^2,
\end{align}
and
\begin{align}
& \expectation \left[ \, \sum_{x \in \set{X}} P_{X|Y}^2(x|Y) \right] \nonumber \\[0.1cm]
\label{27062019a21}
& = \int_{\set{Y}} \mathrm{d}P_Y(y) \, \sum_{x \in \set{X}} P_{X|Y}^2(x|y) \\
\label{27062019a22}
& = \int_{\set{X} \times \set{Y}} \mathrm{d}P_{XY}(x,y) \, P(x|y) \\[0.1cm]
\label{27062019a23}
& = \expectation \bigl[ P_{X|Y}(X|Y) \bigr].
\end{align}
Combining \eqref{27062019a13}--\eqref{27062019a23} gives
\begin{align}
& M \left( \expectation \bigl[ P_{X|Y}(X|Y) \bigr]
- \frac{1 - P_{\set{L}}}{L}
- \frac{P_{\set{L}}}{M-L} \right)^{+} \nonumber \\[0.1cm]
\label{27062019a24}
& \leq \expectation \left[ \chi^2\bigl( P_{X|Y}(\cdot | Y) \, \| \, U_M \bigr)
- \chi^2\bigl( P_Z^{(Y)} \| Q_Z^{(Y)} \bigr) \right] \\[0.1cm]
\label{27062019a25}
& \leq M \left( \expectation \bigl[ P_{X|Y}(X|Y) \bigr]
- \frac{\bigl(1 - P_{\set{L}}\bigr)^2}{L}
- \frac{P_{\set{L}}^2}{M-L} \right),
\end{align}
providing tight upper and lower bounds on
$\expectation \left[ \chi^2\bigl( P_{X|Y}(\cdot | Y) \, \| \, U_M \bigr)
- \chi^2\bigl( P_Z^{(Y)} \| Q_Z^{(Y)} \bigr) \right]$ if $P_{\set{L}}$
is small. Note that the lower bound on the left side of
\eqref{27062019a24} is non-negative since, by the data-processing
inequality for the $\chi^2$ divergence, the right side of
\eqref{27062019a24} should be non-negative (see \eqref{P_Z}--\eqref{Q_Z}).
Finally, combining \eqref{27062019a10}--\eqref{27062019a25}
yields \eqref{list dec.-26062019a}, which proves
Item~\ref{Part a - refined Fano's inequality}).

For proving Item~\ref{Part b - refined Fano's inequality}), the upper bound on
the left side of \eqref{27062019a19} is tightened. If the list decoder selects
the $L$ most probable elements from $\set{X}$ given the value of $Y \in \set{Y}$,
then $P_{\set{L}}(y) \leq 1 - \frac{L}{M}$ for every $y \in \set{Y}$. Hence,
the bound in \eqref{27062019a19} is replaced by the tighter bound
\begin{align}
\label{27062019a26}
& \expectation \bigl[P_{\set{L}}^2(Y)\bigr]
\leq \left(1-\frac{L}{M}\right) P_{\set{L}}.
\end{align}
Combining \eqref{27062019a13}--\eqref{27062019a16},
\eqref{27062019a21}--\eqref{27062019a23} and \eqref{27062019a26} gives
the following improved lower bound in the left side of \eqref{27062019a24}:
\begin{align}
& M \left( \expectation \bigl[ P_{X|Y}(X|Y) \bigr]
- \frac{1 - P_{\set{L}}}{L} \right)^{+} \nonumber \\[0.1cm]
\label{27062019a27}
& \leq \expectation \left[ \chi^2\bigl( P_{X|Y}(\cdot | Y) \, \| \, U_M \bigr)
- \chi^2\bigl( P_Z^{(Y)} \| Q_Z^{(Y)} \bigr) \right].
\end{align}
It is next shown that the operation $(\cdot)^+$ in the left side of
\eqref{27062019a27} is redundant. From \eqref{P err} and \eqref{cond. P err},
\begin{align}
\label{09072019a1}
P_{\set{L}} &= 1 - \sum_{\ell=1}^L \expectation \bigl[ P_{X|Y}\bigl( x_{\ell}(Y)
\, | \, Y \bigr) \bigr] \\[0.1cm]
\label{09072019a2}
&= 1 - \sum_{\ell=1}^L \int_{\set{Y}} \mathrm{d}P_Y(y) \, P_{X|Y}\bigl( x_{\ell}(y)
\, | \, y \bigr) \\[0.1cm]
\label{09072019a3}
&= 1 - \int_{\set{Y}} \mathrm{d}P_Y(y) \, \sum_{\ell=1}^L P_{X|Y}\bigl( x_{\ell}(y)
\, | \, y \bigr) \\[0.1cm]
\label{09072019a4}
&\geq 1 - L \int_{\set{Y}} \mathrm{d}P_Y(y) \, \sum_{\ell=1}^L P_{X|Y}^2\bigl( x_{\ell}(y)
\, | \, y \bigr) \\[0.1cm]
\label{09072019a5}
&\geq 1 - L \int_{\set{Y}} \mathrm{d}P_Y(y) \, \sum_{x \in \set{X}}
P_{X|Y}^2\bigl( x | y \bigr) \\[0.1cm]
\label{09072019a6}
&\geq 1 - L \int_{\set{X} \times \set{Y}} \mathrm{d}P_{XY}(x,y) \; P_{X|Y}(x|y) \\[0.2cm]
\label{09072019a7}
&= 1 - L \, \expectation \bigl[ P_{X|Y}(X|Y) \bigr],
\end{align}
where \eqref{09072019a4} is due to the Cauchy-Schwarz inequality, and
\eqref{09072019a5} is due to the inclusion $\set{L}(y) \subseteq \set{X}$ for
all $y \in \set{Y}$. From \eqref{09072019a1}--\eqref{09072019a7},
$\expectation \bigl[ P_{X|Y}(X|Y) \bigr] \geq \frac{1-P_{\set{L}}}{L}$, which
implies that the operation $(\cdot)^+$ in the left side of \eqref{27062019a27}
is indeed redundant.

Similarly to the proof of \eqref{list dec.-26062019a} (see
\eqref{27062019a10}--\eqref{27062019a12}), \eqref{27062019a27} yields
\eqref{list dec.-26062019b} while ignoring the operation $(\cdot)^+$ in the left side
of \eqref{27062019a27}.

\section{Proof of Theorem~\ref{theorem: LB - variable list size}}
\label{appendix: LB - variable list size}

For every $y \in \set{Y}$, let the $M$ elements of $\set{X}$ be sorted in decreasing
order according to the conditional probabilities $P_{X|Y}(\cdot | y)$. Let
$x_{\ell}(y)$ be the $\ell$-th most probable element in $\set{X}$ given $Y=y$, i.e.,
\begin{align}
\label{02072019a1}
P_{X|Y}(x_1(y) \, | y) \geq P_{X|Y}(x_2(y) \, | y) \geq \ldots \geq P_{X|Y}(x_M(y) \, | y).
\end{align}
The conditional list decoding error probability, given $Y=y$, satisfies
\begin{align}
\label{02072019a2}
P_{\set{L}}(y) & \geq 1 - \sum_{\ell = 1}^{|\set{L}(y)|} P_{X|Y}(x_\ell(y) \, | y) \\
\label{02072019a3}
&:= P_{\set{L}}^{\mathrm{(opt)}}(y),
\end{align}
and the (average) list decoding error probability satisfies
$P_{\set{L}} \geq P_{\set{L}}^{\mathrm{(opt)}}$. Let $U_M$ denote the equiprobable
distribution on $\set{X}$, and let $g_\gamma \colon [0, \infty) \to \Reals$
be given by $g_\gamma(t) := (t-\gamma)^{+}$ with $\gamma \geq 1$, where $u^{+} := \max\{u,0\}$
for $u \in \Reals$. The function $g_\gamma(\cdot)$ is convex, and $g_\gamma(1)=0$ for
$\gamma \geq 1$; the $f$-divergence $D_{g_\gamma}(\cdot \| \cdot)$ is named as the
$E_\gamma$ divergence (see, e.g., \cite{LCV17}), i.e.,
\begin{align}
\label{E_gamma divergence}
E_\gamma(P\|Q) := D_{g_\gamma}(P\|Q), \quad \forall \, \gamma \geq 1,
\end{align}
for all probability measures $P$ and $Q$. For every $y \in \set{Y}$,
\begin{align}
& E_\gamma\bigl( P_{X|Y}(\cdot | y) \, \| \, U_M \bigr) \nonumber \\[0.1cm]
\label{02072019a4}
& \geq E_\gamma \biggl( [1-P_{\set{L}}^{\mathrm{(opt)}}(y),~P_{\set{L}}^{\mathrm{(opt)}}(y)] \,
\| \, \biggl[ \frac{|\set{L}(y)|}{M},~1-\frac{|\set{L}(y)|}{M} \biggr] \biggr) \\[0.1cm]
\label{02072019a5}
&= \frac{|\set{L}(y)|}{M} \cdot g_\gamma\biggl(\frac{M\bigl(1-P_{\set{L}}^{\mathrm{(opt)}}(y)
\bigr)}{|\set{L}(y)|} \, \biggr) +
\biggl(1-\frac{|\set{L}(y)|}{M}\biggr) \; g_\gamma\biggl(\frac{M \,
P_{\set{L}}^{\mathrm{(opt)}}(y)}{M-|\set{L}(y)|} \biggr),
\end{align}
where \eqref{02072019a4} holds due to the data-processing inequality for $f$-divergences, and
because of \eqref{02072019a3}; \eqref{02072019a5} holds due to \eqref{E_gamma divergence}. Furthermore,
in view of \eqref{02072019a1} and \eqref{02072019a3}, it follows that
$\frac{M \, P_{\set{L}}^{\mathrm{(opt)}}(y)}{M-|\set{L}(y)|} \leq 1$ for all $y \in \set{Y}$;
by the definition of $g_\gamma$, it follows that
\begin{align}
\label{02072019a6}
g_\gamma\biggl(\frac{M \, P_{\set{L}}^{\mathrm{(opt)}}(y)}{M-|\set{L}(y)|} \biggr) = 0,
\quad \forall \, \gamma \geq 1.
\end{align}
Substituting \eqref{02072019a6} into the right side of \eqref{02072019a5} gives that,
for all $y \in \set{Y}$,
\begin{align}
& E_\gamma\bigl( P_{X|Y}(\cdot | y) \, \| \, U_M \bigr) \nonumber \\
\label{02072019a7}
&\geq \frac{|\set{L}(y)|}{M} \cdot g_\gamma\biggl(\frac{M\bigl(1-P_{\set{L}}^{\mathrm{(opt)}}(y)
\bigr)}{|\set{L}(y)|} \, \biggr) \\
\label{02072019a8}
&= \biggl(1-P_{\set{L}}^{\mathrm{(opt)}}(y)-\frac{\gamma \, |\set{L}(y)|}{M}\biggr)^{+}.
\end{align}
Taking expectations with respect to $Y$ in \eqref{02072019a7}--\eqref{02072019a8}, and applying
Jensen's inequality to the convex function $f(u) := (u)^{+}$, for $u \in \Reals$, gives
\begin{align}
& \expectation\Bigl[ E_\gamma\bigl( P_{X|Y}(\cdot | Y) \, \| \, U_M \bigr) \Bigr] \nonumber \\
\label{02072019a9}
& \geq \expectation\biggl[ \biggl(1-P_{\set{L}}^{\mathrm{(opt)}}(Y)-\frac{\gamma \,
|\set{L}(Y)|}{M}\biggr)^{+} \biggr] \\
\label{02072019a10}
& \geq \Biggl(1- \expectation\bigl[P_{\set{L}}^{\mathrm{(opt)}}(Y)\bigr]
-\frac{\gamma \, \expectation\bigl[|\set{L}(Y)|\bigr]}{M}\Biggr)^{+} \\
\label{02072019a11}
& = \Biggl(1- P_{\set{L}}^{\mathrm{(opt)}}
-\frac{\gamma \, \expectation\bigl[|\set{L}(Y)|\bigr]}{M}\Biggr)^{+} \\
\label{02072019a12}
& \geq 1- P_{\set{L}}^{\mathrm{(opt)}}
-\frac{\gamma \, \expectation\bigl[|\set{L}(Y)|\bigr]}{M}.
\end{align}
On the other hand, the left side of \eqref{02072019a9} is equal to
\begin{align}
& \expectation\Bigl[ E_\gamma\bigl( P_{X|Y}(\cdot | Y) \, \| \, U_M \bigr) \Bigr] \nonumber \\
\label{02072019a13}
&= \expectation\Biggl[ \frac1M \sum_{x \in \set{X}} \bigl( M P_{X|Y}(x|Y) - \gamma \bigr)^{+} \Biggr] \\
\label{02072019a14}
&= \expectation\Biggl[ \, \sum_{x \in \set{X}} \biggl( P_{X|Y}(x|Y) - \frac{\gamma}{M} \biggr)^{+} \Biggr] \\
\label{02072019a15}
&= \tfrac12 \, \expectation\Biggl[ \, \sum_{x \in \set{X}} \biggl\{ \biggl| P_{X|Y}(x|Y) - \frac{\gamma}{M} \biggr|
+ P_{X|Y}(x|Y) - \frac{\gamma}{M} \biggr\} \Biggr] \\
\label{02072019a16}
&= \tfrac12 \, \expectation\Biggl[ \, \sum_{x \in \set{X}} \biggl| P_{X|Y}(x|Y) - \frac{\gamma}{M} \biggr| \Biggr]
+ \tfrac12 (1-\gamma),
\end{align}
where \eqref{02072019a13} is due to \eqref{E_gamma divergence}, and since $U_M(x) = \tfrac1M$ for all $x \in \set{X}$;
\eqref{02072019a14} and \eqref{02072019a15} hold, respectively, by the simple identities $(cu)^{+} = c \; u^{+}$,
and $u^{+} = \tfrac12 (|u|+u)$ for $c \geq 0$ and $u \in \Reals$; finally, \eqref{02072019a16} holds since
$$\underset{x \in \set{X}}{\sum} \Bigl( P_{X|Y}(x|y) - \frac{\gamma}{M} \Bigr) =
-\gamma + \underset{x \in \set{X}}{\sum} P_{X|Y}(x|y) = 1-\gamma,$$ for all $y \in \set{Y}$.
Substituting \eqref{02072019a13}--\eqref{02072019a16} and rearranging terms gives that
\begin{align}
\label{02072019a17}
P_{\set{L}} \geq  P_{\set{L}}^{(\mathrm{opt})}
\geq \frac{1+\gamma}{2}  - \frac{\gamma \, \expectation\bigl[ |\set{L}(Y)| \bigr]}{M}
- \tfrac12 \, \expectation\Biggl[ \, \sum_{x \in \set{X}} \biggl| P_{X|Y}(x|Y)
- \frac{\gamma}{M} \biggr| \Biggr],
\end{align}
which is the lower bound on the list decoding error probability in \eqref{LB - variable list size}.

We next proceed to prove the sufficient conditions for equality in \eqref{LB - variable list size}.
First, if for all $y \in \set{Y}$, the list decoder selects the $|\set{L}(y)|$ most probable
elements in $\set{X}$ given that $Y=y$, then equality holds in \eqref{02072019a17}. In this case,
for all $y \in \set{Y}$, $\set{L}(y) := \{x_1(y), \ldots, x_{|\set{L}(y)|}\}$ where $x_{\ell}(y)$
denotes the $\ell$-th most probable element in $\set{X}$, given $Y=y$, with ties in probabilities
which are resolved arbitrarily (see \eqref{02072019a1}). Let $\gamma \geq 1$. If, for every
$y \in \set{Y}$, $P_{X|Y}\bigl(x_\ell(y) \, | y)$ is fixed for all $\ell \in \{1, \ldots, |\set{L}(y)|\}$
and $P_{X|Y}\bigl(x_\ell(y) \, | y)$ is fixed for all $\ell \in \{|\set{L}(y)|+1, \ldots, M\}$,
then equality holds in \eqref{02072019a4} (and therefore equalities also hold in \eqref{02072019a7}
and \eqref{02072019a9}). For all $y \in \set{Y}$, let the common values of the conditional
probabilities $P_{X|Y}(\cdot | y)$ over each of these two sets, respectively, be equal to
$\alpha(y)$ and $\beta(y)$. Then,
\begin{align}
\label{02072019a21}
\alpha(y) \, |\set{L}(y)| + \beta(y) \, \bigl(M - |\set{L}(y)|)
= \sum_{x \in \set{X}} P_{X|Y}(x|y) = 1,
\end{align}
which gives the condition in \eqref{02072019a19}. Moreover, if for all $y \in \set{Y}$,
$$1-P_{\set{L}}^{\mathrm{(opt)}}(y)-\frac{\gamma \, |\set{L}(y)|}{M} \geq 0,$$ then the
operation $(\cdot)^+$ in the right side of \eqref{02072019a9} is redundant, which causes
\eqref{02072019a10} to hold with equality as an expectation of a linear function;
furthermore, also \eqref{02072019a12} holds with equality in this case (since an
expectation of a non-negative and bounded function is non-negative and finite). By
\eqref{02072019a19} and \eqref{02072019a3}, it follows that
$P_{\set{L}}^{\mathrm{(opt)}}(y) = 1 - \alpha(y) \, |\set{L}(y)|$ for all $y \in \set{Y}$,
and therefore the satisfiability of \eqref{02072019a20} implies that equalities hold in
\eqref{02072019a10} and \eqref{02072019a12}. Overall, under the above condition, it therefore
follows that \eqref{LB - variable list size} holds with equality. To verify it explicitly,
under conditions \eqref{02072019a19} and \eqref{02072019a20} which have been derived as above,
the right side of \eqref{LB - variable list size} satisfies
\begin{align}
& \frac{1+\gamma}{2} - \frac{\gamma \expectation[|\set{L}(Y)|]}{M}
- \tfrac12 \, \expectation \Biggl[ \, \sum_{x \in \set{X}} \, \biggl| P_{X|Y}(x|Y)
- \frac{\gamma}{M} \biggr| \Biggr] \nonumber \\
&= \frac{1+\gamma}{2} - \frac{\gamma \expectation[|\set{L}(Y)|]}{M} \nonumber \\
\label{02072019a22}
& \hspace*{0.4cm} - \tfrac12 \, \expectation \Biggl[ \biggl(\alpha(Y) - \frac{\gamma}{M}\biggr) \, |\set{L}(Y)|
+ \biggl(\frac{\gamma}{M} - \frac{1-\alpha(Y) \, |\set{L}(Y)|}{M - |\set{L}(Y)|} \biggr) \bigl(M - |\set{L}(Y)|\bigr) \Biggr] \\[0.1cm]
\label{02072019a23}
&= 1 - \expectation \bigl[ \alpha(Y) \, |\set{L}(Y)| \bigr] \\[0.1cm]
\label{02072019a24}
&= \expectation \Biggl[1 - \sum_{\ell=1}^{|\set{L}(Y)|} P_{X|Y} \bigl(x_\ell(Y) \, | Y \bigr) \Biggr] \\
\label{02072019a25}
&= P_{\set{L}},
\end{align}
where \eqref{02072019a22} holds since, under \eqref{02072019a20}, it follows that
$$0 \leq \frac{1-\alpha(Y) \, |\set{L}(Y)|}{M - |\set{L}(Y)|} \leq \frac1M \leq \frac{\gamma}{M}$$
for all $\gamma \geq 1$; \eqref{02072019a23} holds by straightforward algebra, where $\gamma$ is canceled out;
\eqref{02072019a24} holds by the condition in \eqref{02072019a19}; finally, \eqref{02072019a25} holds
by \eqref{P err}, \eqref{cond. P err} and \eqref{02072019a1}.
This indeed explicitly verifies that the conditions in Theorem~\ref{theorem: LB - variable list size}
yield an equality in \eqref{LB - variable list size}.

\section{Proofs of Theorems Related to Tunstall Trees}
\label{appendix: Tunstall}

\subsection{Proof of Theorem~\ref{theorem: closeness to equiprobable}}
\label{subsection: proof - clossness of equiprobable}
Theorem~\ref{theorem: closeness to equiprobable}~\ref{thm: Tunstall-a})
follows from \eqref{Tunstall-a1} (see \cite[Corollary~1]{CicaleseGV06}).

By \cite[Lemma~6]{JelSc72}, the ratio of the maximal to minimal
positive masses of $P_\ell$ is upper bounded by the reciprocal
of the minimal probability mass of the source symbols.
Theorem~\ref{theorem: closeness to equiprobable}~\ref{thm: Tunstall-b})
is therefore obtained from Theorem~\ref{thm: LB/UB f-div}~\ref{Thm. 5-c}).
Theorem~\ref{theorem: closeness to equiprobable}~\ref{thm: Tunstall-c})
consequently holds due to Theorem~\ref{thm: LB/UB f-div}~\ref{Thm. 5-d});
the bound in the right side of \eqref{23062019a5}, which holds for every
number of leaves $n$ in the Tunstall tree, is equal to the limit
of the upper bound in the right side of \eqref{23062019a3} when
we let $n \to \infty$.

Theorem~\ref{theorem: closeness to equiprobable}~\ref{thm: Tunstall-d})
relies on \cite[Theorem~11]{LieseV_IT2006}
and the definition in \eqref{closeness to U_n}, providing an integral
representation of an $f$-divergence in \eqref{int_rep.} under the
conditions in Item~\ref{thm: Tunstall-d}).

\subsection{Proof of Theorem~\ref{theorem: p_min Tunstall}}
\label{subsection: Proof - LB p_min Tunstall}
In view of \cite[Theorem~4]{CicaleseGV18}, if the fixed length of
the codewords of the Tunstall code is equal to $m$, then the compression
rate $R$ of the code satisfies
\begin{align}
\label{UB R Tunstall}
R \leq \frac{\lceil \log_{|\set{X}|}n \rceil \, H(P)}{\log_{|\set{X}|} n
- \Bigl[\frac{\rho \log \rho}{\rho-1} - \log \Bigl(\frac{\mathrm{e}
\rho \log_{\mathrm{e}} \rho}{\rho-1} \Bigr) \Bigr] \frac1{\log|\set{X}|}},
\end{align}
where $H(P)$ denotes the Shannon entropy of the memoryless and stationary
discrete source, $\rho := \frac1{p_{\min}}$, $n$ is the number of leaves
in Tunstall tree, and the logarithms with an unspecified base can be taken
on an arbitrary base in the right side of \eqref{UB R Tunstall}.
By the setting in Theorem~\ref{theorem: p_min Tunstall}, the construction
of the Tunstall tree satisfies $n \leq |\set{X}|^m < n+(D-1)$. Hence,
if $D=2$, then $\log_{|\set{X}|} n = m$; if $D>2$, then
$\lceil \log_{|\set{X}|} n \rceil = m$ (since the length of the codewords
is $m$), and $\log_{|\set{X}|} n > m + \log_{|\set{X}|} \Bigl(1-\frac{D-1}{|\set{X}|^m} \Bigr)$.
Combining this with \eqref{UB R Tunstall} yields
\begin{align}
\label{UB2 R Tunstall}
R \leq
\begin{dcases}
\frac{m H(P)}{m + \biggl\{ \log \Bigl(1-\frac{D-1}{|\set{X}|^m} \Bigr)
- \Bigl[\frac{\rho \log \rho}{\rho-1} - \log \Bigl(\frac{\mathrm{e}
\rho \log_{\mathrm{e}} \rho}{\rho-1} \Bigr) \Bigr] \biggr\}
\frac1{\log |\set{X}|}}, \quad \mbox{if $D>2$,} \\[0.1cm]
\frac{m H(P)}{m - \Bigl[\frac{\rho \log \rho}{\rho-1} - \log \Bigl(\frac{\mathrm{e}
\rho \log_{\mathrm{e}} \rho}{\rho-1} \Bigr) \Bigr] \frac1{\log |\set{X}|}},
\hspace*{3.8cm} \mbox{if $D=2$.}
\end{dcases}
\end{align}
In order to assert that $R \leq (1+\varepsilon) \, H(P)$, it is
requested that the right side of \eqref{UB2 R Tunstall} does not
exceed $(1+\varepsilon) \, H(P)$. This gives
\begin{align}
\label{12072019a1}
\frac{\rho \log \rho}{\rho-1} - \log \biggl(\frac{\mathrm{e} \rho
\log_{\mathrm{e}} \rho}{\rho-1} \biggr) \leq d \log \mathrm{e},
\end{align}
where $d$ is given in \eqref{13072019a1}. In view of the part in
Section~\ref{subsubsection: Alpha divergences} with respect to the
exemplification of Theorem~\ref{thm: LB/UB f-div} for the relative
entropy, and the related analysis in Appendix~\ref{appendix: Lambert-W},
the condition in \eqref{12072019a1} is equivalent to $\rho \leq
\rho_{\max}^{(1)}(d)$ where $\rho_{\max}^{(1)}(d)$ is defined in
\eqref{280519c}. Since $p_{\min} = \frac1{\rho}$, it leads to the
sufficient condition in \eqref{12072019a2} for the requested
compression rate $R$ of the Tunstall code.

\break

\end{document}